\documentclass[12pt,reqno,twoside]{article}

\usepackage{geometry,verbatim}
\usepackage[symbols,nogroupskip,nonumberlist,sort=use]{glossaries-extra} 

\usepackage{anyfontsize}

\usepackage{mathtools,amsmath,amssymb,enumitem,amsthm,nccmath}
\usepackage{graphicx,url}

\usepackage{color,cite,verbatim}

\usepackage[svgnames]{xcolor}

\usepackage[colorlinks=true, linkcolor=Maroon, urlcolor=Maroon]{hyperref}

\hypersetup{citecolor=blue}

\mathtoolsset{showonlyrefs} 

\usepackage{dsfont}
\usepackage{bm}
\usepackage{mathrsfs} 
\usepackage{suetterl} 
\usepackage{caption}
\usepackage{subcaption}		
\usepackage{braket} 
\usepackage{stmaryrd} 

\usepackage[T1]{fontenc}
\usepackage{currvita}

\usepackage{csquotes} 

\usepackage{float}
\usepackage[all]{xy}
\usepackage{tikz}
\usepackage{pgfplots}
\usetikzlibrary{arrows,arrows.meta,backgrounds,positioning,fit}
\usepgfplotslibrary{fillbetween}
\usetikzlibrary{patterns}
\pgfplotsset{width=7cm,compat=newest}

\tikzset{%
dots/.style args={#1per #2}{%
	line cap=round,
	dash pattern=on 0 off #2/#1
}
}

\allowdisplaybreaks[1]		

\newcommand{\ml}[1]{{\color{red} #1}}

\def\N{{\mathbb N}}	
\def\R{{\mathbb R}}				 
\def\Z{{\mathbb Z}}				
\def\T{{\mathbb T}}

\def\P{{\mathbb P}}
\def\E{{\mathbb E}}

\def\bu{{\bf u}}
\def\bx{{\bf x}}

\def\bgam{{\bm\gamma}}
\def\balph{{\bm\alpha}}

\def\cA{{\mathcal A}}

\def\cE{{\mathcal E}}
\def\cF{{\mathcal F}}

\def\cJ{{\mathcal J}}					 

\def\cM{{\mathcal M}}
\def\cN{{\mathcal N}}

\def\cR{{\mathcal R}}

\def\dx#1{\mathrm{d}#1\ }

\def\1{{\bf 1}}

\def\jb#1{\langle#1\rangle} 
\newcommand{\floor}[1]{\left\lfloor #1 \right\rfloor}
\newcommand{\ceil}[1]{\left\lceil #1 \right\rceil}

\def\dist{\operatorname{dist}}

\newcommand{\symdiff}{\triangle} 

\makeatletter
\def\namedlabel#1#2{\begingroup					
#2%
\def\@currentlabel{#2}%
\phantomsection\label{#1}\endgroup
}
\makeatother

\theoremstyle{plain}
\newtheorem{theorem}{Theorem}[section]		
\newtheorem{definition}[theorem]{Definition}	
\newtheorem{assumption}{Assumption}
\newtheorem{proposition}[theorem]{Proposition}		   
\newtheorem{example}[theorem]{Example}

\newtheorem{lemma}[theorem]{Lemma}

\newtheorem{remark}[theorem]{Remark}
\newtheorem*{remark*}{Remark}

\numberwithin{equation}{section}

\usepackage{scalerel}[2014/03/10]
\usepackage[usestackEOL]{stackengine}
\def\avint{\,\ThisStyle{\ensurestackMath{%
		\stackinset{c}{.2\LMpt}{c}{.5\LMpt}{\SavedStyle-}{\SavedStyle\phantom{\int}}}%
	\setbox0=\hbox{$\SavedStyle\int\,$}\kern-\wd0}\int}

	\def\vep{{\varepsilon}}

	\def\ergav{\cA}

	\def\errMB{\textnormal{err}^{\textnormal{(MB)}}}

	\def\dioLB{\kappa^{\textnormal{(pair)}}}
	\def\pairUB{M^{\textnormal{(pair)}}}
	\def\subldist{d_{\cA_1,\cA_2}}

	\def\inter{e^{\mathrm{(inter)}}}
	
	\def\mono{e^{\mathrm{(mono)}}}
	\def\Vmono{V^{\mathrm{(mono)}}}
	\def\Vinter{V^{\mathrm{(inter)}}}
	\def\sitepotmono{\Phi^{\mathrm{(mono)}}}
	\def\etot{e_{\mathrm{tot}}}
	\def\misfit{\Phi^{\mathrm{(misfit)}}}
	\def\GSFE{\Phi^{\mathrm{(GSFE)}}}

	\def\sitepot{\Phi^{\mathrm{(inter)}}}
	\def\sitepotdouble{\Psi^{\mathrm{(inter)}}}

	\def\atd{a_0}        
	\def\moirelen{\rho_\cM}
	\def\MoireCell{\Gamma_\cM}
	\def\MoireTorus{\MoireCell^{\mathrm{(per)}}}
	\def\OtherTorus{\Gamma_{3-j}^{\mathrm{(per)}}}
	
	\def\mavint{\avint_{\gls{GamM}}}
	\def\mavintrescj{\avint_{A_j^{-1}\gls{GamM}}}	
	\def\mavintresc#1{\avint_{A_{#1}^{-1}\gls{GamM}}}

	\def\MoireML{\cR^*_\cM}
	\def\MoireRL{\cR_\cM}
	\def\MoireRLV{A_\cM}
	\def\MoireMLV{B_\cM}
	\def\MoireG{G_\cM}

	\def\appdel#1#2{\mathds{1}_{\cR_{#1}^*}^{(#2)}}
	\def\appdelN#1{\appdel{#1}{N}}

	\def\confiso#1{\Pi_{\Gamma_{#1}}}

	\def\bot{D_{1\to 2}}            
	\def\bto{D_{2\to 1}}
	\def\disregj{D_{j\to3-j}}
	\def\disregjrev{D_{3-j\to j}}
	\def\sympl{\cJ}
	\def\mfrac#1{\left\{#1\right\}_\cM}
	\def\Lfrac#1#2{\left\{#1\right\}_{#2}}
	\def\mfloor#1{\floor{#1}_\cM}
	\def\Lfloor#1#2{\floor{#1}_{#2}}
	\def\tcor{v_{\Gamma_2}}
	
	\def\jcor{v_{\gls{Gamj}}}

	\glsxtrnewsymbol[description={Lattice of layer $j$ \eqref{def-cRj}}]{cRj}{\ensuremath{\cR_j}}
	\glsxtrnewsymbol[description={Vector $\bgam=(\gamma_1,\gamma_2)$ consisting of the origins $\gamma_j$ of layer $j$, $j=1,2$ \eqref{def-bu-bgam-balph}}]{bgam}{\ensuremath{\bgam}}
    \glsxtrnewsymbol[description={Vector $\balph=(\alpha_1,\alpha_2)$ consisting of the sublattice degrees $\alpha_j\in\cA_j$ of layer $j$, $j=1,2$ \eqref{def-bu-bgam-balph}}]{balph}{\ensuremath{\balph}}
	\glsxtrnewsymbol[description={Vector $\bu=(u_1,u_2)$ consisting of the (moir\'e-periodic) displacement functions $u_j$ of layer $j$, $j=1,2$ \eqref{def-bu-bgam-balph}}]{bu}{\ensuremath{\bu}}
	\glsxtrnewsymbol[description={Square truncated lattice in layer $j$, of side-length $2N+1$ \eqref{def-Rj-trunc}}]{cRjN}{\ensuremath{\cR_j(N)}}
	\glsxtrnewsymbol[description={Unit cell of layer $j$ \eqref{def-Gammaj}}]{Gamj}{\ensuremath{\Gamma_j}}
	\glsxtrnewsymbol[description={Reciprocal lattice of layer $j$ \eqref{def-rec-lat}}]{cRj*}{\ensuremath{\cR_j^*}}
	\glsxtrnewsymbol[description={Matrix of primitive translation vectors in layer $j$ \eqref{def-Aj}}]{Aj}{\ensuremath{A_j}}
	\glsxtrnewsymbol[description={Matrix of reciprocal primitive translation vectors in layer $j$ \eqref{def-Bj}}]{Bj}{\ensuremath{B_j}}
	\glsxtrnewsymbol[description={Disregistry with respect to layer $j$ \eqref{def-Lfrac-Lfloor}}]{Lfracxj}{\ensuremath{\Lfrac{x}{j}}}
	\glsxtrnewsymbol[description={Projection onto $\cR_j$ \eqref{def-Lfrac-Lfloor}}]{Lfloorxj}{\ensuremath{\Lfloor{x}{j}}}
	\glsxtrnewsymbol[description={Shift vector in layer $j$ corresponding to the subllatice degree $\alpha_j\in\cA_j$ \eqref{eq-unrelaxed}}]{taualphj}{\ensuremath{\tau^{(\alpha_j)}_j}}
	\glsxtrnewsymbol[description={Moir\'e lattice \eqref{def-MoireRL}}]{cRm}{\ensuremath{\MoireRL}}
	\glsxtrnewsymbol[description={Reciprocal moir\'e lattice \eqref{def-MoireML}}]{cRm*}{\ensuremath{\MoireML}}
	\glsxtrnewsymbol[description={Matrix of moir\'e primitive translation vectors \eqref{def-MoireRLV}}]{Am}{\ensuremath{\MoireRLV}}
	\glsxtrnewsymbol[description={Matrix of moir\'e reciprocal primitive translation vectors \eqref{def-MoireMLV}}]{Bm}{\ensuremath{\MoireMLV}}
	\glsxtrnewsymbol[description={Sublattice degrees associated with layer $j$ \eqref{eq-unrelaxed}}]{cAj}{\ensuremath{\cA_j}}
	\glsxtrnewsymbol[description={Moir\'e unit cell \eqref{def-moirecell}}]{GamM}{\ensuremath{\MoireCell}}
	\glsxtrnewsymbol[description={Moir\'e disregistry \eqref{def-mfrac}}]{mfracx}{\ensuremath{\mfrac{x}}}
	\glsxtrnewsymbol[description={Projection onto moir\'e lattice \eqref{def-mfloor}}]{mfloorx}{\ensuremath{\mfloor{x}}}
	\glsxtrnewsymbol[description={Total (generalized) lattice position in layer $j$, associated with position $x$ and sublattice degree $\alpha_j\in\cA_j$ \eqref{eq-lat-pos}, \eqref{eq-gen-lat-pos}}]{Yjx}{\ensuremath{Y_j(x,\alpha_j)}}
	\glsxtrnewsymbol[description={Monolayer manybody potential in layer $j$, acting on an atom of sublattice degree $\alpha_j\in\cA_j$, Section \ref{sec-energies}}]{Vmonoj}{\ensuremath{\Vmono_{j,\alpha_j}}}
	\glsxtrnewsymbol[description={Sublattice-averaged monolayer site-potential in layer $j$  \eqref{def-sitepotmono}}]{SPmonoj}{\ensuremath{\sitepotmono_j}}
	\glsxtrnewsymbol[description={Interlayer manybody potential in layer $j$, acting on an atom of sublattice degree $\alpha_j\in\cA_j$, Section \ref{sec-energies}}]{Vinterj}{\ensuremath{\Vinter_{j,\alpha_j}}}
	\glsxtrnewsymbol[description={Sublattice-averaged interlayer site-potential in layer $j$ dependent on opposite layer and moir\'e disregistry \eqref{def-sitepotdouble}}]{SPinterj}{\ensuremath{\sitepotdouble_j}}
	\glsxtrnewsymbol[description={Diagonal of $\gls{SPinterj}$ \eqref{def-sitepot}}]{SPinterjdiag}{\ensuremath{\sitepot_{j,\gls{bgam}}}}
	\glsxtrnewsymbol[description={Interlayer pair potential between atoms of types $\alpha_1\in\cA_1$ in layer 1 and $\alpha_2\in\cA_2$ in layer 2, $\gls{balph}=(\alpha_1,\alpha_2)$ \eqref{def-pair-potential}}]{pairpot}{\ensuremath{v_{\gls{balph}}}}
	\glsxtrnewsymbol[description={Rescaled Dirichlet kernel associated with the truncated lattice $\cR_j(N)$ \eqref{def-appdelNi}}]{Dirichlet}{\ensuremath{\appdelN{j}}}
    \glsxtrnewsymbol[description={Disregistry matrix of layer $j$ with respect to layer $3-j$ \eqref{def-disregj}}]{disregj}{\ensuremath{\disregj}}
   	\glsxtrnewsymbol[description={Fourier transform over the domain $\MoireCell$, see \eqref{def-FT}; similarly defines $\cF_{\gls{GamM}\times\gls{Gamj}}$ as Fourier transform over the domain $\gls{GamM}\times\gls{Gamj}$, see \eqref{def-double-FT}}]{FTMoire}{\ensuremath{\cF_{\gls{GamM}}}}
   	\glsxtrnewsymbol[description={Moir\'e length scale ratio, see \eqref{def-moirelen}}]{moirelen}{\ensuremath{\moirelen}}
	
	\makenoidxglossaries

	\usepackage{color, colortbl}
	
	\definecolor{UTorange}{RGB}{180,100,0} 
	\definecolor{UTgrey}{RGB}{51, 63, 72}
	
	\definecolor{UMNmaroon}{RGB}{122,0,25} 
	\definecolor{UMNgold}{RGB}{255,204,51}
	
	\definecolor{primary}{RGB}{122,0,25}
	\definecolor{secondary}{RGB}{255,204,51}
	
	\definecolor{CBred}{RGB}{218, 85, 38}
	\definecolor{CBorange}{RGB}{246, 137, 61}
	\definecolor{CByellow}{RGB}{254, 188, 56}
	\definecolor{CBbeige}{RGB}{216, 198, 180}
	\definecolor{CBgray}{RGB}{105, 127, 144}
	\definecolor{CBgreen}{HTML}{5BA300}
	\definecolor{CBblue}{HTML}{0073E6}
	
	\makeatother
	
\makeatletter
\def\@tocline#1#2#3#4#5#6#7{\relax
\ifnum #1>\c@tocdepth 
\else
\par \addpenalty\@secpenalty\addvspace{#2}%
\begingroup \hyphenpenalty\@M
\@ifempty{#4}{%
	\@tempdima\csname r@tocindent\number#1\endcsname\relax
}{%
	\@tempdima#4\relax
}%
\parindent\z@ \leftskip#3\relax \advance\leftskip\@tempdima\relax
\rightskip\@pnumwidth plus4em \parfillskip-\@pnumwidth
#5\leavevmode\hskip-\@tempdima
\ifcase #1
\or\or \hskip 2em \or \hskip 3em \else \hskip 4em \fi%
#6\nobreak\relax
\hfill\hbox to\@pnumwidth{\@tocpagenum{#7}}\par
\nobreak
\endgroup
\fi}
\makeatother

\usepackage{fancyhdr}

\fancypagestyle{headings}{%
\cfoot{\thepage}
}
\usepackage{orcidlink}

\usepackage[affil-it]{authblk}

\begin{document}

\author{Michael Hott\orcidlink{0000-0003-4243-6585}%
	\thanks{E-mail: \texttt{mhott@umn.edu}} \and Alexander B. Watson\orcidlink{0000-0002-7566-4851}%
	\thanks{E-mail: \texttt{abwatson@umn.edu}} \and Mitchell Luskin\orcidlink{0000-0003-1981-199X}%
	\thanks{E-mail: \texttt{luskin@umn.edu}}}
\affil{School of Mathematics,\\ University of Minneapolis, Twin Cities,\\ Minneapolis, MN 55455, USA}

\date{\today}
\title{From incommensurate bilayer heterostructures to Allen-Cahn: An exact thermodynamic limit}

\newgeometry{margin=1in}

\maketitle

\begin{abstract}
	We give a complete and rigorous derivation of the mechanical energy for twisted 2D bilayer heterostructures without any approximation beyond the existence of an empirical many-body site energy. Our results apply to both the continuous and discontinuous continuum limit.  Approximating the intralayer Cauchy-Born energy by linear elasticity theory and assuming an interlayer coupling via pair potentials, our model reduces to a modified Allen-Cahn functional. We rigorously control the error, and, in the case of sufficiently smooth lattice displacements, provide a rate of convergence for twist angles satisfying a Diophantine condition. 
\end{abstract}

\paragraph{Statements and Declarations} The authors do not declare financial or non-financial interests that are directly or indirectly related to the work submitted for publication.

\paragraph{Data availablity} The manuscript has no associated data.

\paragraph{Acknowledgements} MH's and ML's research was partially supported by Targeted Grant Award No. 896630 through the Simons Foundation. AW's, and ML's research was supported in part by DMREF Award No. 1922165 through the National Science Foundation. The authors are grateful to Eric Canc\`es and Xiaowen Zhu for helpful comments on the first version of the manuscript. The authors are grateful to Tianyu Kong for providing figures \ref{fig:moire-per-disreg} and \ref{fig:graphene}. The authors would like to thank Dmitriy Bilyk, Max Chaudkhari, Max Engelstein, Fabian Faustlich, and Ziyan (Zoe) Zhu for helpful comments, stimulating discussions, and for pointing out helpful references.

\tableofcontents

\restoregeometry 



\section{Introduction}
Setting a two-dimensional lattice on another with a small rotation creates quasi-periodic moir\'e patterns on a superlattice scale given by the period of oscillation of the interlayer disregistry, see figure \ref{fig:moire-per-disreg}. This moir\'e effect can be understood as analoguous to the \emph{beating phenomenon} in acoustics. At small twist angles, the moir\'e scale acts as a slow variable, whereas the individual layer constants act as fast variables. Consequently, moir\'e heterostructures are a suitable playground for multiscale analysis.
\begin{figure}[ht!p]
    \centering
    \includegraphics[height=0.86\textheight]{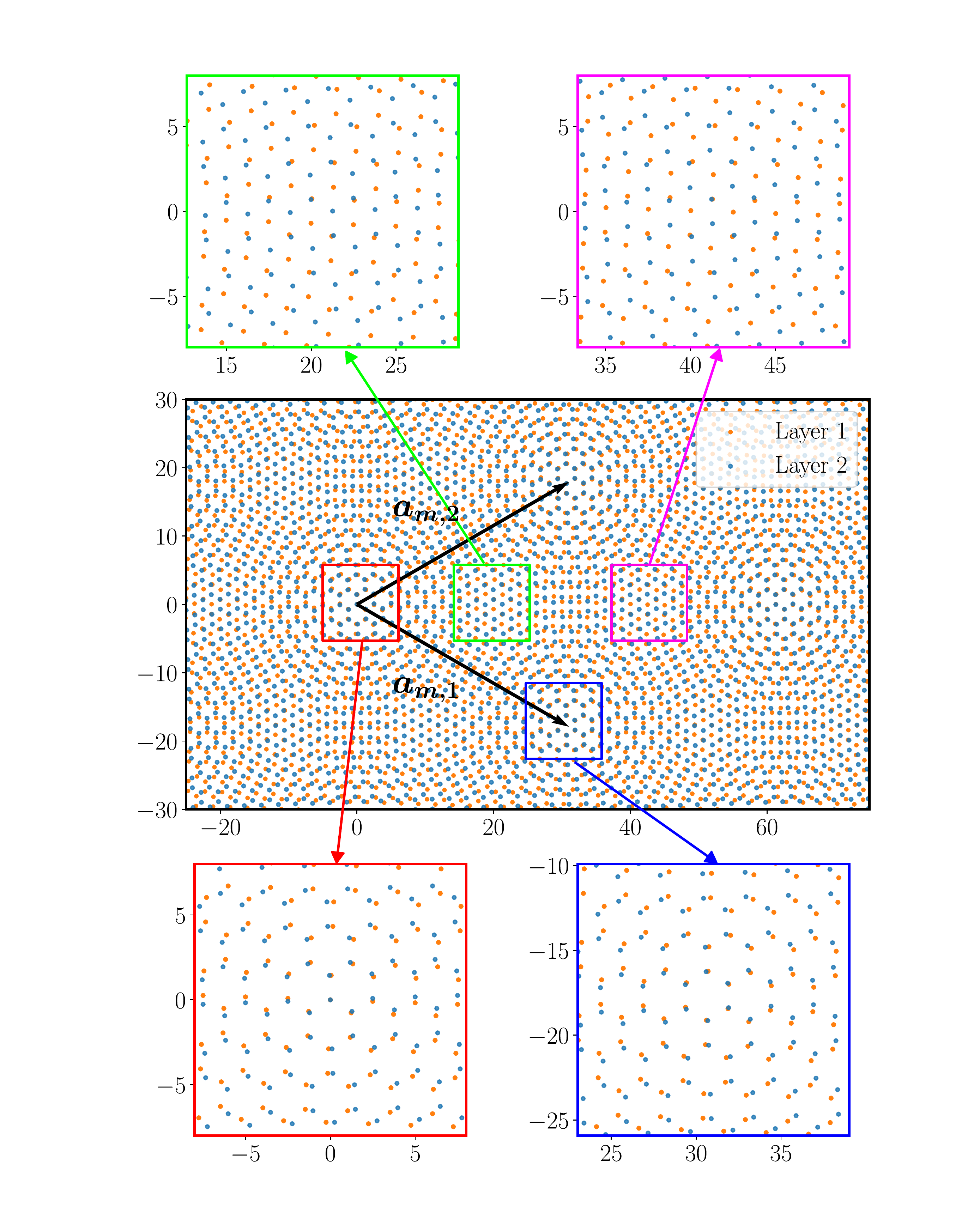}
    \caption{Twisted bilayer, with moir\'e vectors; pictures on top show local environment in symmetry-related AB and BA (Bernal) stacked regions, pictures below zoom in on two AA stacked regions separated by a moir\'e superlattice vector. See section \ref{sec-AC} for a description and discussion of the energetics of AA, AB, and BA stacking.}
    \label{fig:moire-per-disreg}
\end{figure}
In recent years, moir\'e materials consisting of a few relatively twisted layers of 2D-crystal structures have attracted great attention since the moir\'e length scale offers the possibility of discovering previously inaccessible electronic phases and other phenomena ~\cite{KimHof13}.  
\par The seminal paper in 2011 by Bistritzer and MacDonald \cite{bistritzer2011moire} predicted that twisted bilayer graphene (TBG), at an angle of $\theta\approx 1.1^\circ$ -- the so-called magic angle -- exhibits almost flat electronic bands near the charge neutrality point. As a consequence of the flat bands, there is a high density of states, which has been known to provide a platform for emergent electronic phases involving strong electron interactions. This has been experimentally verified, 
e.g., in \cite{Cao2018,Cao2018a}, where both a superconducting and a correlated insulating phase have been established. A class of other 2D materials is given by \emph{transition metal dichalcogenides} (TMDs). Popular choices for transition metals include Mo and W and for chalcogens include Se, S, and Te. In particular, moir\'e materials have become quantum simulators, i.e., highly tunable experimental realizations for effective theories predicted by theorists. For overviews on moir\'e materials, we refer to, e.g., \cite{catarina-amorim-castro-lopes-peres2019tBGintro,rubio2021quantum-simulator,Mutalik2019quanta-article }.

Structural relaxation, the process by which atoms deviate from their monolayer equilibrium positions in order to minimize the total mechanical energy of the twisted bilayer, can significantly modify the electronic properties of moir\'e materials. For example, relaxation enhances the band gap between the flat electronic bands and other bands in twisted bilayer graphene \cite{KimRelax18}. In a simplified model of twisted bilayer graphene with strong relaxation, the chiral model, the almost flat bands become exactly flat \cite{becker-humbert-zworski2022finestructure-magicangle,beckerzworski2022magicangle,watsonluskin2021,TKV2019magicanglecond}. It is thus critical to develop systematic methods for modeling structural relaxation and its effects on electronic properties of moir\'e materials.

Modeling structural relaxation is complicated by the fact that the atomic structure of twisted bilayer 2D materials is generally aperiodic (incommensurate) except at special commensurate angles \cite{2DPerturb15}. Phenomenological continuum models for structural relaxation at commensurate twist angles have been proposed in \cite{Srolovitz2015,Nam2017,SanJose2014sg,Tadmor2017}.

A systematic approach not restricted to commensurate twist angles or bilayer structures was proposed in \cite{Cazeaux-Massatt-Luskin-ARMA2020}.  In this approach, the displacements $u_j$ of atoms in layer $j$ of an incommensurate twisted 2D bilayer structure are assumed to be continuous functions of the local configuration, or disregistry, which can be identified with a two-dimensional torus.   
This is based on the more general concept of a {\it hull}, a compact parametrization of all possible local environments, which has been introduced in the context of electronic properties of general aperiodic structures, see, e.g., \cite{bellissard-quasi-per1982,bellissard1994noncommutative,bellissard2002coherent} and references therein.  

\ml{For more general multilayer structures with $p>2$ layers \cite{zhu-carr-massatt2020twisted-tunable,zhang-tsai-zhu2021correlatedinsulating}, 
the hull can be identified with a $2p-2$ dimensional torus and hence there does not exist a two-dimensional periodic continuum model in real space.
The configuration-space model proposed in \cite{Cazeaux-Massatt-Luskin-ARMA2020} has become increasingly used to model trilayer and more general structures \cite{zhu_relaxation_2020}.}

\par Such a parametrization is known to be valid for the Frenkel-Kontorova model in one dimension, using Aubry-Mather theory \cite{1978Aubry,1980Aubry,aubry1983frenkel-kontorova}. In addition, as a function of the disregistry, a minimizer, if it exists, is smooth as long as the intralayer potential dominates the interlayer coupling; see \cite{1978Aubry,1980Aubry,aubry1983frenkel-kontorova} and also \cite{cazeauxrippling2017} for a simplified model. When this condition does not hold, the minimizer can be discontinuous at a dense set of points \cite{1978Aubry,1980Aubry,aubry1983frenkel-kontorova,cazeauxrippling2017}. This motivated us to extend the previous framework to include non-smooth displacements. In two dimensions, we are not aware of any proof that minimizers of the twisted bilayer energy necessarily admit such a continuous parameterization, although the assumption is well-motivated given that intralayer lattice bonds are known to be much stronger than interlayer couplings in TBG and other van der Waals 2D heterostructures.

\par The total mechanical energy of the bilayer is, then, assumed to be decomposable into a monolayer and an interlayer coupling contribution. For the monolayer part, \cite{Cazeaux-Massatt-Luskin-ARMA2020} starts with many-body site potentials and then passes directly to the Cauchy-Born approximation, inspired by, e.g., \cite{ortnertheil2013,blanc2002molecular}. For the interlayer coupling, they utilize a stacking penalization, the \emph{generalized stacking fault energy} (GSFE), see also \cite{dai2016twisted}, dependent on a {\it modulated disregistry}, i.e., on the relaxed local environment. Their modulated disregistry is obtained by projecting an interpolated relaxed position onto the relaxed structure of the other layer, see figure \ref{fig-mod-disreg}. 

\begin{figure}[htp]
	\centering
	\includegraphics[width=0.9\textwidth]{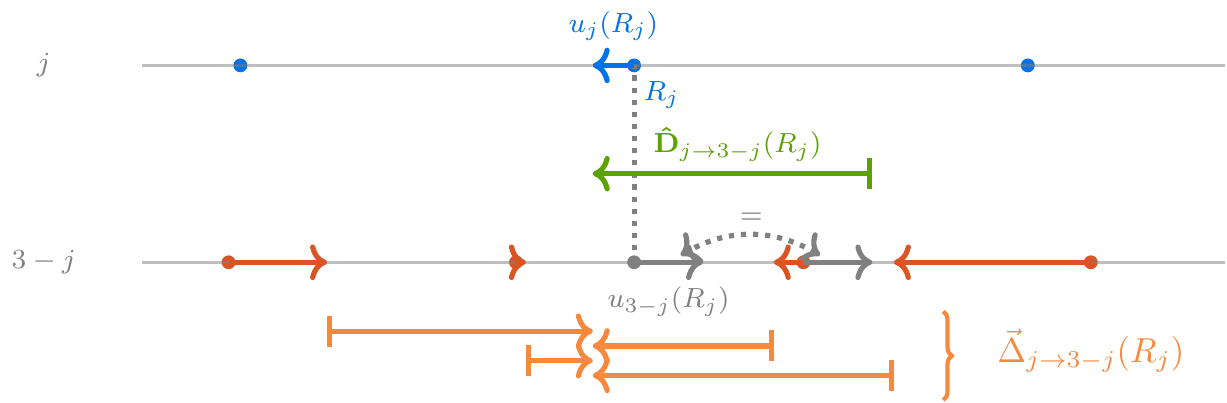}
	\caption{Modulated disregistry ${\bf \hat{D}}_{j\to3-j}(R_j)$ chosen in \cite{Cazeaux-Massatt-Luskin-ARMA2020}. In comparison, the sequence $\vec{\Delta}_{j\to3-j}(R_j)$ of all modulated disregistries chosen in the present work.}
	\label{fig-mod-disreg}
\end{figure}

To define the modulated disregistry more precisely, let \gls{Aj} denote the matrix whose columns are the primitive translation vectors of lattice $\gls{cRj}=A_j\Z^2$ describing layer $j$ (see the beginning of Section \ref{sec-model}), and
\begin{align}\label{def-disregj}
	\gls{disregj} \, := \, I - A_{3-j}A_j^{-1} 
\end{align}
denote the \emph{disregistry matrix of layer $j$ with respect to layer $3-j$}, see Section \ref{sec-model} below for more context. We choose Lagrangian coordinates for the lattice displacements $u_j$. Then the modulated disregistry chosen in \cite{Cazeaux-Massatt-Luskin-ARMA2020} is given by 
\begin{align}
	{\bf \hat{D}}_{j\to3-j}(R_j)=\gls{disregj} R_j+u_j(R_j)-u_{3-j}(R_{j}) \, , \quad R_j\in\gls{cRj} \, ,
\end{align}
see figure \ref{fig-mod-disreg}. In particular, since $u_{3-j}$ is not evaluated at a lattice vector of layer $3 - j$, the chosen modulated disregistry in \cite{Cazeaux-Massatt-Luskin-ARMA2020} only approximates the difference vectors of nearby atoms from different layers; see Section \ref{sec-conv-small-angle} for a more thorough discussion. 

\par In order to pass to a thermodynamic limit, \cite{Cazeaux-Massatt-Luskin-ARMA2020} employs an ergodic property  for continuous functions of the interlayer disregistry in incommensurate lattices, derived in \cite{generalizedkubo2017,dos17} in the context of 2D incommensurate systems -- systems with a relative irrational angle, see Assumption \ref{ass-incomm} below. After they approximate the Cauchy-Born energy by linear elasticity, their resulting model is given by
\begin{align} \label{ARMA-energy}
	\begin{split}
		&\frac12 \sum_{j=1}^2 \mavint \dx{x}\Big(  Du_j(x):\cE_j:Du_j(x) \\
		& \qquad +\GSFE_j\big(\gls{disregj} x + u_j(x)-u_{3-j}(x)\big)\Big) \, ,
	\end{split}
\end{align}
where $\gls{GamM}$ denotes a moir\'e unit cell, see \eqref{def-moirecell} below, and the GSFE\ $\GSFE_j:\R^2\to\R$ is periodic with respect to to lattice $\cR_{3-j}$ of the other layer. Furthermore, we abbreviate by
\begin{align}
	\cE_{j,abcd} \, := \, \lambda_j\delta_{ab}\delta_{cd} \, + \, \mu_j(\delta_{ad}\delta_{bc}+\delta_{ac}\delta_{bd}) \, , \label{def-elast-tensor}
\end{align}
the elasticity tensor where $\lambda_j,\mu_j>0$ are referred to as \emph{Lam\'e parameters}. We adjusted the sign-convention in \eqref{ARMA-energy} to that chosen in the present work.

\par In applications, the GSFE is obtained via a density functional theory (DFT) approach by minimizing the energy of two untwisted parallel layers at a given $xy$-shift corresponding to the disregistry, while the system is allowed to relax in the $z$-direction, see \cite{dai2016twisted,Carr2018}. Periodic boundary conditions can then be chosen to account for the lattice periodicity. The relaxation model obtained in \cite{Cazeaux-Massatt-Luskin-ARMA2020} has been very successfully applied to determine the relaxed structure of twisted bilayer systems at small twist angles, see, e.g., \cite{Carr2018}. The goal of the present work is to develop a new approach to modeling relaxation of twisted bilayer 2D materials which builds on and refines several aspects of the approach of \cite{Cazeaux-Massatt-Luskin-ARMA2020}. In particular, we expect our model to retain accuracy even at relatively large twist angles, and lend itself better to the systematic study of lattice vibrations (phonons), which are known to play a significant role in the electronic properties of many materials.

The first key point of our approach is that we start from an exact, i.e., not interpolated, discrete model for the interlayer energy. More precisely, the interlayer site potentials depend on the \emph{finite difference stencils}
\begin{align} \label{eq:seq_of_diff}
	\vec{\Delta}_{j\to3-j}(R_j) \, = \, \big(R_j+u_j(R_j)-R_{3-j}-u_{3-j}(R_{3-j})\big)_{R_{3-j}\in\cR_{3-j}}
\end{align}
of displaced lattice positions, see figure \ref{fig-mod-disreg}. This approach allows us to pass to the thermodynamic limit, without first requiring other approximations. Crucially, our approach still relies on the assumption that the lattice displacements depend on the local disregistry. Under this assumption, our model encodes all of the structure of the underlying discrete model, and allows for further approximations at later steps. Passing to the thermodynamic limit with such a penalization requires a more general form of the Birkhoff theorem used in \cite{Cazeaux-Massatt-Luskin-ARMA2020}. Since we do not invoke any interpolation or small-angle approximation here, we expect this interlayer energy to retain accuracy for arbitrary twist angles. We additionally generalize \cite{Cazeaux-Massatt-Luskin-ARMA2020} in allowing for finitely-many independent sublattice degrees of freedom, and relaxation in the vertical direction, perpendicular to the layers.

\par Although we emphasize the generality of our approach, it is illuminating to present our configuration-space model under simplifying assumptions. In particular, this allows for comparison between our model and \eqref{ARMA-energy}. Let us assume the interlayer many-body potential is given by an even pair potential $v$. We ignore sublattice degrees of freedom and lattice shifts, and restrict $u_j$ to take values in $\R^2$, i.e., we ignore out-of-plane relaxation, and assume smooth $u_j$ so that the Cauchy-Born approximation applies. After further approximation by linear elasticity, our model is given by 
\begin{align} \label{our-energy}
	\begin{split}
		&\frac12 \sum_{j=1}^2\mavint \dx{x} Du_j(x):\cE_j:Du_j(x) \\
		& \quad + \frac1{|\gls{GamM}|} \int_{\R^2} \dx{\xi} v\big(A_1 \xi + u_{1}(A_1\xi) - A_2\xi-u_2(A_2\xi)\big) \, .
	\end{split}
\end{align}
For the most general case of monolayer and interlayer coupling energy, we refer to equations \eqref{def-mono-TD}, and \eqref{def-inter-TD-pair}, respectively. We show the difference between the models for a specific choice of GSFE depending on $v$ in Section \ref{sec-conv-small-angle}. A crucial benefit of our approach is that it is formulated for the more general case of many-body potentials; another advantage is that our result holds for arbitrary angles, while the result in \cite{Cazeaux-Massatt-Luskin-ARMA2020} was derived only in the small-angle limit. We believe that our analysis is extendable to multiple layers as in \cite{Cazeaux-Massatt-Luskin-ARMA2020}. In Section \ref{sec-CB}, we provide some more details on the Cauchy-Born and linear elasticity approximation. We would like to point out that the interlayer term in \eqref{our-energy} takes into account that all distances between points of different layers are attained.

\par The next main result is, under additional regularity assumptions on the displacements $u_j$, an estimate on the convergence rate for passage to the thermodynamic limit. For this, we introduce the notion of \emph{Diophantine 2D rotations}, a two-dimensional analogue of Diophantine numbers, to quantify the incommensurability of twisted 2D bilayer lattices. We show that an ergodic average can be rewritten in Fourier representation as a weighted sum involving an appropriately rescaled Dirichlet kernel. The leading order term in orders of the truncation corresponds to the limit of the ergodic average. The lower order terms can be bounded by the tail-behavior of the Dirichlet kernel, if we assume an appropriate Diophantine angle condition, defining Diophantine 2D rotations. Our concept of Diophantine 2D rotations extends the ideas developed in \cite{cazeaux-luskin-CB-2017} in the context of coupled one-dimensional chains. Metric Diophantine approximation theory is a vibrant research area in analytic number theory. We refer to \cite{bilyk2014discrepancy,KuipersNiederreiter1974,beresnevich2016metric} and references therein for historic reviews, and, e.g., to \cite{bilyk2011directional,bilyk2016diophantine,champagne2023dio-constr,hammonds2023kdiophantinemtuples,odorney2023diophantine} and references therein for more recent results. 

The final contribution of the present work is that we prove the existence of a thermodynamic limit for the energy density under least restrictive assumptions. These assumptions ensure integrability in the limit energy density functional. The starting point of our approach is a weaker, namely $L^1$, parametrization of the atomic displacement functions $u_j$ as functions of their relative position with respect to the moir\'e unit cell $\gls{GamM}$, modulo a moir\'e lattice shift. To be more precise, we pass to the thermodynamic limit for any $(u_1,u_2)$ such that the limiting site energy is in $L^1(\Gamma_\mathcal{M})$, where $\Gamma_{\mathcal{M}}$ is the moir\'e cell; see Proposition \ref{prop-single-erg-thm}, and Theorem \ref{thm-rough-conv}. We show that $u_j \in L^\infty(\Gamma_\mathcal{M}) \subset L^1(\Gamma_\mathcal{M})$ is sufficient for this in Lemma \ref{lem-f-galpha-L1-suff-cond}.
\par This parametrization is given meaning by embedding the twisted bilayer model within the family of models defined by all possible layer shifts (modulo shifts which leave the bilayer invariant) and applying Birkhoff-type ergodic theorems which hold up to a measure zero set of shifts, see Section \ref{sec-rough-displacements} for more details. In particular, this perspective allows us to pass to an almost sure thermodynamic limit without assuming a continuous parametrization of the atomic displacements as in \cite{Cazeaux-Massatt-Luskin-ARMA2020}.

\par One of the promising properties of the models presented here and in \cite{Cazeaux-Massatt-Luskin-ARMA2020} are that they lend themselves to methods from calculus of variations for continuous systems, even though the underlying system is discrete. As we will discuss below in Section \ref{sec-AC}, the resulting model is connected to the Allen-Cahn or Ginzburg-Landau energy functional \cite{malena2018,malena2023}. For references on the study of the Allen-Cahn equation, we refer to, e.g., \cite{cabrecinti2021allencahn-stable,cinti-davila-delpino2016allencahn,figalli2020stable,restrepo-maggi2022allencahn,tachim2022allencahn-large}  and references therein, for the Ginzburg-Landau equations, e.g., to \cite{guopei2022ginzburglandau,ignatjerrard2021GinzburgLandau,savin2010GinzburgLandau-minimal,serfaty2017meanfield-GinzburgLandau-GP}  and references therein, and for results on more general semilinear elliptic equations, e.g., to \cite{arendt2023semilinear,cazenaveSLEE2006book,diaz2023semilinear,dupaigneStableSolBook,dupaigne2021regularity,joseph2023semilinear}, and references therein.  
\par An interesting question is whether an existence and regularity theory can be developed for the thermodynamic limit energy involving many-body monolayer potentials, before passing to the Cauchy-Born approximation. We do not address this question in the present work, instead developing the existence and regularity theory only for displacements for which the intralayer energy can be simplified using the Cauchy-Born approximation, such as $C^2$-displacements. After this simplification, the existence and regularity theory are standard, since the Euler-Lagrange equation is, then, elliptic; see, e.g., \cite{Cazeaux-Massatt-Luskin-ARMA2020}.
\par We comment on the practical application of our results, especially to the study of lattice vibrations (phonons) in twisted bilayer materials. First, we note that passing to the thermodynamic limit without restricting to smooth or even continuous displacements $u_j$ should allow us to model more subtle relaxation patterns and vibrations than those which can be described using the Cauchy-Born or linear elasticity approximation. The use of the GSFE in \cite{Koshino2019} is limited to nearest neighbor interactions and does not include out-of-plane displacement.  The phonon dispersion method given in \cite{Lu2022} includes longer-ranged interactions in the GSFE and out-of-plane displacements.  Our interlayer energy is exact, but it needs approximation since it includes infinite range interactions. 
A systematic approximation of our interlayer energy can give a method to compute the phonon dispersion with controlled accuracy.  Finally, we note that our new interlayer misfit energy has the potential to be highly accurate if implemented with the new generation of machine learned empirical potentials \cite{ortner21,shapeev16}. These issues will be the subjects of forthcoming works.

\section{Model description\label{sec-model}}

In order to state our model, we start by introducing standard notation to describe general Bravais lattices, see, e.g., \cite{massatt-luskin-ortner2017electronicdensity}. A 2D Bravais lattice consists of 2 \emph{primitive translation vectors}, as well as a \emph{basis} with associated translation vectors, see fig. \ref{fig:graphene}. We associate the basis with \emph{sublattice degrees of freedom}. In the case of hexagonal lattices, the sublattice degrees of freedom refer to an additional atom position given a fixed unit cell. In the case of TMDs, it allows to distinguish between various species of atoms, the metals and chalcogens. Atomic orbitals present another choice of sublattice degrees of freedom. To understand the main ideas, the reader may want to ignore the sublattice degrees of freedom on first reading.

\subsection{Notation}

\begin{figure}[tp]
    \centering
    \includegraphics[width=0.6 \textwidth]{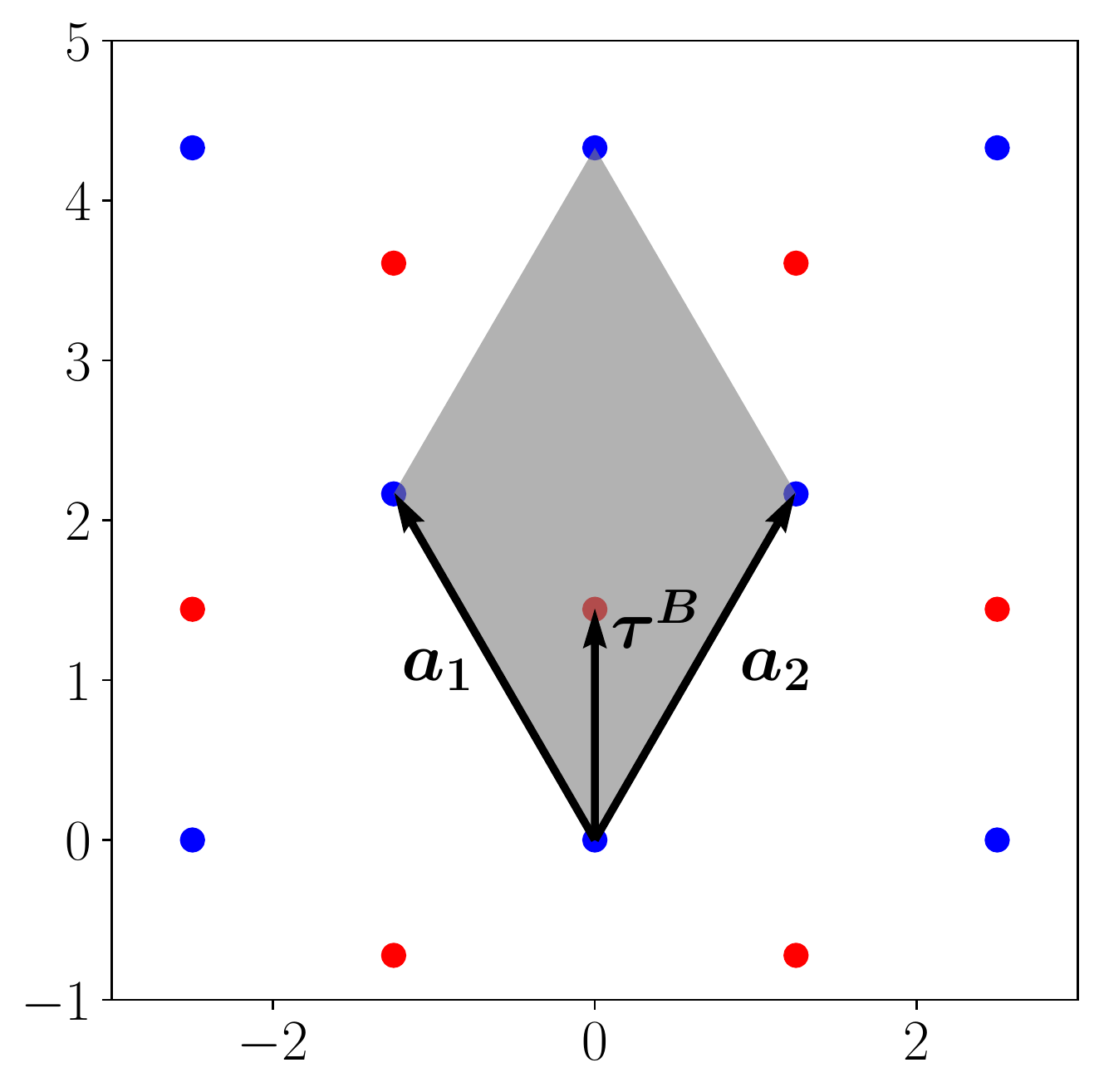}
    \caption{Hexagonal graphene lattice. $a_j$ refers to the primitive translation vectors, while $\tau^{B}$ is the translation vector for an additional position within a given unit cell. Blue dots refer to the $A$ lattice, while red dots refer to the $B$ lattice.}
    \label{fig:graphene}
\end{figure}

\subsubsection{Lattices} 

Let $\theta\in\R$, $q>0$, and $A\in\text{GL}_2(\R)$. Define
\begin{align} \label{def-Aj}
	A_1 \, := \, q^{-1/2}R_{-\theta/2} A \, , \quad A_2 \, := \, q^{1/2}R_{\theta/2} A \, , \quad R_\phi \, := \, \begin{pmatrix}
		\cos \phi & -\sin \phi\\
		\sin\phi & \cos\phi
	\end{pmatrix} \, .
\end{align}
Let
\begin{align}\label{def-cRj}
	\gls{cRj} \, := \, A_j\Z^2 \, , \, j\in\{1,2\},
\end{align}
denote two crystal layers at a relative twist angle $\theta$ and with a relative lattice mismatch $q$, and abbreviate the respective unit cells
\begin{align} \label{def-Gammaj}
	\gls{Gamj} \, :=\,A_j[0,1)^2 \, , \, j\in\{1,2\} \, .
\end{align}
In addition, denoting $\llbracket n_1,n_2\rrbracket:=[n_1,n_2]\cap\Z$ whenever $n_1,n_2\in\Z$, $n_1\leq n_2$, we introduce the truncated lattices
\begin{equation} \label{def-Rj-trunc}
	\gls{cRjN} \, := \, A_j\llbracket -N,N\rrbracket^2 \, , \, j\in\{1,2\} \, , \, N\in\N \, .
\end{equation}
Furthermore, let
\begin{align} \label{def-Bj}
	\gls{Bj} \, := \, 2\pi A_j^{-T} \, , \, j\in\{1,2\},
\end{align}
denote the matrix of reciprocal vectors, and define the reciprocal lattices
\begin{align} \label{def-rec-lat}
	\gls{cRj*} \, := \, \gls{Bj}\Z^2 \,,\, j\in\{1,2\} \, .
\end{align}
Throughout this work, we let $j\in\{1,2\}$.

In addition to the lattices $\gls{cRj}$, we also introduce the {\em sublattice degrees of freedom $\gls{cAj}$, $j\in\{1,2\}$.} $\gls{cAj}$ is assumed to be a finite set, and to each $\alpha_j\in\gls{cAj}$, we assign a unique shift vector $\gls{taualphj}\in\R^2$, see fig. \ref{fig:graphene}. Moreover, as an additional degree of freedom, we allow layer $j$ to be shifted by $\gamma_j\in\gls{Gamj}$. Consequently, the unrelaxed positions of layer $j$ are given by
\begin{equation}\label{eq-unrelaxed}
    \gamma_j \, + \, R_j \, + \, \gls{taualphj}, \,  \ \alpha_j\in\gls{cAj}.
\end{equation}

\begin{figure}
	\begin{subfigure}[c]{0.5\textwidth}
		\includegraphics[width=0.9\textwidth]{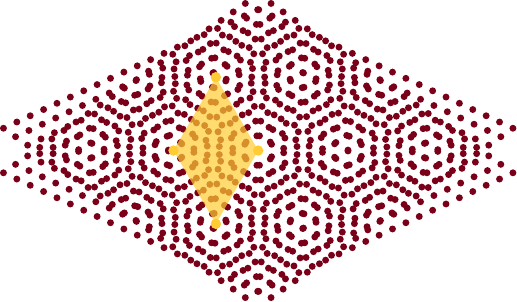}
	\end{subfigure}\quad 
	$\xrightarrow[\cong]{\confiso{\cM}}$ \quad 
	\begin{subfigure}[c]{0.3\textwidth}
		\includegraphics[width=0.3\textwidth]{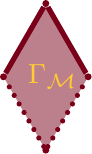}
	\end{subfigure}
	\caption{Twisted bilayer with its moir\'e lattice highlighted in yellow.}
	\label{fig-moirelattice}
\end{figure}
\par 
As explained above, the moir\'e lattice is another important lattice in incommensurate bilayer heterostructures. We will first define this lattice and then explain its importance. The primitive translation vectors for the \emph{moir\'e frequency lattice} then are given by
\begin{align} \label{def-MoireMLV}
	\gls{Bm} \, := \, B_1-B_2  \, ,
\end{align}
generating the \emph{moir\'e frequency lattice}
\begin{align} \label{def-MoireML}
	\gls{cRm*} \, := \, \gls{Bm}\Z^2 \, .
\end{align}
The corresponding moir\'e direct lattice vectors are given by
\begin{align} \label{def-MoireRLV}
	\gls{Am} \, &:= \, 2\pi \MoireMLV^{-T} \, = \, (A_1^{-1}-A_2^{-1})^{-1} \, .
\end{align}
Then we define \emph{moir\'e lattice} 
\begin{align} 
	\gls{cRm} \, := \, \gls{Am}\Z^2 \label{def-MoireRL}
\end{align}
together with its unit cell
\begin{align}
	\gls{GamM} \, := \, \gls{Am} [0,1)^2 \, .  \label{def-moirecell}
\end{align}

We now explain the importance of the moir\'e lattice. As can be seen in figure \ref{fig-moirelattice}, the moir\'e lattice is the lattice with respect to which the union of the lattices $\cR_1\cup\cR_2$ appears almost periodic. This phenomenon can be explained from two different perspectives.

\par First, the moir\'e lattice is the lattice of periodicity for the beating pattern created by superposing plane waves with the periodicity of each layer. Observe that we have that
	\begin{align}
    e^{i(B_1e_j)\cdot x} \, = \, e^{i(B_2e_j)\cdot x} \mbox{\ for\ }j=1,2 \quad \iff \quad x\in 2\pi (B_1-B_2)^{-T}\Z^2 \, ,
	\end{align}
 and hence constructive interference at the points of the moir\'e lattice.

Second, we would like to formalize the idea that the local environment is moir\'e (quasi-)periodic, as can be seen in fig. \ref{fig:moire-per-disreg}. For that, observe that, for a lattice vector $R_1\in\cR_1$ in layer 1, its position relative to the lattice of layer 2 is
\begin{equation}
    R_1 \mod \cR_2 \, . 
\end{equation}
Equivalently, we can consider the position of $R_1$ relative to its projection into layer 2
\begin{equation}
    (I-A_2A_1^{-1})R_1 \mod \cR_2 \, .
\end{equation}
This motivates the definition of the disregistry matrices $\disregj$ introduced in \eqref{def-disregj}, given by
\begin{equation}\label{eq-disregj-moire-per}
    \bot \, = \, I-A_2A_1^{-1} \, = \, -A_2\MoireRLV^{-1} \, , \quad \bto \, = \, I-A_1A_2^{-1} \, = \, A_1\MoireRLV^{-1} \, .
\end{equation}
With these, we define the \emph{local configuration functions}
\begin{equation}
    x\mapsto (I-A_2A_1^{-1})x \mod \cR_2 \, , \quad x\mapsto (I-A_1A_2^{-1})x \mod \cR_1 \, . 
\end{equation}
Due to \eqref{eq-disregj-moire-per}, both of these local configuration functions are periodic with respect to the moir\'e lattice $\MoireRL$. This establishes the moir\'e periodicity for a notion of the local environment. 

Notice that there are other descriptions of the local environment which are not moir\'e-periodic, such as the \emph{disregistries} defined below in \eqref{def-mfrac} and \eqref{def-Lfrac-Lfloor}, although these notions are related to moir\'e-periodic quantities through 
\eqref{eq:layer_and_moire_disreg}.

\begin{lemma}[Moir\'e length scale]\label{lem-moire-len}
	For any $x\in\R^2$, we have that
	\begin{align}
		|\gls{Am} x| \, = \, [(q^{1/2}-q^{-1/2})^2+4\sin^2(\theta/2)]^{-\frac12}|Ax| \, .
	\end{align}
	Similarly, we have for any $k\in\R^2$ that
	\begin{align}
		|\gls{Bm} k| \, = \, [(q^{1/2}-q^{-1/2})^2+4\sin^2(\theta/2)]^{\frac12}|2\pi A^{-T}k| \, .
	\end{align}
\end{lemma}
For a proof, we refer to Appendix \ref{app-latt-calc}. Motivated by this lemma, we define the characteristic \emph{moir\'e length scale ratio} given by
\begin{equation}\label{def-moirelen}
	\gls{moirelen} \, := \, [(q-q^{-1})^2+4\sin^2(\theta/2)]^{-\frac12} \, .
\end{equation}

In order to specify almost periodicity, we impose the following restrictions.

\begin{assumption}\label{ass-incomm}
	$(\cR_1,\cR_2)$ is \emph{incommensurate}, i.e., 
	\begin{align} \label{eq-incomm}
		G_1 \, + \, G_2 \, = \, 0 \mbox{\ \emph{for}\ } G_j\in \gls{cRj*} \quad \mbox{\emph{iff}} \quad G_1=G_2=0 \, .
	\end{align}
\end{assumption}

Assumption \ref{ass-incomm} has been introduced in \cite{generalizedkubo2017,dos17}. In addition, we impose the following condition.

\begin{assumption}\label{ass-reci-incomm}
	$(\cR_1^*,\cR_2^*)$ is incommensurate.
\end{assumption}

\begin{remark}\label{rem-Rj-Rm-incomm}
	Under assumption \ref{ass-incomm}, it can be readily verified that also $\gls{cRj}$ and $\cR_\cM$ are incommensurate.  In fact, let $\MoireG=\gls{Bm} m\in \gls{cRm*}$, $m\in\Z^2$, and $G_1=B_1n\in\cR_1^*$, $n\in\Z^2$. Then we have that 
	\begin{equation}
		\MoireG \, + \, G_1 \, = \, B_1(m+n) \, - \, B_2n \, = \, 0
	\end{equation}
	iff $m=n=0$, see \eqref{def-MoireMLV}. Analogously, $(\cR_1^*,\cR_2^*)$ being incommensurate implies that $(\gls{cRj*},\cR_\cM^*)$, $j=1,2$ is incommensurate.
\end{remark}

A straight-forward calculation yields the following result.

\begin{lemma}\label{lem-incomm-angles}
	We have that $(\cR_1,\cR_2)$ is incommensurate iff $(q,\theta)\in\R^+\times[0,2\pi)$ satisfy
	\begin{equation}
		qA^TR_\theta A^{-T}\Z^2\cap\Z^2 \, = \, \{0\} \, .
	\end{equation}
\end{lemma}

As a consequence of Assumption \ref{ass-reci-incomm}, the relative positions of $\gls{cRj}$ with respect to the moir\'e lattice never repeat. We want to use this assumption to define a continuous extension of the displacement functions,  defined on the (discrete) lattice $\gamma_j+\cR_j$, and corresponding energies. In particular, we use assumption \ref{ass-reci-incomm} to relabel lattice positions by their relative position within a translate of a moir\'e unit cell. To this end, we define for $x\in\R^2$
\begin{equation} \label{def-mfrac}
	\gls{mfracx} \, := \, \sum_{R_\cM\in\gls{cRm}} (x-R_\cM)\mathds{1}_{\gls{GamM}}(x-R_\cM) \, (= \, x \mod \MoireRL) \, ,
\end{equation}
and
\begin{equation}\label{def-mfloor}
	\gls{mfloorx} \, := \, x \, - \, \gls{mfracx} \, .
\end{equation}
Observe that
\begin{equation}
	\mathds{1}_{\gls{GamM}}(x-R_\cM)=0 \quad \Leftrightarrow  \quad x\in \gls{GamM}+R_\cM \, .
\end{equation}
Analogously, we define the decomposition with respect to the lattices $\gls{cRj}$
\begin{equation}\label{def-Lfrac-Lfloor}
	x \, = \, \gls{Lfloorxj} \, + \, \gls{Lfracxj}
\end{equation}
with $\gls{Lfloorxj}\in\gls{cRj}$, $\gls{Lfracxj}\in \gls{Gamj}$. An analogous notation has already been used in \cite{generalizedkubo2017}. We call $\mfrac{\cdot}$ the \emph{(local) disregistry} with respect to $\gls{cRm}$ and $\Lfrac{\cdot}{j}$ the \emph{(local) disregistry} with respect to $\gls{cRj}$.

\par The quasi-moir\'e-periodicity of the local disregistry can also be understood by virtue of Lemma \ref{lem-disreg-transf}, for we have that
	\begin{align} \label{eq:layer_and_moire_disreg}
		\Lfrac{R_1}{2} \, 
		= \, \bot \mfrac{R_1} \, + \, A_2\begin{pmatrix}
			1\\1
		\end{pmatrix} \, , \quad
		\Lfrac{R_2}{1} \, 
		= \, \bto \mfrac{R_2} \, .
	\end{align}
	Exact periodicity would mean that similar formulas relating the moir\'e disregistry and the individual layer disregistry would hold for general $x\in\R^2$ instead of lattice points $R_j$. However, in Remark \ref{rem-gen-disreg-trafo-nonexist}, we prove that there exists no such affine transformation. In the incommensurate case, there is only quasi-periodicity. However, in case of \emph{commensurate} layers, the local disregistry is periodic with respect to a lattice which is a subset of the moir\'e lattice $\gls{cRm}$. Note that in the commensurate case, methods for periodic systems, such as Bloch-Floquet theory, become available. For this reason, and since the commensurate case is in any case not generic, see Proposition \ref{prop-dio-existence}, we ignore this case in the present work.

\par As a consequence of incommensurability of $(\cR_1^*,\cR_2^*)$ and $(\gls{cRj*},\cR_\cM^*)$, respectively, both 
\begin{align}
	\confiso{3-j} \, &: \, \gls{cRj}+\gamma_j \to \Gamma_{3-j} \, , \quad R_j+\gamma_j\mapsto \Lfrac{R_j+\gamma_j}{3-j} \, , \label{def-confiso}\\
	\confiso{\cM} \, &: \, \gls{cRj}+\gamma_j \to \gls{GamM} \, , \quad R_j+\gamma_j\mapsto \mfrac{R_j+\gamma_j} \, , \label{def-confisom}
\end{align}
are one-to-one, $j=1,2$. For a discussion of this bijection, we refer, e.g., to \cite{massatt2021confiso}. We call both, $\confiso{3-j}$ and $\confiso{\cM}$, \emph{configuration space projections}, see also figure \ref{fig-moirelattice}.

\par Before continuing, we briefly want to comment why we can simultaneously choose $(\mathcal{R}_1,\mathcal{R}_2)$ \emph{and} $(\mathcal{R}_1^*,\mathcal{R}_2^*)$ to be incommensurate. For that, we introduce the notion of \emph{Diophantine $2D$-rotations}, that corresponds to a quantified version of incommensurability. 

\begin{definition}\label{defi-diophantine}
	Given $q>0$, we call $\theta\in\R$ a $(q,K,\sigma)$-\emph{Diophantine 2D rotation} iff $K>0,\sigma>0$, and for all $n\in\Z^2\setminus\{0\}$
	\begin{equation}\label{def-diophantine}
		\dist\big(q^{\pm1} A^TR_{\pm \theta}A^{-T} n,\Z^2\big) \geq \frac{K}{|n|^{2\sigma}} \, .
	\end{equation}
\end{definition}

This notion of Diophantine 2D rotations is exactly the condition that allows us to quantify the error in the ergodic approximation, see Proposition \ref{prop-ergod-thm} below. As in the case of Diophantine numbers, we have that Diophantine 2D rotations are generic.

\begin{proposition}\label{prop-dio-existence}
	For all $q>0$, $\sigma>\frac{1433}{1248}\approx 1.1482$ and for Lebesgue-almost all $\theta\in\R$ there exists $K>0$ s.t. $\theta$ is a $(q,K,\sigma)$-Diophantine 2D rotation.
\end{proposition}

Since almost every angle $\theta$ is a Diophantine 2D rotation, the corresponding lattice $(\mathcal{R}_1,\mathcal{R}_2)$ is incommensurate. By applying the same argument to the reciprocal lattices, we also have that for almost every $\theta$, $(\mathcal{R}_1^*,\mathcal{R}_2^*)$ is incommensurate. In particular, for (Lebesgue-)almost every $\theta$, both, $(\mathcal{R}_1,\mathcal{R}_2)$ and $(\mathcal{R}_1^*,\mathcal{R}_2^*)$, are incommensurate.

\par As the classical results for the one-dimensional case, our proof of Proposition \ref{prop-dio-existence}, see Appendix \ref{app-dio-proof}, relies on the Borel-Cantelli Lemma.

\par As usual, the absolute value function $|\cdot|$ will have a context-dependent meaning. For finite sets $M$, $|M|$ denotes the cardinality of $M$, for bounded measurable sets $\Omega\subseteq \R^n$, $|\Omega|$ denotes the respective Lebesgue measure, and for $x\in\R^n$, $|x|$ denotes the Euclidean norm.

\begin{remark}
	Bourgain-Watt \cite{bourgain2017mean} proposed the bound
	\begin{align}
		|B_r\cap \Z^2| \, = \, \pi r^2 \, + \, O_{r\to\infty}(r^{\frac{517}{824}+\vep}) 
	\end{align}
	for any $\vep>0$, where $\frac{517}{824}\approx 0.6274$, compared to $\frac{131}{208}\approx0.6298$ in \ref{eq-sq-lattice-ball-asymp} in the proof of Proposition \ref{prop-dio-existence}. In particular, this would allow us to lower the lower bound of $\sigma$ to $\frac{5671}{4944}\approx 1.1471$ compared to $\frac{1433}{1248}\approx 1.1482$ above.
	\par The general lattice problem is a very active research area. For related works, we refer to \cite{bourgain2016oppenheim,eskin2005quadratic,goetzemargulis2022distribution,ghosh2018oppenheim} and references therein.
\end{remark}

\subsubsection{Energies\label{sec-energies}}

As described in the introduction, we assume that the lattice displacements $u_j$ are periodic functions of their respective disregistry. Since we assume that $(\cR_1,\cR_2)$ is incommensurate, this is equivalent to parametrizing $u_j$ as moir\'e-periodic functions instead. For simplicity, we assume that $u_j\in C_{\mathrm{per}}(\MoireCell;(\R^3)^{\cA_j})$, s.t. the total lattice positions are given by  
\begin{align}\label{eq-lat-pos}
	Y_j(\gamma_j+R_j,\alpha_j) \, := \, \gamma_j \, + \, R_j \, + \, \tau_j^{(\alpha_j)} \, + \, u_j(\mfrac{\gamma_j+R_j},\alpha_j) \, .
\end{align}
Here, we identified the $\R^2$-valued points $R_j$, $\gamma_j$, $\tau_j^{(\alpha_j)}$ with their respective $\R^3$-embeddings $(R_j,0)$ resp. $(\gamma_j,0)$ resp. $(\tau_j^{(\alpha_j)},0)$. We allow the lattice displacement $u_j$ to take values in $\R^3$, in order to account for relaxation in $z$-direction.
\par By identifying $u_j$ with its periodic extension to $\R^2$, we also define the generalized total lattice positions
\begin{equation} \label{eq-gen-lat-pos}
\gls{Yjx} \, := \, x \, + \, \tau_j^{(\alpha_j)} \, + \, u_j(x,\alpha_j) \, , \quad (x,\alpha_j)\in\R^2\times\gls{cAj} \, . 
\end{equation}
Throughout this work, we abbreviate
\begin{align}\label{def-bu-bgam-balph}
\gls{bu} \, = \, (u_1,u_2) \, , \quad \gls{bgam}=(\gamma_1,\gamma_2)\in\Gamma_1\times\Gamma_2 \, , \quad \gls{balph}=(\alpha_1,\alpha_2)\in\cA_1\times\cA_2 \, .
\end{align}

As stated above, we assume that the total energy of the coupled bilayer system can be decomposed into a monolayer and interlayer contribution. For each, we assume that they can be expressed via sums over site-energies over the individual layers. We assume that these site-energies depend on the finite difference stencils originating at the respective site, as well as the sublattice degree type. 

\par We will now present the monolayer and interlayer energy densities, and explain how their thermodynamic limits can be obtained via ergodic theorems. Here, we state our results in a general context that only requires the existence of the limiting energy functional, in order to state results that hold both, in the cases of smooth and non-smooth displacements. In the case of smooth displacements and for lattice shifts $\gamma_j=0$, we even establish a rate of convergence, see Section \ref{sec-statement-of-results-smooth}. We postpone a discussion of the non-smooth case to Section \ref{sec-rough-displacements}.

\par We present our results in terms of many-body potentials dependent on finite difference stencils, which yield a more accurate description of materials than just interatomic pair potentials, such as Kolmogorov-Crespi \cite{KolmogorovCrespi}. However, to better illustrate our ideas, we decided to also include the case of interlayer pair potentials.

\paragraph{Monolayer potential}

Let
\begin{equation}
\gls{Vmonoj} \, : \, (\R^3)^{\gls{cRj}\times\gls{cAj}} \to\R 
\end{equation}
with $\alpha_j\in\gls{cAj}$ be given. We assume that $\gls{Vmonoj}$ only depends on the finite (relative) difference stencils 
\begin{equation}\label{eq-rel-diff-stenc}
	\big(u_j(R_j+\gamma_j+R_j',\alpha_j') - u_j(R_j+\gamma_j,\alpha_j)\big)_{\substack{R_j'\in\gls{cRj},\alpha_j'\in\gls{cAj}}} \, .
\end{equation}
The reason, for which we may consider dependence only on relative difference stencils for the monolayer potential, is that the index $R_j'$ encodes a distance between the evaluation points in the finite differences; this distance is independent of the reference point $R_j+\gamma_j$. 

 \begin{example}[Monolayer pair potentials]\label{ex-mono-PP}
	In order to illustrate the needed decay assumptions, we may consider pair potentials $w_j$. For simplicity, let us reduce to a single sublattice degree of freedom and ignore lattice shifts. Then the model many-body site potential takes the form
	\begin{equation}
		R_j\mapsto \sum_{R_j'\in\cR_j\setminus\{0\}}\big(w_j\big(R_j'+u_j(R_j+R_j')-u_j(R_j)\big)-w_j(R_j')\big) \, .
	\end{equation}
 \end{example}
 
 For this reason, one needs to assume a decay assumption in $R_j'$ in the case of monolayer potentials. Below, we will state assumptions on the site potentials
\begin{equation}
    R_j\mapsto\Vmono_{j,\alpha_j}\Big(\big(u_j(R_j+\gamma_j+R_j',\alpha_j') - u_j(R_j+\gamma_j,\alpha_j)\big)_{\substack{R_j'\in\gls{cRj},\alpha_j'\in\gls{cAj}}} \Big)
\end{equation}
that allow us to pass to a thermodynamic limit. For clarity, we omit sufficient assumptions on $\Vmono_{j,\alpha_j}$ and $u_j$ individually which would imply our assumptions on the site potentials; these could be obtained by adapting the ideas of, e.g.,\cite{blanc2002molecular,ortnertheil2013}.

For $\gamma_j\in\gls{Gamj}$, we introduce the monolayer energy density

\begin{align}\label{def-monolayer}
\begin{split}
	& \mono_{j,N,\gamma_j}(u_j) \, :=\frac1{|\gls{Gamj}|(2N+1)^2}\sum_{\alpha_j\in\gls{cAj}}\sum_{R_j\in \gls{cRjN}}\\
	& \quad \quad \gls{Vmonoj}\Big(\big(u_j(R_j+\gamma_j+ R_j',\alpha_j') \, - \, u_j(R_j+\gamma_j,\alpha_j)\big)_{(R_{j}',\alpha_j')\in\cR_{j}\times \gls{cAj}}\Big) \, ,
\end{split}
\end{align}
where 
\begin{equation}
|A_j[-N-1/2,N+1/2]^2| \, = \, |\gls{Gamj}|(2N+1)^2    
\end{equation}
is the area of the corresponding truncated region covered by the nuclei in layer $j$. 

\par Observe that, since we defined $u_j$ as $\gls{cRm}$-periodic extensions, 
\begin{align}\label{def-sitepotmono}
\begin{split}
	\MoveEqLeft \gls{SPmonoj} [u_j](x) \, := \\
	&\frac1{|\gls{Gamj}|}\sum_{\alpha_j\in\gls{cAj}} \gls{Vmonoj}\Big(\big(u_j(x+R_j',\alpha_j') \, - \, u_j(x,\alpha_j)\big)_{\substack{R_j'\in\gls{cRj},\\\alpha_j'\in\gls{cAj}}}\Big) 
\end{split}
\end{align}
is $\gls{cRm}$-periodic. If $R_j\in\gls{cRj}$, $\gamma_j\in\gls{Gamj}$, $\gls{SPmonoj} [u_j](R_j+\gamma_j)$ denotes the sublattice-averaged site-potential experienced at $R_j+\gamma_j$. In particular, \eqref{def-monolayer} becomes
\begin{align} \label{eq-mono-energy-sitepot}
\begin{split}
	\mono_{j,N,\gamma_j}(u_j) \, = \, \frac1{(2N+1)^2}\sum_{R_j\in \gls{cRjN}} \gls{SPmonoj} [u_j](\mfrac{R_j+\gamma_j}) \, . 
\end{split}
\end{align}  
Then, the ergodic theorem Proposition \ref{prop-single-erg-thm} implies, after substituting $x\mapsto A_j x$, that the thermodynamic limit of the monolayer energy is given by
\begin{align} 
	\MoveEqLeft\mono_{j}(u_j) \, :=\\
	& \frac1{|\gls{Gamj}|}\sum_{\alpha_j\in\gls{cAj}}\mavintrescj \dx{\xi}  \gls{Vmonoj}\Big(\big(u_j(A_j\xi+R_j',\alpha_j') \, - \, u_j(A_j\xi,\alpha_j)\big)_{\substack{R_j'\in\gls{cRj},\\\alpha_j'\in\gls{cAj}}}\Big) \\
	&= \, \mavint \dx{x} \gls{SPmonoj} [u_j](x) \, ,\label{def-mono-TD}
\end{align}
where $\avint_\Omega=\frac1{|\Omega|}\int_\Omega$ denotes the averaging integral over $\Omega$. This result is given in Theorem \ref{thm-smooth-conv}.

\paragraph{Interlayer coupling}
\begin{definition}\label{def-translation-inv}
A function 
\begin{equation}
	V \, : \, (\R^3)^{\gls{cRj}\times\gls{cAj}} \to\R 
\end{equation} 
is \emph{translation-invariant} iff for any $(a_{R_j,\alpha_j})_{(R_j,\alpha_j)\in\gls{cRj}\times\gls{cAj}}\in (\R^3)^{\gls{cRj}\times\gls{cAj}}$ and any $R_j'\in \gls{cRj}$, we have that
\begin{equation} V\big((a_{R_j+R_j',\alpha_j})_{(R_j,\alpha_j)\in\gls{cRj}\times\gls{cAj}}\big) \, = \, V\big((a_{R_j,\alpha_j})_{(R_j,\alpha_j)\in\gls{cRj}\times\gls{cAj}}\big) \, . 
\end{equation}
In addition, $V$ is \emph{even} iff $V(x)=V(-x)$.
\end{definition}

Let
\begin{equation}
\gls{Vinterj} \, : \, (\R^3)^{\cR_{3-j}\times\cA_{3-j}} \to\R \, ,
\end{equation}
with $\alpha_j\in\gls{cAj}$ be even and translation-invariant. As a special instance, let $\gls{pairpot}:\R^3\to\R$, with $\alpha_k\in\cA_k$, $k=1,2$, $\gls{balph}=(\alpha_1,\alpha_2)$, denote an even pair potential, and consider
\begin{align}\label{def-pair-potential}
\begin{split}
	\MoveEqLeft \gls{Vinterj}\big((a_{R_{3-j},\alpha_{3-j}})_{\substack{R_{3-j}\in\cR_{3-j},\\\alpha_{3-j}\in\cA_{3-j}}}\big)  \, = \, \frac12\sum_{\substack{R_{3-j}\in\cR_{3-j}\\\alpha_{3-j}\in\cA_{3-j}}}\gls{pairpot}(a_{R_{3-j},\alpha_{3-j}}) \, .
\end{split}
\end{align}
Such interatomic pair potentials $\gls{pairpot}$ have been constructed using DFT, e.g., in \cite{leven2016interlayer,Hod10,schmidt2015interatomic}.

\par The fact that we are using the ordered index $\gls{balph}=(\alpha_1,\alpha_2)$ to label pair potentials, as opposed to, e.g., $(\alpha_j,\alpha_{3-j})$, is to indicate that pair interactions only depend on the pair of involved atoms, as opposed to the order of the pair. Whenever there is a $\alpha_j$ or $\alpha_{3-j}$ on the LHS of an equation, and $\alpha_1$ or $\alpha_2$ on the RHS, the corresponding equation is to be read individually for $j=1$ or $j=2$, to avoid confusion.

We assume that $\gls{Vinterj}$ depends on the finite (total) difference stencils 
\begin{equation}
	\big( Y_j(R_j+\gamma_j,\alpha_j) -Y_{3-j}(R_{3-j}+\gamma_{3-j},\alpha_{3-j})\big)_{\substack{R_{3-j}\in\cR_{3-j},\alpha_{3-j}\in\cA_{3-j}}} \, .
\end{equation}
In the case of the interlayer potential, the distance between the evaluation points does depend on the reference point $R_j+\gamma_j$, or its disregistry, respectively.

\begin{example}[Interlayer pair potentials]\label{ex-int-PP}
	For simplicity we can again consider the special case of pair potentials. Let us again restrict to a single sublattice degree of freedom, and assume $\gamma_j=0$. Then the interlayer site-potential is given by
	\begin{align}
		\MoveEqLeft R_j\mapsto \sum_{R_{3-j}\in\cR_{3-j}}v\big(R_{3-j}+u_{3-j}(R_{3-j})-R_j-u_j(R_j)\big)\\
		& = \, \sum_{R_{3-j}\in\cR_{3-j}}v\big(R_{3-j}-\Lfrac{R_j}{3-j}+u_{3-j}(R_j+R_{3-j}-\Lfrac{R_j}{3-j})-u_j(R_j)\big) \, ,
	\end{align}
	where we recall \eqref{def-Lfrac-Lfloor}. In contrast to monolayer potentials, we now have to correct the positions by their local disregistry. For a generalization of this idea to many-body potentials, we refer to Lemma \ref{lem-site-pot}. Here one could also, analogously to Example \ref{ex-mono-PP}, subtract the constant contribution $\sum_{R_{3-j}\in\cR_{3-j}}v(R_{3-j}-R_j)$. Since it does not affect our results, we chose not to include it.
\end{example}

 Consequently, we need to impose a decay assumption which is based on the magnitude of the total distance $Y_j(x)-Y_{3-j}(y)$ in the argument. We prove the existence of a thermodynamic limit based on regularity of the site potentials
\begin{equation}
    R_j\mapsto \Vinter_{j,\alpha_j}\Big(\big( Y_j(R_j+\gamma_j,\alpha_j) -Y_{3-j}(R_{3-j}+\gamma_{3-j},\alpha_{3-j})\big)_{\substack{R_{3-j}\in\cR_{3-j},\alpha_{3-j}\in\cA_{3-j}}}\Big) \, .
\end{equation}
Again, we expect that this regularity would follow from appropriate assumptions on $\Vinter_{j,\alpha_j}$ and $u_j$ along the lines of \cite{blanc2002molecular,ortnertheil2013}.

Recalling \eqref{eq-gen-lat-pos}, define the interlayer coupling energy density of the $j$-th layer
\begin{align} \label{def-interlayer}
\begin{split}
	\inter_{j,N,\gls{bgam}}(\gls{bu}) \, &:= \, \frac1{|\gls{Gamj}|(2N+1)^2}\sum_{R_j\in \gls{cRjN}} \sum_{\alpha_j\in\gls{cAj}}\gls{Vinterj}\Big(\big(Y_j(R_j+\gamma_j,\alpha_j)\\
	&\qquad - \, Y_{3-j}(R_{3-j}+\gamma_{3-j},\alpha_{3-j})\big)_{\substack{R_{3-j}\in\cR_{3-j},\\\alpha_{3-j}\in\cA_{3-j}}}\Big) \, .
\end{split}
\end{align}

We would like to study the thermodynamic limit of the energy density \eqref{def-interlayer} by employing an appropriate ergodic theorem. Such an ergodic theorem lets us map a functional defined on the discrete systems $\gls{cRj}$ to a functional defined on the continuous system $\gls{GamM}$, via the configuration space isomorphism $\confiso{\cM}$. As a useful tool to map in between unit cells of the involved lattices, we recall from \eqref{def-disregj}, see also \eqref{eq-disregj-moire-per}, the disregistry matrices:
	\begin{align}
\bot \, = \, I \, - \, A_2A_1^{-1} \, = \, -A_2\MoireRLV^{-1} \, , \quad \bto \, = \, I \, - \, A_1A_2^{-1} \, = \, A_1\MoireRLV^{-1} \, .
\end{align}
In particular, we have that
\begin{equation}\label{eq-disregj-maps-unitcells}
\gls{disregj}\gls{GamM} \, = \, (-1)^j\Gamma_{3-j} \, .
\end{equation}
The disregistry matrices $\gls{disregj}$ have the crucial property that they map between the different notions of configurations if restricted to the individual non-shifted lattices, see Lemma \ref{lem-disreg-transf}.

\par As above, we define the sublattice-averaged generalized interlayer site potentials 
\begin{align}\label{def-sitepotdouble}
\begin{split}
	\MoveEqLeft\gls{SPinterj} [\gls{bu}](x,y) \, :=\\
	&\frac1{|\gls{Gamj}|}\sum_{\alpha_j\in\gls{cAj}} \gls{Vinterj}\Big(\big(y -R_{3-j}+ \tau_j^{(\alpha_j)}- \tau_{3-j}^{(\alpha_{3-j}) } + u_j(x,\alpha_j)\\
	&  -\, u_{3-j}(x-y+R_{3-j},\alpha_{3-j})\big)_{\substack{R_{3-j}\in\cR_{3-j},\\\alpha_{3-j}\in\cA_{3-j}}}\Big) 
\end{split}
\end{align}
for all $x,y\in\R^2$. Employing Lemma \ref{lem-site-pot}, we obtain that
\begin{align} \label{eq-inter-energy-sitepot}
\begin{split}
	\MoveEqLeft\inter_{j,N,\gls{bgam}}(\gls{bu}) \, =\\
	&\frac1{(2N+1)^2}\sum_{R_j\in \gls{cRjN}}\gls{SPinterj} [\gls{bu}](\mfrac{R_j+\gamma_j},\Lfrac{R_j+\gamma_j-\gamma_{3-j}}{3-j}) \, .
\end{split}
\end{align}
Observe that, due to the $\gls{cRm}$-periodic extension of $u_j$ and due to the dependence on the collection of all lattice shifts with respect to $\cR_{3-j}$, $\gls{SPinterj} [u]$ is $\gls{cRm}\times\cR_{3-j}$-periodic. 

Here, we prove a mean ergodic theorem, Proposition \ref{prop-double-erg-thm}: Let $h\in L^1_{\emph{loc}}(\R^4)$ be $\gls{cRm}\times\cR_{3-j}$-periodic. Recalling \eqref{eq:layer_and_moire_disreg}, we then show that 
\begin{align}
\begin{split}
	\MoveEqLeft \frac1{(2N+1)^2 }\sum_{R_j\in\gls{cRjN}}h(\mfrac{R_j+\omega_\cM},\Lfrac{R_j+\omega_{3-j}}{3-j})\\
	& \xrightarrow[]{N\to\infty} \, \mavint \dx{x}h(x,\gls{disregj} x+A_{3-j}A_j^{-1}\omega_\cM-\omega_{3-j})\\
    & = \, \mavintrescj \dx{\xi} h\big(A_j\xi+\omega_\cM,(A_j-A_{3-j})\xi+\omega_{3-j}\big) 
\end{split}
\end{align}
converges for almost all $(\omega_\cM,\omega_{3-j})\in\gls{GamM}\times\Gamma_{3-j}$ and also in $L^1(\gls{GamM}\times\Gamma_{3-j})$. Both, this result and also the previously mentioned ergodic theorem Prop. \ref{prop-single-erg-thm}, are special instances of a more general ergodic theorem for group actions, as can be found, e.g., in \cite{Damanik2022Oerg-schroedinger-book,kechris2010global,KerrLi-erg-th}, see also Proposition \ref{prop-pw-erg-thm-gen}. These ideas have been previously applied to incommensurate multilayer structures, see \cite[Theorem 2.1]{massatt-luskin-ortner2017electronicdensity}, \cite[Proposition 3.5]{generalizedkubo2017}.

\par We then obtain in Theorem \ref{thm-rough-conv} that, in the thermodynamic limit $N\to\infty$, the interlayer coupling potential converges to
\begin{align}\label{def-inter-TD-pair}
	\begin{split}
		\inter_{j,\gls{bgam}}(\gls{bu}) \, :=& \,\frac1{|\gls{Gamj}|}\sum_{\alpha_j\in\gls{cAj}}\mavintrescj \dx{\xi}  \gls{Vinterj}\Big(\big( Y_j(A_j\xi+\gamma_j,\alpha_j)\\
		& \quad -Y_{3-j}(A_{3-j}\xi+\gamma_{3-j}+R_{3-j},\alpha_{3-j})\big)_{\substack{R_{3-j}\in\cR_{3-j},\\\alpha_{3-j}\in\cA_{3-j}}}\Big) \\
	    =& \,  \mavint \dx{x} \gls{SPinterjdiag} [\gls{bu}](x) \, ,
	\end{split}
\end{align}
where
\begin{equation}\label{def-sitepot} 
	\gls{SPinterjdiag} [\gls{bu}](x) \, := \, \gls{SPinterj} [\gls{bu}](x,\gls{disregj} x+A_{3-j}A_j^{-1}\gamma_j-\gamma_{3-j}) \, .
\end{equation}
Lemma \ref{lem-MBpot-pair-FT-formula} yields in the special case of even pair potentials $\gls{pairpot}$ that
\begin{align}\label{def-inter-TD-pair-simplify}
	\begin{split}
		\inter_{j,\gls{bgam}}(\gls{bu}) \, & =\, \frac{1}{2|\gls{GamM}|}\sum_{\substack{\alpha_k\in\cA_k\\k=1,2}}\int_{\R^2} \dx{\xi} \gls{pairpot}\big(A_1\xi + \gamma_1  +\tau_1^{(\alpha_1)}+u_1(A_1\xi+\gamma_1,\alpha_1)\\
		& \quad -A_2\xi-\gamma_2-\tau_2^{(\alpha_2)}- u_2(A_2\xi+\gamma_2,\alpha_2)\big) \, .
	\end{split}
\end{align}

\subsection{Results for smooth displacements\label{sec-statement-of-results-smooth}}

We now present our main result for the case of smooth displacements under some simplifying assumptions. Under these assumptions, our notion of Diophantine 2D rotations, recall Definition \ref{defi-diophantine}, allows us to establish a rate of convergence to the thermodynamic limit of the atomistic energy density. We emphasize that, even in this setting, our limiting energy functional appears to be novel. 

We start by assuming that $\gamma_1=\gamma_2=0$, and we abbreviate 
\begin{align}
    \Phi^{\mathrm{(inter)}}_j := \sitepot_{j,\bf{0}} \, , \quad e^{\mathrm{(inter)}}_{j,N}:=\inter_{j,N,\bf{0}} \, , \quad
    e^{\mathrm{(inter)}}_j  := \inter_{j,\bf{0}} \, , \quad   e^{\mathrm{(mono)}}_{j,N} :=\mono_{j,N,0} \, .
\end{align}
We denote
\begin{equation} \label{def-subldist}
	\subldist \, := \, \max_{\substack{\alpha_{\ell}\in\cA_\ell\\ \ell=1,2}}|\tau_{1}^{(\alpha_1)}- \tau_{2}^{(\alpha_2)}| \, ,
\end{equation}
and for $\gls{bu}=(u_1,u_2)\in L^\infty(\R^2;\R^6)$, with components $u_{j,k}$, $j=1,2$, $k\in\{x,y,z\}$, 
\begin{equation}\label{def-Luz} 
	L_{\gls{bu}}^z \, := \, \subldist \, + \,  \|u_{1,z}\|_\infty + \|u_{2,z}\|_\infty \, .
\end{equation}
Moreover, let 
\begin{equation}\label{def-jb}
	\jb{x} \, := \, (1+|x|^2)^{\frac12}
\end{equation}
for any $x\in\R^n$, $n\in\N$. For $k,n\in\N$ and $s>0$ and a Lipschitz domain $\Omega\subseteq \R^n$, let 
\begin{equation}\label{def-weighted-sob}
	\|f\|_{W^{k,\infty}_s(\Omega)} \, := \, \sum_{j=0}^k \|\jb{\cdot}^{s}D^jf\|_{L^\infty(\Omega)}
\end{equation}
denote the weighted Sobolev norm, and let 
\begin{equation}\label{def-weighted-sob-space}
	W^{k,\infty}_s(\Omega) \, := \, \{f\in W^{k,\infty}(\Omega) \, \mid \, \|f\|_{W^{k,\infty}_s(\Omega)} \, < \, \infty \}
\end{equation}
denote the associated weighted Sobolev space. In addition, let 
\begin{equation}\label{def-weighted-leb}
	L^\infty_{s}(\Omega) \, := \, \{ f\in L^\infty(\Omega) \, \mid \, \|\jb{\cdot}^sf\|_{L^\infty(\Omega)} \, < \, \infty \}
\end{equation}
denote the weighted Lebesgue space.
\par For $f\in L^1(\gls{GamM})$ and $\MoireG\in\gls{cRm*}$, we define the Fourier transform
\begin{align}\label{def-FT}
\hat{f}(\MoireG) \, := \, \gls{FTMoire}(f)(\MoireG) \, := \,  \mavint \dx{x} e^{-i\MoireG\cdot x} f(x) \, .
\end{align}
The next result is our main result in the case of smooth displacements.

\begin{theorem}\label{thm-smooth-conv}
    Let $s>1$, $\sigma>\frac{1433}{1248}$, and $\theta\in[0,2\pi)$ be a $(q,K,\sigma)$-Diophantine 2D rotation for some $K>0$. Let $\ell\in\{\mathrm{mono},\mathrm{inter}\}$ and $j\in\{1,2\}$. Assume that \begin{equation}
		\Phi^{\mathrm{(\ell)}}_j [u_j]\in L^1(\gls{GamM}) \, , \,  \gls{FTMoire}\big(|\nabla|^{2(\sigma+s)}\Phi^{\mathrm{(\ell)}}_j [u_j]\big)\in \ell^\infty(\gls{cRm*}) \, .
	\end{equation}
    In particular, if we assume pair interlayer potentials, assume that for $k=1,2$, $\alpha_k\in\cA_k$, $u_k(\cdot,\alpha_k)\in W^{6,\infty}(\R^2)$, and that, see \eqref{def-weighted-sob-space},
	\begin{align} \label{ass-valpha-reg}
		\gls{pairpot} \in W^{6,\infty}_{2r}(\R^2\times[-L_{\gls{bu}}^z,L_{\gls{bu}}^z]) \, .
	\end{align}
	Then there exists a constant $C>0$ such that
    \begin{equation}
        \big|e^{(\ell)}_{j,N}-e^{(\ell)}_j\big| \, \leq\, \frac{C}{2N+1} \, .
    \end{equation}
\end{theorem}

In the case of many-body potentials, this theorem is a consequence of the quantitative ergodic theorem, Proposition \ref{prop-ergod-thm}. The proof that the conditions in the case of interlayer pair potentials are sufficient, is contained in Section \ref{sec-proofs-main-thm}.

\begin{remark}
    We provide bounds on the constants $C>0$ that are explicit, except for the implicit dependence of the Diophantine constant $K>0$ on the chosen angle $\theta$ and the decay $\sigma$. In the case of many-body potentials, we obtain
    \begin{align}\label{def-error-MB}
	\begin{split}
		& \frac{2\sqrt{2}}{K}\Big(\frac{\gls{moirelen} \|A\|_2}{2\pi}\Big)^{2(\sigma+s)}\big(\zeta(2s)+2^{-s}\zeta(s)^2\big)\\
		& \quad  \sup_{\MoireG\in\gls{cRm*}}|\MoireG|^{2(\sigma+s)} |\gls{FTMoire}(\Phi^{\mathrm{(\ell)}}_j [u_j])(\MoireG)| \, , 
	\end{split}
\end{align}
where $\zeta(\sigma)$ denotes the Riemann zeta function, and we recall $\gls{moirelen}$ from \eqref{def-moirelen}. Now abbreviate
\begin{align}\label{def-dioLB}
	\begin{split}
		\MoveEqLeft\dioLB(\theta,\sigma) \, := \\
		\sup&\bigg(\bigcup_{N\in\N}\bigcap_{\substack{n\in\Z^2:\\|n|^2\geq N}}\bigg\{K>0\mid\dist\big(A^TR_{\pm\theta} A^{-T} n,\Z^2\big) \geq \frac{K}{|n|^{2\sigma}} \bigg\}\bigg) \, ,
	\end{split}
\end{align}
and define
\begin{equation}\label{def-pairUB}
	\pairUB(\theta) \, := \, \inf_{\sigma\in\big(\frac{1433}{1248},2\big)}\frac1{\dioLB(\theta,\sigma)}\big(\zeta(6-2\sigma)+2^{\sigma-3}\zeta(3-\sigma)^2\big) \, .
\end{equation}
Then an upper bound on the constant $C>0$ in the case of interlayer pair potentials is given by
\begin{align}\label{def-errTB}
	\begin{split}
		& \frac{203\sqrt{2}}{|\Gamma_1||\Gamma_2|}\Big(\frac{\gls{moirelen} \|A\|_2}{2\pi}\Big)^6 \pairUB(\theta) (q^6+q^{-6}+1)\frac{5^{r-1}\pi}{r-1}\\
		& \quad \big(1+\subldist+\|u_1\|_{W^{6,\infty}} + \|u_2\|_{W^{6,\infty}}\big)^{6+2r} \sum_{\substack{\alpha_k\in\cA_k\\k=1,2}} \|\gls{pairpot}\|_{W^{6,\infty}_{2r}(\R^2\times[-L_{\gls{bu}}^z,L_{\gls{bu}}^z])} \, .
	\end{split}
\end{align}

\end{remark}

\begin{remark}
	One could reduce the regularity assumption 
	\begin{equation}
		\gls{FTMoire}\big(|\nabla|^{2(\sigma+s)}\gls{SPmonoj} [u_j]\big)\in \ell^\infty(\gls{cRm*})
	\end{equation} 
	in Theorem \ref{thm-smooth-conv} to assumptions on $V^{(\ell)}_{j,\alpha_j}$ and $u_j$ individually, in analogy to the case of pair potentials. However, in order to maintain a clear presentation, we decided to omit such a result, and leave the details to the interested reader.
\end{remark}

\subsection{Discussion of results}

\subsubsection{Regularity assumptions} 

We now want to comment on the imposed regularity on $u_j$ and $\gls{pairpot}$. Our results require smoothness of $u_j$, i.e., $u_j\in W^{6,\infty}(\R^2)$, in order to show convergence with a rate. For an almost sure ergodic theorem to hold, see Proposition \ref{prop-double-erg-thm}, we only require $u_j \in L^\infty(\R^2)$. We address this case in Section \ref{sec-rough-displacements} below.

\par We do not expect general displacements $u_j$ to be smooth. As explained above, if we consider a model consisting of a monolayer contribution and an interlayer coupling as constructed above, this has been analyzed for the Frenkel-Kontorova model \cite{aubry1983frenkel-kontorova} and a simplified coupled linear chain \cite{cazeauxrippling2017} in the one dimensional case. Their result is that the energy minimizer is a continuous function of the disregistry whenever the monolayer energy dominates the interlayer coupling. However, we do not address the question of smoothness of minimizers in our work. In the Cauchy-Born approximation, see Section \ref{sec-CB} below, we believe that standard ideas in the theory of semi-linear elliptic equations, as applied in \cite{Cazeaux-Massatt-Luskin-ARMA2020}, can be utilized to show that, in case of $C^\infty$-smooth and polynomially decaying pair potentials $\gls{pairpot}$, critical points $\gls{bu}=(u_1,u_2)$ inherit $C^\infty$ smoothness. 
\par In order to obtain a rate of convergence, we require that the weighted Sobolev norm of $\gls{pairpot}$
\begin{align} 
	\sum_{j=0}^6 \|\jb{\cdot}^{2r}D^j\gls{pairpot}\|_{L^\infty(\R^2\times [-L_{\gls{bu}}^z,L_{\gls{bu}}^z])} \, < \, \infty 
\end{align}
for some $r>1$, is finite, where $\jb{\cdot}=(1+|x|^2)^{\frac12}$ and 
\begin{equation}
	L_{\gls{bu}}^z\, \sim \, 1+\|u_{1,z}\|_\infty +\|u_{2,z}\|_\infty \, ,
\end{equation}
where the involved constant depends on the lattice, see \eqref{def-Luz} below. In particular, $\gls{pairpot}$ obeys the tail-behavior 
\begin{equation}\label{eq-v-decay}
	|\gls{pairpot}({\bf x},z)| \, \lesssim \, \frac1{|{\bf x}|^{2+}} \quad \mbox{as\ } |{\bf x}|\to\infty
\end{equation}
for ${\bf x}\in\R^2$, $z\in [-L_{\gls{bu}}^z,L_{\gls{bu}}^z]$. As a special case, this includes the Yukawa potential $\frac{e^{-\kappa |{\bf x}|}}{|{\bf x}|}$, which accounts for screening effects. Observe that we need to regularize the potential near the origin. 
\par Notice that in the perpendicular $z$-direction, we have no tail constraints, but again we need to regularize the potential near the origin. In particular, we are free to choose the potential to behave like a Lennard-Jones type potential in $z$-direction, away from the origin, which accounts for the van der Waals interaction between the layers.
\par In the case that no rate of convergence is obtained, we impose 
\begin{equation}
	\|\jb{\cdot}^{2r}\gls{pairpot}\|_{L^\infty(\R^2\times [-L_{\gls{bu}}^z,L_{\gls{bu}}^z])} \, < \, \infty 
\end{equation}
for some $r>1$, amounting to the same decay rate \eqref{eq-v-decay}.

\subsubsection{Cauchy-Born Approximation\label{sec-CB}}

We want to further approximate the monolayer contribution $\mono$ by taking a continuum limit. This leads to the \emph{Cauchy-Born approximation}, which is a popular approximation in elasticity theory. For that, we introduce a length-scale $\atd\ll1$ and rescale
\begin{equation} \label{def-A0}
	A \, =: \, \atd A_0 \, ,
\end{equation}
where $A_0$ is independent of $\atd$. Consequently, $\gls{cRj}^{(0)}:=\gls{cRj}/\atd$ is $\atd$-independent. For simplicity, we assume only a single sublattice degree of freedom, $|\gls{cAj}|=1$. In addition, we restrict $u_j$ to take values in the $xy$-plane.

\par Recalling \eqref{def-sitepotmono}, and assuming that $\gls{bu}$ is sufficiently regular, say $\gls{bu}\in C^2(\gls{GamM};\R^4)$, we then obtain for $\atd\ll1$ that 
\begin{align}
	\begin{split}
		\gls{SPmonoj} [u_j](x) \, \approx& \, \frac1{|\gls{Gamj}|} \Vmono_{j}\Big(\atd Du_j(x)\cdot \gls{cRj}^{(0)} \Big) \\
		=& \, W_{j}\big(Du_j(x)\big) \, ,
	\end{split}
\end{align}
where $W_j$ denotes the \emph{Cauchy-Born energy density function}. In the linear elasticity approximation for $|Du_j|$ small enough, and assuming isotropy, one has 
\begin{align}
	W_j(M) \, \approx & \, \frac12 M:\cE_j:M \, , \quad \mbox{for\ all\ } M\in\R^{2\times2} \, , \label{def-W}
\end{align}
with the elasticity tensor $\cE_j$ given in \eqref{def-elast-tensor}. In particular, the monolayer energy density \eqref{def-mono-TD} satisfies for $\atd\ll 1$
\begin{align}\label{def-mono-CB}
	\mono_j(u_j) \, \approx \, \frac12\mavint \dx{x} Du_j(x):\cE_j:Du_j(x) \, .
\end{align}

\begin{example}\label{ex-tbg}
	Let us, for instance, consider twisted bilayer graphene.
	\par Let $\theta\in[0,2\pi]$ be an incommensurate angle and let 
	\begin{equation}
		\atd \, := \, \sqrt{3}\cdot 1.42 \mbox{\ \normalfont nm} \, \approx \, 2.46 \mbox{\ \normalfont nm} \, . \label{def-a0-TBG}
	\end{equation}
	Define
	\begin{align}
		A \, := \, \atd \begin{pmatrix}
			\frac{\sqrt{3}}2 & \frac{\sqrt{3}}2 \\
			-\frac12 & \frac12 
		\end{pmatrix} 
	\end{align}
	and recall from \eqref{def-Aj} that
	\begin{align}
		A_1 \, = \, R_{-\theta/2} A \, , \quad A_2 \, = \, R_{\theta/2} A \, .
	\end{align}
	Graphene has two sublattice degrees of freedom, denoted by $A$ and $B$, with associated shift vectors 
	\begin{align}
		\tau^{(A)}_j \, = \, -\frac\atd2 R_{(-1)^j\theta/2}\hat{e}_1 \, , \quad
		\tau^{(B)}_j \, = \, \frac\atd2 R_{(-1)^j\theta/2}\hat{e}_1 \, .
	\end{align} 
	For simplicity, we can ignore the sublattice degree of freedom in order to compute the lattice displacements at the center of mass. Then we can set the relaxed lattice positions to
	\begin{align}
		\Tilde{Y}_j(R_j+\gamma_j,\alpha) \, := \, R_j+\gamma_j + \tau_j^{(\alpha)} + u_j(R_j+\gamma_j)
	\end{align}
	for $j=1,2$, $\alpha=A,B$.
	\par The Lam\'e parameters for a single sheet of graphene are given by
	\begin{equation} \label{eq-graphene-lame}
		\lambda \, = \, 37,950 \mbox{\normalfont\ meV/unit area} \, , \quad \mu \, = \, 47,352 \mbox{\normalfont\ meV/unit area} \, ,
	\end{equation}
	see, e.g., \cite{Carr2018}. In order to allow for out-of-plane relaxation, we need to extend the monolayer energy model, e.g., to 
	\begin{align}
		\mavint \dx{x}\frac12\Big(\lambda \sum_{k=1}^2\vep_{k,k;j}(x)^2+2\mu \sum_{k,\ell=1}^2\vep_{k,\ell;j}(x)^2+\kappa |D^2u_{j,3}(x)|^2\Big) \, ,
	\end{align}
	where $\vep_{k,\ell;j}:=\frac12(\partial_k u_{j,\ell}+\partial_\ell u_{j,k}+\partial_ku_{j,3}\partial_\ell u_{j,3})$ is the strain tensor, $\lambda$, $\mu$ are the Lam\'e parameters as before, and $\kappa$ is the bending rigidity, see, e.g., \cite{peng2020strain}. As an interatomic pair potential, we may choose potentials that are radial in the horizontal direction. E.g., we may choose a Morse potential for the monolayer pair potentials, and a Lennard-Jones potential accounting for the van der Waals coupling between the layers, i.e., 
	\begin{equation}
		\gls{pairpot}(\bx,z) \, \equiv \, v(\bx,z) \, := \, v_{\mathrm{Morse}}(|\bx|)v_{\mathrm{LJ}}(|z|)
	\end{equation}
	with
	\begin{align}
		v_{\mathrm{Morse}}(r) \, :=& \, E_0 \Big[\Big(e^{-\kappa_0(r-r_0)}-1)^2-1\Big] \, , \\
		v_{\mathrm{LJ}}(r) \, :=& \, 4\vep_0\Big[\Big(\frac\sigma{r}\Big)^{12}-\Big(\frac\sigma{r}\Big)^{6}\Big] \, .
	\end{align}
	Here $E_0$ denotes the dissociation energy, $\kappa_0$ relates to the stiffness, $r_0$ refers to the equilibrium position, see \cite{schmidt2015interatomic}, $\vep_0,\sigma>0$ are van der Waals parameters. Strictly speaking, we need to regularize $v_z$ to remove the singularity at the origin for our methods to apply. 
	\par More specialized interlayers potentials have been developed for graphene \cite{KolmogorovCrespi} and hBN \cite{Hod10}, for example. 
\end{example}

\subsubsection{Relationship with Allen-Cahn energy functional\label{sec-AC}} 

As in the previous section, let us restrict to a single sublattice degree of freedom, $|\gls{cAj}|=1$, and assume that $\tau_j^{(\alpha_j)}=0$. Assume that $v$ is an interlayer pair potential that is radially decaying in $xy$-direction, and set $\gamma_j=0$. In addition, we restrict $u_j$ to take values in the $xy$-plane.

\par We claim that, in the Cauchy-Born approximation, $\etot(u_1,0),$ resp., $\etot(0,u_2),$ are modified Allen-Cahn functionals. For that, define a \emph{stacking fault potential} 
\begin{align}
	\misfit_j(x_{3-j}) \, := \, & \frac1{|\gls{Gamj}|}\sum_{R_{3-j}\in\cR_{3-j}}  v\big(x_{3-j} -R_{3-j}\big) \label{def-misfit} 
\end{align}
for all $x_{3-j}\in\Gamma_{3-j}$, $j=1,2$. We note that the generalized stacking fault energy (GSFE) more commonly used in numerical computations of relaxation \cite{Cazeaux-Massatt-Luskin-ARMA2020,Carr2018,cazeaux-clark-engelke-kim2022relaxation} approximates the misfit energy by a DFT computation on a periodic cell at zero twist angle. By construction, $\misfit_j$ is periodic with respect to $\cR_{3-j}$, $j=1,2$. $\misfit_j(x_{3-j})$ is maximal on the lattice points $R_{3-j}\in\cR_{3-j}$, which we will refer to as \emph{AA stacking}, and minimal at 
\begin{equation}
	R_{3-j}+\frac13 A_{3-j}\begin{pmatrix}
		1\\1
	\end{pmatrix} \, ,
\end{equation}
which we will refer to as \emph{AB stacking}, and at  
\begin{equation}
	R_{3-j}+\frac23 A_{3-j}\begin{pmatrix}
		1\\1
	\end{pmatrix} \, ,
\end{equation}
which we will refer to as \emph{BA stacking}, see figures \ref{fig:moire-per-disreg} and \ref{fig-GSFE}.
In the AA stacking, nuclei from different layers are stacked on top of each other, while in the AB/BA stacking, they sit in between three nuclei. 
\par Recalling \eqref{def-sitepotdouble} and \eqref{def-disregj}, observe that we have the identities
\begin{align}
	\sitepotdouble_{1}[u_1,0](x,\bot x) \, = \, \misfit_1(\bot x + u_1(x)) \, , \label{eq-sitepot-misfit1}\\
	\sitepotdouble_{2}[0,u_2](x,\bto x) \, = \, \misfit_2(\bto x + u_2(x)) \, . \label{eq-sitepot-misfit2}
\end{align}
Moreover, Lemma \ref{lem-MBpot-pair-FT-formula} implies that
\begin{align}
	\begin{split}
		\mavint \dx{x} \sitepot_{1,0}[\gls{bu}](x) \, = \, \mavint \dx{x} \sitepot_{2,0}[\gls{bu}](x) \, .
	\end{split}
\end{align}
In the Cauchy-Born approximation $\atd\ll1$, see \eqref{def-mono-CB}, the total energy functional becomes
\begin{align}
	\begin{split}
		\etot(\gls{bu}) \, \approx \, \frac12\sum_{j=1}^2\mavint \dx{x} \big(Du_j(x):\cE_j:Du_j(x) \, + \, \sitepot_{j,0}[\gls{bu}](x)\big) \, .
	\end{split}
\end{align}
In particular, employing \eqref{eq-sitepot-misfit1}, we have for $\atd\ll1$ that
\begin{align}
	\begin{split}
		\etot(u_1,0) \approx \mavint \dx{x} \bigg(\frac12 Du_1(x):\cE_1:Du_1(x) \, + \, \misfit_1\big(\bot x+u_1(x)\big)\bigg) \, .
	\end{split}
\end{align}
In order to recognize the scaling-behavior with $\theta\to0$, let us rescale $y:=\bot x$, and define $U_1(y):=u_1(\bot^{-1}y)$. For simplicity, we assume that the layers are purely twisted, i.e., $q=1$. Then \eqref{eq-moirerlv-formula} and \eqref{eq-moire-scale-matrix} below imply that 
\begin{equation}
    \MoireRLV \, = \, 2\sin(\theta/2) \sympl A \, , \quad \sympl \, := \, \begin{pmatrix}
        0 & -1 \\ 1 & 0 
    \end{pmatrix} \, .
\end{equation}
Due to $\bot=-A_2\MoireRLV^{-1}=2\sin(\theta/2)A_2A^{-1}\sympl$, see Lemma \ref{lem-disregj-calc}, we obtain, defining $\vep:=2\sin(\theta/2)$, for some $\tilde{\cE}_1$ 
\begin{equation}
    \avint_{\Gamma_2}\dx{y}\bigg(\frac{\vep^2}2 DU_1(y):\tilde{\cE}_1:DU_1(y) \, + \, \misfit_1\big(y+U_1(y)\big) \bigg) \, .
\end{equation}
Rescaling this energy by $\frac1{\vep}$, we thus obtain
\begin{equation}
    \avint_{\Gamma_2}\dx{y}\bigg(\frac{\vep}2 DU_1(y):\tilde{\cE}_1:DU_1(y) \, + \, \frac1\vep\misfit_1\big(y+U_1(y)\big) \bigg) \, ,
\end{equation}
the algebraic form of the expression obtained in \cite{Cazeaux-Massatt-Luskin-ARMA2020}. The obtained energy resembles the Allen-Cahn energy functional, except for the additional explicit dependence on $x$. Without setting $u_2=0$, one can show that for $U_2(y):=u_2(-\bto^{-1}y)$ the resulting rescaled energy takes the form
\begin{equation}\label{eq-IAP-rescaled}
    \frac{\vep}2\sum_{j=1}^2\avint_{\Gamma_{3-j}}\dx{y} DU_j(y):\tilde{\cE}_j:DU_j(y) \, + \, \frac1\vep\int_{\R^2}\dx{y}V\big(y+U_1(y)-U_2(y)\big)  \, ,
\end{equation}
where $\tilde{\cE}_2$ and $V=\frac{v}{|\Gamma_1||\Gamma_2|}$ are chosen appropriately. 

\subsubsection{Comparison with previous result\label{sec-conv-small-angle}}

We now give a more detailed comparison of the previous model \eqref{ARMA-energy} and our new model \eqref{our-energy} in the simplest possible setting. Again, we restrict to a single sublattice degree of freedom in each layer, assume that the layers are centered at the origin, $\gamma_j=0$, and not mismatched, $q=1$. We will argue that the stacking fault energy $\misfit_j$, see \eqref{def-misfit}, obtained in this work provides a microscopic model for the GSFE $\GSFE_j$ studied in \cite{Cazeaux-Massatt-Luskin-ARMA2020,Carr2018,cazeaux-clark-engelke-kim2022relaxation,relaxation-domain-wall-formation2018}, see fig. \ref{fig-GSFE}. However, the resulting model even under this modelling assumption is different from ours obtained in this work. 

\begin{figure}
    \centering
    \begin{subfigure}[c]{0.5\textwidth}
    \includegraphics[width=1\textwidth]{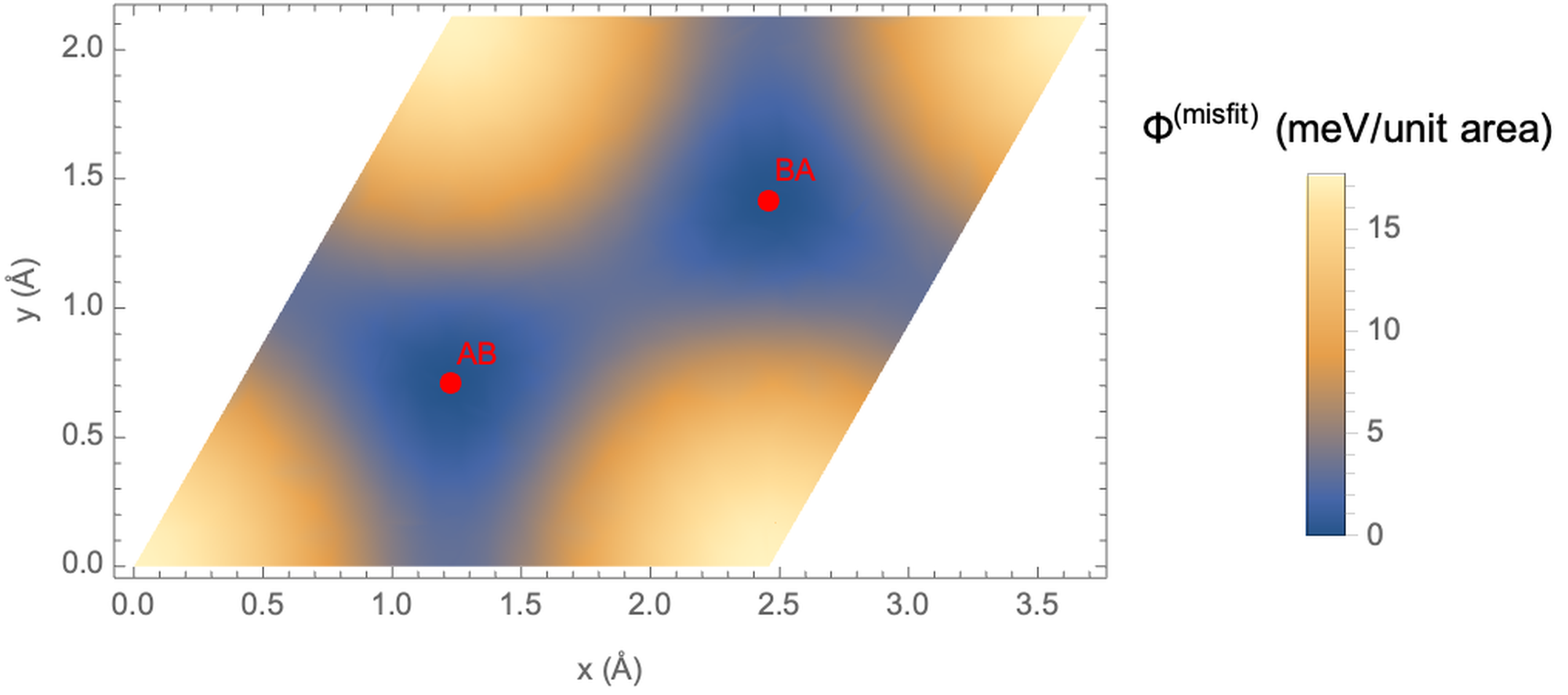}
    \caption{}
    \label{fig-gsfe-dft}    
    \end{subfigure}
    \begin{subfigure}[c]{0.4\textwidth}
        \includegraphics[width=1\textwidth]{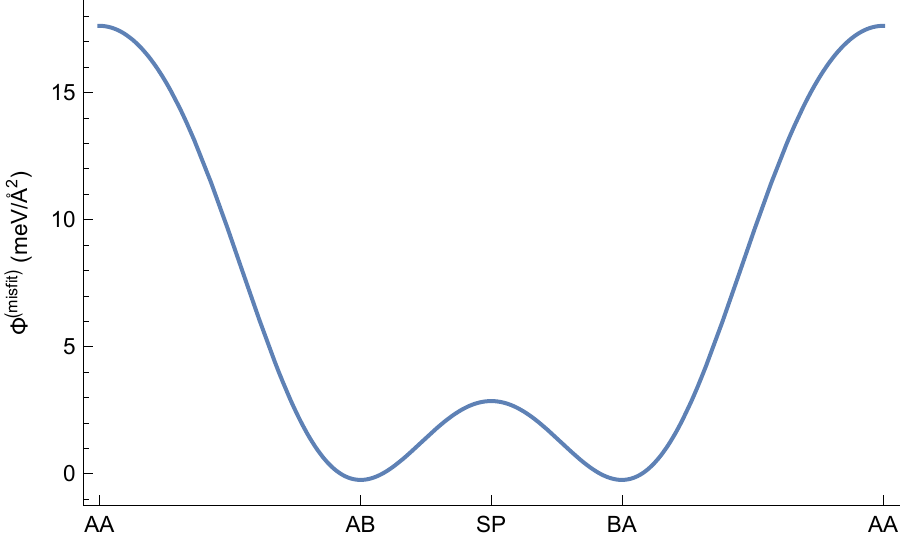}
    \caption{}
    \label{fig-gsfe-dft-diag}    
    \end{subfigure}
    \hfill
    \begin{subfigure}[c]{0.5\textwidth}
    \includegraphics[width=1\textwidth]{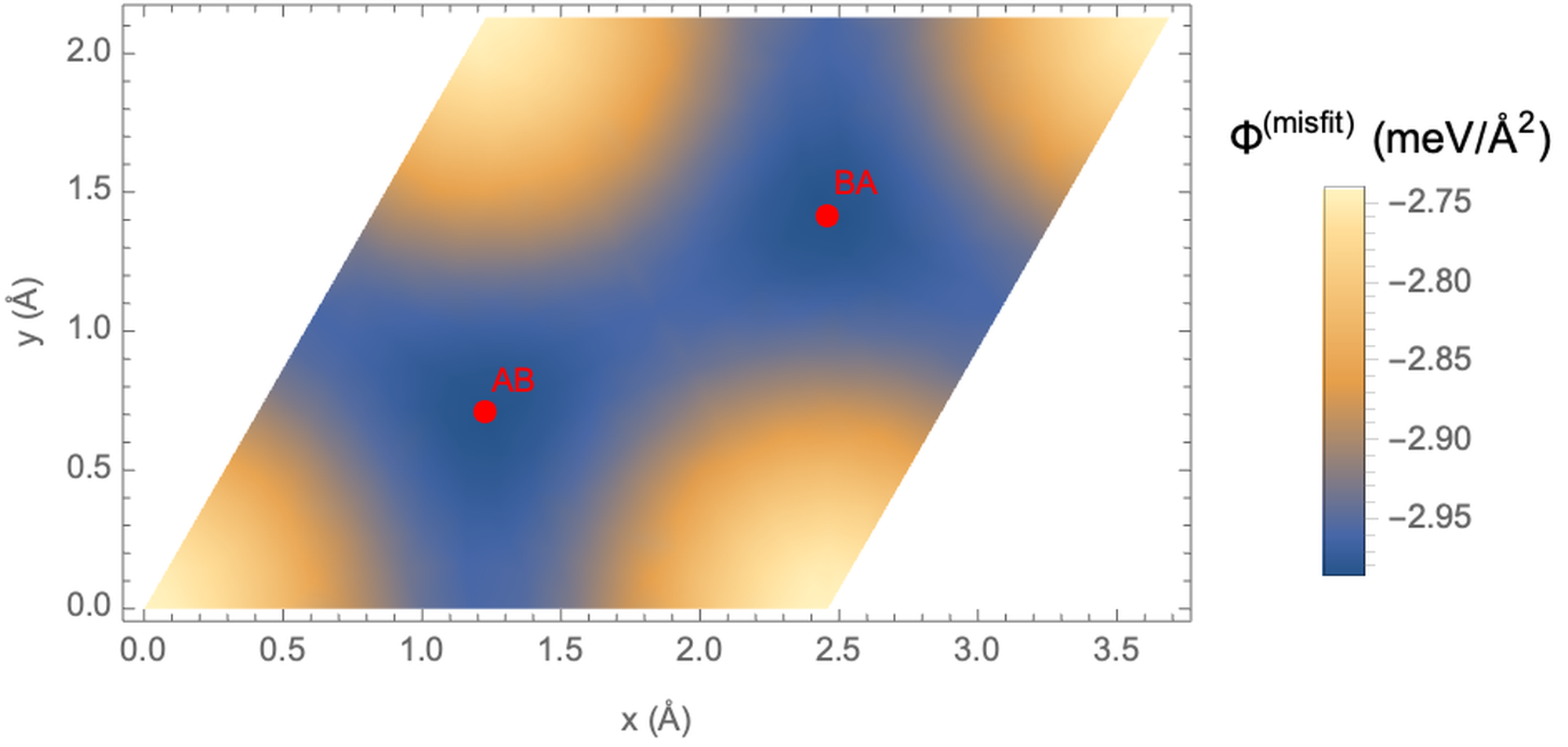}
    \caption{}
    \label{fig-gsfe-morse}    
    \end{subfigure}
    \begin{subfigure}[c]{0.4\textwidth}
        \includegraphics[width=\textwidth]{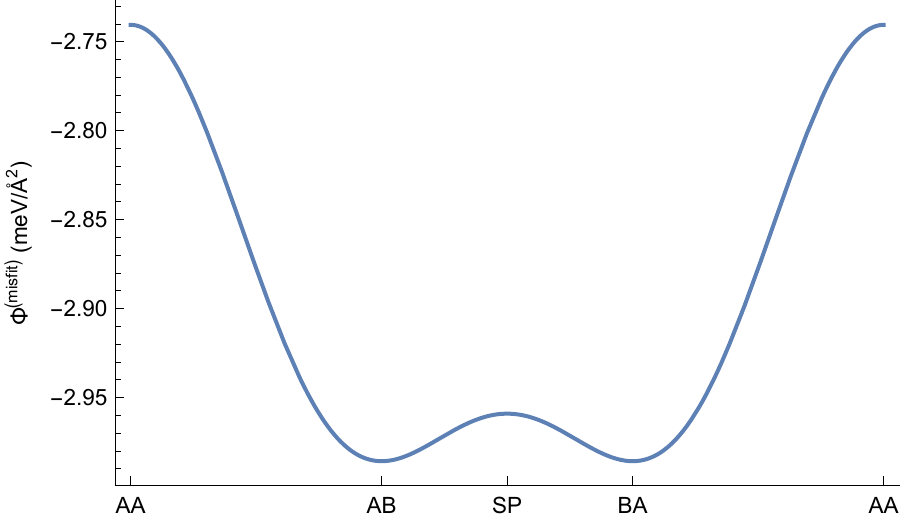}
    \caption{}
    \label{fig-gsfe-morse-diag}    
    \end{subfigure}
    \caption{Stacking fault energies for graphene-graphene interlayer potential. Red dots in left diagrams indicate AB and BA stacking, respectively. The corners correspond to AA stacking. On the right, the profile along the diagonal through the red dots is displayed. Fig. \ref{fig-gsfe-dft} and \ref{fig-gsfe-dft-diag} refer to the GSFE computed in \cite{relaxation-domain-wall-formation2018}. Fig. \ref{fig-gsfe-morse} and \ref{fig-gsfe-morse-diag} refer to misfit energy computed from Morse potential with parametrization $(E_0,\kappa_0,r_0)=(2.8437\mbox{\ meV},1.8168\mbox{\ \r{A}}^{-1},3.6891\mbox{\ \r{A}})$ taken from \cite{oconnor2015airebo}.}
    \label{fig-GSFE}
\end{figure}

Setting $\GSFE_j=\misfit_j$, we obtain, starting from the interlayer energy in \eqref{ARMA-energy},

\begin{align}
	\MoveEqLeft\frac1{|\gls{GamM}|}\int_{\gls{GamM}}\dx{x} \GSFE_j\big(\gls{disregj} x +u_j(x)-u_{3-j}(x)\big) \\
	&= \, \frac1{|\gls{GamM}||\gls{Gamj}|} \int_{\gls{GamM}}\dx{x} \sum_{R_{3-j}\in\cR_{3-j}}v\big(\gls{disregj} x-R_{3-j}+u_j(x)-u_{3-j}(x)\big) \, .
\end{align}
As above, we rescale $y:=\disregj x$, and define $U_j(y):=u_j((-1)^{j+1}\disregj^{-1}y)$. Then a straight-forward calculation using Lemma \ref{lem-disregj-calc} yields 
\begin{equation}
	\frac1{|\Gamma_1||\Gamma_2|}\int_{\Gamma_{3-j}}\dx{y} \sum_{R_{3-j}\in\cR_{3-j}}v\big( y-R_{3-j}+U_j(y)-U_{3-j}(A_jA_{3-j}^{-1}y)\big) \, .
\end{equation}
Putting technicalities aside, we interchange orders of integration and summation and substitute $x-R_{3-j}\to x$ and use $\cR_{3-j}$-periodicity of $U_j$ to obtain
\begin{equation}
	\frac1{|\Gamma_1||\Gamma_2|}\int_{\R^2}\dx{y}v\big( y+U_j(y)-U_{3-j}(A_jA_{3-j}^{-1}y)\big) \, .
\end{equation}
Including the monolayer contribution and rescaling the total energy by $\frac1\vep$, we thus obtain 
\begin{equation}\label{eq-GSFE-rescaled}
	\frac{\vep}2\sum_{j=1}^2\avint_{\Gamma_{3-j}}\dx{y} DU_j(y):\tilde{\cE}_j:DU_j(y) \, + \, \frac1\vep\int_{\R^2}\dx{y}V\big(y+U_1(y)-U_2(R_{-\theta}y)\big)  \, .
\end{equation}
Comparing \eqref{eq-GSFE-rescaled} and \eqref{eq-IAP-rescaled}, we notice that $U_2$ in the argument of $V$ is evaluated at different positions. The difference of arguments is of size $|x-R_{-\theta}x|\approx \vep|x|$ as $\vep\to0$. It was shown in \cite{cazeaux-clark-engelke-kim2022relaxation} that  for a minimizer $(U_1^{eq},U_2^{eq})$, $\|D U_j^{eq}\|_2 \lesssim \frac1\vep$. In particular, numerical computations show that domain walls of thickness $\vep$ form, where $U_j^{eq}$ has a transition of order $O(1)$ from AB to BA stacking. This phenomenon has been explained more rigorously in the context of the Allen-Cahn energy, see, e.g., \cite{chodosh2023lecture} for a nice introduction. Consequently, we expect our derived energy to vary significantly, of order $O(1)$ at the domain walls, from that obtained in \cite{Cazeaux-Massatt-Luskin-ARMA2020}. We plan to address this issue in future works.
 
\subsubsection{Outlook}

In the Cauchy-Born and linear elasticity approximation, we can employ ideas from elliptic theory to establish existence, uniqueness, and regularity of energy minimizers. This remains true as long as the interlayer potential is a compact perturbation to the intralayer Cauchy-Born energy. In particular, one can use a contraction map argument as in \cite{Cazeaux-Massatt-Luskin-ARMA2020}. We do not elaborate on further details here, and defer to future works instead. 
\par Furthermore, one can extend the intralayer model that we use beyond linear elasticity, as was studied, e.g., in \cite{malena2018,malena2023}. In order to include out-of-plane relaxation, we need to include an out-of-plane penalization for the monolayer contribution, see, e.g., \cite{peng2020strain} and example \ref{ex-tbg} above.
	
\par While we develop a model for bilayer systems only, we can extend the model to systems with an arbitrary finite number of layers, by summing over the pair-wise coupling of layers. For the case of many-body site potentials, we believe that the ideas presented in \cite{Cazeaux-Massatt-Luskin-ARMA2020} lend themselves to extending our results to the multi-layer case.  

\subsection{Results for rough displacements\label{sec-rough-displacements}} 

As stated above, it can be useful to access an energy functional when the displacement functions are not necessarily smooth. Due to the use of ergodicity, these results will only hold almost surely, up to some relative lattice shifts. We will see that to even formulate the problem for rougher conditions on the displacement functions, we need to introduce a weak notion of lattice dependence for the displacement functions. In contrast, above the existence of (moir\'e-)periodic continuous functions attaining assigned values on the lattices yields a strong constraint on the attainable values. We now introduce a weak lattice dependence, that we refer to as \emph{reconstructability}, which will allow us to relax this constraint. Reconstructability in layer $j$ means that there exists a single function $u_j$ such that, for almost all lattice origins $\gamma_j$, the values of $u_j$ assigned to the corresponding shifted lattice $\gamma_j+ \cR_j$, allow us to recover all Fourier coefficients of $u_j$. Consequently, our statements will only hold up to a nullset of lattice shifts $\gamma_1,\gamma_2$.

\par We start by stating an almost sure and mean ergodic theorem that will allow us to rigorously analyze thermodynamic limits of elastic energies that are functions of non-smooth displacements $u_j\in L^1\big(\gls{GamM};(\R^3)^{\gls{cAj}}\big).$

\begin{proposition}[Ergodic Theorem]\label{prop-single-erg-thm}
For all $f\in L^1(\gls{GamM})$ 
\begin{equation}
	\lim_{N\to\infty} \frac1{(2N+1)^2 }\sum_{R_j\in\gls{cRjN}}f(\mfrac{R_j+\gamma_j}) \, = \, \mavint \dx{x} f(x) \, .
\end{equation} 
converges almost everywhere in $\gls{Gamj}$ and also in the $L^1(\gls{Gamj})$-sense.
\end{proposition}
Proposition \ref{prop-single-erg-thm} follows from a more general ergodic theorem stated here as Proposition \ref{prop-pw-erg-thm-gen}. Our first use of Proposition \ref{prop-single-erg-thm} will be to prove in Proposition \ref{prop-reconstr} that the Fourier transform of $u_j\in L^1\big(\gls{GamM};(\R^3)^{\gls{cAj}}\big)$ can be reconstructed as a thermodynamic limit from $u_j(\mfrac{\gamma_j+\gls{cRj}},\gls{cAj})$ for almost every $\gamma_j\in\gls{Gamj}$.

We now define lattice deformations based on $u_j\in L^1\big(\gls{GamM};(\R^3)^{\gls{cAj}}\big).$
\begin{assumption}\label{ass-pos-labelling}
There exist displacements $u_j\in L^1\big(\gls{GamM};(\R^3)^{\gls{cAj}}\big)$ such that the total positions on the reference collection of lattices 
\begin{align}
	\big((\gamma_j+\gls{cRj})\times \gls{cAj}\big)_{\gamma_j\in\gls{Gamj}} \, , \quad j\in\{1,2\} \, , 
\end{align}
are given by the deformation \eqref{eq-lat-pos}.
\end{assumption}

Assumption \ref{ass-pos-labelling} implies that whenever $\mfrac{\gamma_j+R_j}=\mfrac{\gamma_j'+R_j'}$, we have that the associated displacements
\begin{equation}
u_j(\mfrac{\gamma_j+R_j},\alpha_j)=u_j(\mfrac{\gamma_j'+R_j'},\alpha_j)
\end{equation}
coincide. Since $(\cR_1^*,\cR_2^*)$ is incommensurate, this can only happen for $\gamma_j\neq\gamma_j'$, ensuring that we do not impose a moir\'e-periodicity constraint. 

\par Observe that the total lattice positions are well-defined only up to a nullset of lattice shifts $\gamma_j\in\gls{Gamj}$.  Next we will discuss in which sense the data of $u_j$ on a single lattice suffices to reconstruct $u_j$.

\begin{definition}\label{defi-reconst}
We say that $u_j\in L^1\big(\gls{GamM};(\R^3)^{\gls{cAj}}\big)$ is \emph{reconstructable} along $\gamma_j+\gls{cRj}$ iff for all $\MoireG\in\gls{cRm*}$, $\alpha_j\in\gls{cAj}$, we have that
\begin{align}
	\begin{split}
		\lim_{N\to\infty} \frac1{(2N+1)^2}\sum_{R_j\in\gls{cRjN}} e^{-i\MoireG\cdot (\gamma_j+R_j)} u_j(\mfrac{\gamma_j+R_j},\alpha_j) \, = \, \hat{u}_j(\MoireG,\alpha_j) \, .
	\end{split}
\end{align}
\end{definition}
In particular, if $u_j\in L^p\big(\gls{GamM};(\R^3)^{\gls{cAj}}\big)$, $p>1$, is reconstructable along $\gamma_j+\gls{cRj}$, the Carleson-Hunt Theorem, see, e.g., \cite[Chapter 1]{duoandikoetxea2001fourier}, implies that for almost all $x\in\gls{GamM}$ and all $\alpha_j\in\gls{cAj}$
\begin{align}
\MoveEqLeft u_j(x,\alpha_j) \, = \, \\
& \sum_{\MoireG\in\gls{cRm*}}e^{i\MoireG\cdot x} \lim_{N\to\infty} \frac1{(2N+1)^2}\sum_{R_j\in\gls{cRjN}} e^{-i\MoireG\cdot (\gamma_j+R_j)} u_j(\mfrac{\gamma_j+R_j},\alpha_j) \, .
\end{align}
\par If $u_j,v_j\in L^1\big(\gls{GamM};(\R^3)^{\gls{cAj}}\big)$, we have the uniqueness property that
\begin{equation}
\hat{u}_j(\MoireG,\alpha_j) \, = \, \hat{v}_j(\MoireG,\alpha_j)    
\end{equation}
for all $(\MoireG,\alpha_j)\in\gls{cRm*}\times\gls{cAj}$ implies that $u_j=v_j$, see \cite[Corollary 1.11]{duoandikoetxea2001fourier}.

\par As a consequence of Proposition \ref{prop-single-erg-thm}, we find the following result.

\begin{proposition}[Almost sure reconstructability]\label{prop-reconstr}
Every $u_j\in L^1\big(\gls{GamM};(\R^3)^{\gls{cAj}}\big)$ is reconstructable along $\gamma_j+\gls{cRj}$ for almost every $\gamma_j\in\gls{Gamj}$. 
\end{proposition}
For a proof, we refer to Appendix \ref{app-latt-calc}.

We are now ready to formulate the main result for rough displacements.

\begin{theorem}\label{thm-rough-conv}
	Assume that $\gls{SPmonoj} [u_j]\in L^1(\gls{GamM})$, and that $
		\gls{SPinterj} [\gls{bu}]\in L^1(\gls{GamM}\times\Gamma_{3-j})$ and for any $\gls{bgam}=(\gamma_1,\gamma_2)\in\Gamma_1\times\Gamma_2$ that $\gls{SPinterjdiag} [\gls{bu}]\in L^1(\gls{GamM})$, respectively. 
        \par In the case of interlayer pair potentials, let $r>1$ and assume that, for $k=1,2$, $\alpha_k\in\cA_k$, $u_k(\cdot,\alpha_k)\in L^{\infty}(\R^2)$, and that
	    \begin{align} 
		\gls{pairpot} \in L^{\infty}_{2r}(\R^2\times[-L_{\gls{bu}}^z,L_{\gls{bu}}^z]) \, .
	    \end{align}
        We then have for almost all $\gls{bgam}\in\Gamma_1\times \Gamma_2$ that 
	\begin{align} 
        \lim_{N\to\infty}\mono_{j,N,\gamma_j}(u_j) \, = \, \mono_{j}(u_j) \, , \quad \lim_{N\to\infty} \inter_{j,N,\gls{bgam}}(\gls{bu}) \, = \,  \inter_{j,\gls{bgam}}(\gls{bu}) \, .
	\end{align}.
\end{theorem}
In the case of many-body potentials, this result is a direct consequence of the ergodic theorem Proposition \ref{prop-single-erg-thm}. The proof that the conditions in the case of interlayer pair potentials are sufficient, are contained in Section \ref{sec-proofs-main-thm}.

\subsection{Sketch of the proof}

The technical tools used for the thermodynamic limit stem from ergodic theory. We prove an almost sure and mean ergodic theorem, see Proposition \ref{prop-double-erg-thm}, and a quantitative ergodic theorem, see Proposition \ref{prop-ergod-thm}. The latter is inspired by \cite{cazeaux-luskin-CB-2017} which shows uniform convergence with respect to shifts of the lattice origin, for the one-dimensional case. In contrast to \cite{cazeaux-luskin-CB-2017}, our approach is two-dimensional and based on studying the convergence properties of the approximate Dirichlet kernel associated with the ergodic average over the truncated lattice. We directly compute an $L^1$-limit for smooth functions, which, by density, yields a general $L^1$-limit. In order to obtain convergence almost everywhere, we invoke a general ergodic theorem \ref{prop-ergod-thm} as stated, e.g., in \cite{KerrLi-erg-th}. For a more detailed description of our approach, see Section \ref{sec-erg-thms}.

\section{Ergodic theorems\label{sec-erg-thms}}

We start by stating ergodic theorems in a sufficiently general setting, as developed, e.g., in \cite{KerrLi-erg-th,kechris2010global}, see also \cite{Damanik2022Oerg-schroedinger-book}. Then we explain how these general ideas can be applied in the present setting.

\subsection{Review of general ergodic theorems} 

Given a probability space $(\Omega,\Sigma,\P)$ and a group $G$, recall that a group action $G\curvearrowright \Omega$ is \emph{measure-preserving} iff
\begin{equation}
	g.\P(B) \, := \, \P(g^{-1}.B) \, = \, \P(B)
\end{equation}
for all $B\in\Sigma$ and $g\in G$. Moreover, we call a measurable map $f:\Omega\to\R$ $G$-\emph{invariant} iff $g.f:=f\circ g^{-1}=f$ for any $g\in G$ $\P-$almost surely. Similarly, we call measurable set $B$ $G$-invariant iff $\mathds{1}_B$ is $G$-invariant. Let $\Sigma_G$ denote the sub $\sigma$-algebra of $\Sigma$ of all $G$-invariant sets in $\Sigma$.
\par Given a measurable function $f:\Omega\to\R$ and a finite subset $H\subseteq G$, we define the \emph{ergodic average}
\begin{equation}
	\ergav_H(f) \, := \, \frac1{|H|}\sum_{g\in H} g^{-1}.f \, . \label{def-ergav} 
\end{equation}
A sequence $(H_N)_{N\in\N}$ of non-empty finite subsets of $G$ is called a \emph{F\o lner sequence} iff 
\begin{equation}
	\lim_{N\to\infty}\frac{|gH_N\symdiff H_N|}{|H_N|} \, = \, 0 
\end{equation}
for all $g\in G$, where $A\symdiff B$ denotes the symmetric difference of sets $A,B$. $(H_N)_{N\in\N}$ is \emph{tempered} iff there is $C>0$ s.t.
\begin{equation}
	\Big|\bigcup_{k=1}^{N-1}H_k^{-1}H_N\Big| \, \leq\, C|H_N|
\end{equation}
for all $N\in\N$, $N\geq2$.

\par We finish this general introduction with the following pointwise ergodic theorem as stated, e.g., in \cite[Theorem 4.28]{KerrLi-erg-th}.

\begin{proposition}[Pointwise Ergodic Theorem]\label{prop-pw-erg-thm-gen}
	Let $(\Omega,\Sigma,\P)$ be a probability space and $G\curvearrowright \Omega$ be a measure-preserving group action. Let $(H_N)_{N\in\N}$ be a tempered F\o lner sequence. Then we have that for all $f\in L^1(\Omega)$ and 
	\begin{equation}
		\lim_{N\to\infty}\ergav_{H_N}(f) \, = \, \E[f|\Sigma_G]
	\end{equation}
	$\P-$almost surely and in the $L^1$-sense, where $\E[\cdot|\Sigma_G]$ denotes the conditional expectation given the $\sigma$-algebra $\Sigma_G$ of $G$-invariant events. 
\end{proposition}

\begin{remark}
	Recall that the conditional expectation $\E[f|\Sigma_G]$ is the unique $(\Omega,\Sigma_G)$-measurable function such that
	\begin{equation}
		\int_C\dx{\P(\omega)} \E[f|\Sigma_G](\omega) \, = \, \int_C \dx{\P(\omega)} f(\omega)
	\end{equation}
	for all ($G$-invariant) sets $C\in \Sigma_G$.
\end{remark}

\subsection{Ergodicity in incommensurate bilayer systems} 

\par We now turn to the system at hand and explore how ergodicity arises as a consequence of the underlying incommensurate geometry. To apply the general ergodic setting above, we set $G=\gls{cRj}$ and $H_N:=\gls{cRjN}$, see \eqref{def-Rj-trunc}. It is straight-forward to verify that $(\gls{cRjN})_{N\in\N}$ is both F\o lner and tempered. We leave the details to the interested reader.
\par Let $\T^2=\R^2/\Z^2$ denote the standard 2-torus, and
\begin{align} 
	\MoireTorus \, := \, \gls{Am} \T^2 \, , \quad \OtherTorus \, := \, A_{3-j}\T^2 \label{def-latticetori}
\end{align}
denote the tori associated with the moir\'e unit cell $\gls{GamM}$ and with the unit cell $\Gamma_{3-j}$ of layer $3-j$, respectively. Let $\mathcal{B}$ denote the respective Borel $\sigma$-algebra over $\Omega$.

\par Then we choose $\Omega=\MoireTorus\times\OtherTorus$ and study the Lebesgue-measure preserving group action of $\gls{cRj}$ on $\Omega$ given by translations, i.e., $R_j.\omega:=\omega+(R_j,R_j)$. As a consequence of Proposition \ref{prop-pw-erg-thm-gen}, we prove the following statement.

\begin{proposition}[Pointwise and Mean Ergodic Theorem]\label{prop-double-erg-thm}
	Let $h\in L^1_{\emph{loc}}(\R^4)$ be $\gls{cRm}\times\cR_{3-j}$-periodic. Then
 
	\begin{align}
		\begin{split}
			\lim_{N\to\infty}\ergav_{\gls{cRjN}}(h)(\omega_\cM,\omega_{3-j})\, &= \, \mavintrescj \dx{\xi} h\big(A_j\xi+\omega_\cM,(A_j-A_{3-j})\xi+\omega_{3-j}\big) \\
			&= \, \mavint \dx{x} h\big(x,\disregj(x-\omega_\cM)+\omega_{3-j}\big) \\
			&= \, \avint_{\Gamma_{3-j}} \dx{y} h\big(\disregj^{-1}(y-\omega_{3-j})+\omega_\cM,y\big)
		\end{split}
	\end{align}
	converges for almost all $(\omega_\cM,\omega_{3-j})\in\gls{GamM}\times\Gamma_{3-j}$ and also in $L^1(\gls{GamM}\times\Gamma_{3-j})$.
\end{proposition}

\begin{remark}
	Let $h\in L^1_{\emph{loc}}(\MoireTorus\times\OtherTorus)$. In addition, let $\mathcal{B}_{\gls{cRj}}$ denote the $\sigma$-algebra of $\gls{cRj}$-invariant Borel subsets of $\MoireTorus\times\OtherTorus$, where the $\gls{cRj}$-action is given as above by $R_j.\omega=\omega+(R_j,R_j)$. Then one can prove directly that
	\begin{align}
		\E[h|\Sigma_{\gls{cRj}}](\omega_\cM,\omega_{3-j})\, = \, \mavintrescj \dx{\xi} h\big(A_j\xi+\omega_\cM,(A_j-A_{3-j})\xi+\omega_{3-j}\big) \, .
	\end{align}
	Lebesgue-almost everywhere, where the conditional expectation on the l.h.s. is taken with respect to the uniform distribution on $\MoireTorus\times\OtherTorus$.
	\par We find the argument presented in the proof of Proposition \ref{prop-double-erg-thm} appealing because it provides a means to calculate the conditional expectation.
\end{remark}

Using the notion of Diophantine 2D rotations, see Definition \ref{defi-diophantine}, we obtain the following quantitative ergodic theorem.

\begin{proposition}[Quantitative Ergodic Theorem]\label{prop-ergod-thm}
	Let $q>0$, $\sigma>\frac{1433}{1248}$, $s>1$, and let $\theta$ be a $(q,K,\sigma)$-Diophantine 2D rotation for some $K>0$. Let $f\in L^1(\gls{GamM})$ be such that $\gls{FTMoire}\big(|\nabla|^{2(\sigma+s)}f\big)\in \ell^\infty(\gls{cRm*})$. Then we have for all $N\in\N$ large enough that 
	\begin{align}\label{eq-sum-trafo-approx-quant}
		\begin{split}
			\MoveEqLeft\bigg|\frac1{\big(2N+1\big)^2 }\sum_{R_j\in\gls{cRjN}}f(\mfrac{R_j}) \, - \, \mavint \dx{x} f(x)\bigg|\\
			\leq & \, \frac{2\sqrt{2}}{K}\Big(\frac{\gls{moirelen} \|A\|_2}{2\pi}\Big)^{2(\sigma+s)}\big(\zeta(2s)+2^{-s}\zeta(s)^2\big) \\
			& \, \sup_{\MoireG\in\gls{cRm*}}|\MoireG|^{2(\sigma+s)} |\hat{f}(\MoireG)| \frac1{2N+1}\, . 
		\end{split}
	\end{align}
\end{proposition}
For proofs of these ergodic theorems, we refer to Appendix \ref{app-erg-thm-proof}.

\begin{remark}
	Our methods would allow us to adapt Proposition \ref{prop-ergod-thm} to the more general setting in Proposition \ref{prop-double-erg-thm}. However, for presentational purposes, we chose to omit such a result and leave the details to the interested reader.
\end{remark}  

We now want to motivate how to prove Propositions \ref{prop-double-erg-thm} and \ref{prop-ergod-thm}.

\subsection{Dirichlet kernel as a discrepancy function\label{sec-dirichlet}}

Our idea is based on using periodicity of the functions which we are averaging over the individual lattices by representing them via their Fourier expansion. Averaging the corresponding Fourier basis $(\exp(iG\cdot(\cdot)))_G$ over a discrete set leads to the emergence of the Dirichlet kernel, normalized by the number of lattice points. This is a common idea in Fourier analysis, see, e.g., \cite{duoandikoetxea2001fourier}, and has been applied, e.g., in \cite{Trefethen2014,chenhott2023} to derive a discretization error in Riemann sums.

Let $h\in L^p(\gls{GamM}\times\Gamma_{3-j})$, $p>1$. In addition to the Fourier transform defined in \eqref{def-FT}, for $h\in L^1(\gls{GamM}\times\gls{Gamj})$ and $(\MoireG,G_j)\in(\gls{cRm*},\gls{cRj*})$, define 
\begin{align}\label{def-double-FT}
	\begin{split}
		\MoveEqLeft \hat{h}(\MoireG,G_j) \, := \,  \cF_{\gls{GamM}\times\gls{Gamj}}(h)(\MoireG) \, := \\
		& \avint_{\gls{Gamj}} \dx{y} \mavint \dx{x} e^{-i(\MoireG\cdot x+G_j\cdot y)} h(x,y) \, .    
	\end{split}
\end{align}
Here, we abuse notation and interpret $\hat{h}$ in the local context.
\par Using this Fourier transform, we obtain that for almost all $(\omega_\cM,\omega_{3-j})\in\MoireTorus\times\OtherTorus$
\begin{align}
	\begin{split}
		\MoveEqLeft\frac1{(2N+1)^2}\sum_{R_j\in\gls{cRjN}}h(\mfrac{R_j+\omega_\cM},\Lfrac{R_j+\omega_{3-j}}{3-j}) \, = \\
		& \sum_{R_j\in\gls{cRjN}}\sum_{\substack{G_{3-j}\in\cR_{3-j}^*,\\\MoireG\in\gls{cRm*}}} e^{i(\MoireG\cdot\mfrac{R_j+\omega_\cM}+G_{3-j}\cdot\Lfrac{R_j+\omega_{3-j}}{3-j})} \hat{h}(\MoireG,G_{3-j}) \, .
	\end{split}
\end{align}
Realizing that $\gls{cRm*}\cdot \mfloor{x}\subseteq 2\pi\Z$ for any $x\in\R^2$, we have that
\begin{equation} \label{eq-mfrac-exp}
	e^{i\MoireG\cdot\gls{mfracx}} \, = \, e^{i\MoireG\cdot x}
\end{equation}
for any $x\in\R^2$. A straight-forward calculation analogous to the computation of the Dirichlet kernel yields
\begin{equation} \label{def-appdelNi}
	\begin{split}
		\gls{Dirichlet}(G) \, &:= \, \frac1{(2N+1)^2}\sum_{R_j\in\gls{cRjN}} e^{iG\cdot R_j} \\
		&= \, \prod_{\ell=1}^2 \frac{\sin\big((2N+1)(A_j^T G)_\ell/2 \big)}{(2N+1)\sin\big((A_j^T G)_\ell/2 \big)} \, ,
	\end{split}
\end{equation}
where $y_\ell$ denotes the $\ell^{\text{th}}$ Euclidean coordinate of $y$. Using this notation, we obtain that
\begin{align}\label{eq-sum-double-trafo} 
	\begin{split}
		\MoveEqLeft\frac1{(2N+1)^2}\sum_{R_j\in\gls{cRjN}}h(\mfrac{R_j+\omega_\cM},\Lfrac{R_j+\omega_{3-j}}{3-j}) \, = \\
		& \sum_{\substack{G_{3-j}\in\cR_{3-j}^*,\\\MoireG\in\gls{cRm*}}}  \gls{Dirichlet}(G_{3-j}+\MoireG) \,  \hat{h}(\MoireG,G_{3-j})  e^{i(\MoireG\cdot\omega_\cM+G_{3-j}\cdot\omega_{3-j})} \, . 
	\end{split}
\end{align}
If $f\in L^p(\gls{GamM})$, analogous steps yield
\begin{align}\label{eq-sum-trafo} 
	\frac1{(2N+1)^2}\sum_{R_j\in\gls{cRjN}}f(\mfrac{R_j}) \, =& \, \sum_{\MoireG\in\gls{cRm*}} \gls{Dirichlet}(\MoireG) \, \hat{f}(\MoireG) \, . 
\end{align}
Next, we will study the convergence properties of $\gls{Dirichlet}$.

\paragraph{Properties of the Dirichlet kernel}

For all $G_j\in\gls{cRj*}$, we have that
\begin{align}\label{eq-appdelN-const-on-Rj*}
	\gls{Dirichlet}(G_j) \, &= \, 1 \, .
\end{align}
Moreover, for all $G\in\R^2\setminus\gls{cRj*}$, it holds that
\begin{equation}\label{eq-appdelNi-pw-bd}
	|\gls{Dirichlet}(G)| \, \leq \, \frac1{2N+1}\min_{\ell=1,2} \frac1{\big|\sin\big((A_j^T G)_\ell/2\big)\big|} \, .
\end{equation}
In particular, we obtain that for all $G\in\R^2$
\begin{equation} \label{eq-appdelN-limit}
	\lim_{N\to\infty}\gls{Dirichlet}(G) \, = \, \mathds{1}_{\gls{cRj*}}(G) \, .
\end{equation}
In particular, if we can interchange the order of taking the limit $N\to\infty$ and summing over $G_{3-j}$ and $\MoireG$, \eqref{eq-sum-double-trafo} implies that
\begin{align}
	\begin{split}
		\MoveEqLeft\frac1{(2N+1)^2}\sum_{R_j\in\gls{cRjN}}h(\mfrac{R_j+\omega_\cM},\Lfrac{R_j+\omega_{3-j}}{3-j}) \, = \\
		& \xrightarrow[]{N\to\infty}\sum_{\substack{G_{3-j}\in\cR_{3-j}^*,\\\MoireG\in\gls{cRm*}}}  \mathds{1}_{\gls{cRj*}}(G_{3-j}+\MoireG) \,  \hat{h}(\MoireG,G_{3-j})  e^{i(\MoireG\cdot\omega_\cM+G_{3-j}\cdot\omega_{3-j})} \, . 
	\end{split}
\end{align}
Again, assuming interchangeability of involved sums and integrals, we arrive, after a straight-forward calculation, at 
\begin{equation}
	\mavintrescj \dx{x} h\big(A_jx+\omega_\cM,(A_j-A_{3-j})x+\omega_{3-j}\big) \, ,
\end{equation}
as stated in Proposition \ref{prop-double-erg-thm}. In order to interchange limits, sums, and integrals, we use the fact that we study the convergence in $L^1(\MoireTorus\times\OtherTorus)$, in order to restrict to $C^\infty$-smooth functions.

\par In order to extract a rate of convergence in \eqref{eq-sum-double-trafo}, we need to analyze the tail-behavior of $\gls{Dirichlet}$. For general irrational angles, the rate of convergence can be arbitrarily slow. However, we show in the proof of Proposition \ref{prop-ergod-thm}, in Appendix \ref{app-erg-thm-proof}, that for Diophantine 2D rotations $\theta$, see Definition \ref{defi-diophantine}, and all $\MoireG\in\gls{cRm*}\setminus\{0\}$, we have that
\begin{equation}\label{eq-delta-bd}
	|\gls{Dirichlet}(\MoireG)| \, \leq \, \frac{C_{\sigma,q,A,\theta}|\MoireG|^{2\sigma}}{2N+1}
\end{equation}
for any $\sigma>\frac{1433}{1248}$, where $C_{\sigma,q,A,\theta}$ depends only on $\sigma$, $q$, $A$, and $\theta$. In particular, employing \eqref{eq-delta-bd}, \eqref{eq-sum-trafo} implies that
\begin{align}\label{eq-sum-trafo-approx-0}
	\begin{split}
		\lefteqn{\bigg|\frac1{(2N+1)^2}\sum_{R_j\in\gls{cRjN}}f(\mfrac{R_j}) \, - \, \mavint \dx{x} f(x)\bigg|}\\
		&\leq \, C_{\sigma,q,A,\theta}\sum_{\MoireG\in\gls{cRm*}\setminus\{0\}}  |\MoireG|^{2\sigma} |\hat{f}(\MoireG)| \frac1{2N+1}\, . 
	\end{split}
\end{align}
Lemma \ref{lem-sum-1/Gs} then implies for any $s>1$ that
\begin{align}\label{eq-sum-trafo-approx}
	\begin{split}
		\lefteqn{\bigg|\frac1{(2N+1)^2}\sum_{R_j\in\gls{cRjN}}f(\mfrac{R_j}) \, - \, \mavint \dx{x} f(x)\bigg|}\\
		&\leq C_{\sigma,q,A,\theta} \, 4\Big(\frac{\gls{moirelen} \|A\|_2}{2\pi}\Big)^{2s}\big(\zeta(2s)+2^{-s}\zeta(s)^2\big)\\
		& \, \sup_{\MoireG\in\gls{cRm*}}|\MoireG|^{2(\sigma+s)} |\hat{f}(\MoireG)|\frac1{2N+1} \, , 
	\end{split}
\end{align}
where the constant $C_{\sigma,q,A,\theta}$ is the same as in \eqref{eq-delta-bd}.

\section{Proofs of results\label{sec-proofs}}

We now want to show how the ergodic theorems obtained in the previous section can be used to show the existence of the thermodynamic limits of the mechanical energy densities presented above. To that end, we show that the energy densities for lattice truncations can be represented as averages over appropriate periodic functions. In the case of interlayer pair potentials, we show how to represent the resulting limit energy density in a more convenient way.

\subsection{Equivalence of disregistries}

Recall the disregistry matrices
\begin{align} \label{id-disregj}
	\gls{disregj} \, &= \, I \, - \, A_{3-j}A_j^{-1} \, , 
\end{align}
see \eqref{eq-disregj-moire-per} and \eqref{def-Aj} above.

\begin{lemma}\label{lem-disregj-calc} The disregistry matrices have the following properties.
	\begin{enumerate}[label=\textnormal{(\arabic*)}]
		\item $
		\gls{disregj} \, = \, (-1)^jA_{3-j}\MoireRLV^{-1}$, \label{itm-disregj}
		\item $
		-\disregjrev\disregj^{-1} = -\disregj^{-1}\disregjrev = A_jA_{3-j}^{-1} = I-\disregjrev
		$, \label{itm-dis-disinv}
		\item $|\disregj^{-1}x|=q^{(-1)^{j+1}1/2}\moirelen^{-1} |x|$, \label{itm-disinv-norm}
		\item $|(\disregj^{-1}-I)x|=q^{(-1)^j1/2}\moirelen^{-1} |x|$. \label{itm-disinv-1-norm}
	\end{enumerate}
\end{lemma}
\begin{proof}
	We show the computations for $j=1$; $j=2$ can be computed analogously. We have proved \ref{itm-disregj} already above in \eqref{eq-disregj-moire-per}. For \ref{itm-dis-disinv}, we employ \ref{itm-disregj} and obtain
	\begin{align}
		-\bto\bot^{-1} \, = \, A_1\MoireRLV^{-1}\gls{Am} A_2^{-1} \, = \, A_1A_2^{-1}\, = \, I-\bto \, ,
	\end{align}
	where, in the last step, we employ \eqref{eq-disregj-moire-per}. Similarly, using \ref{itm-disregj} again, we have that
	\begin{align}
		\begin{split}
			-\bot^{-1}\bto \, &= \, \gls{Am} A_2^{-1} A_1\MoireRLV^{-1}\, = \, \gls{Am} (A_2^{-1} A_1-I)\MoireRLV^{-1} \, + \, I \\
			&= \, - \gls{Am} \MoireRLV^{-1}A_1\MoireRLV^{-1} \, + \, I \, = \, A_1 A_2^{-1} \, .
		\end{split}
	\end{align}
	Observe that, due to \eqref{def-Aj}, $|A_2x|=q^{\frac12}|Ax|$. By \ref{itm-disregj}, we thus have that
	\begin{align}\label{eq-disregj-norm}
		|\bot^{-1}x| \, = \, |A_2\MoireRLV^{-1}x| \, = \, q^{1/2}|A\MoireRLV^{-1}x| \, = \, q^{\frac12}\moirelen^{-1} |x| \, ,
	\end{align}
	where, in the last step, we employed Lemma \ref{lem-moire-len}. By \ref{itm-dis-disinv}, we have that
	\begin{equation}
		\bot^{-1}-I \, = \, \bot^{-1}(I-\bot) \, = \, -\bto^{-1} \, ,
	\end{equation}
	which is why \ref{itm-disinv-1-norm} follows from \ref{itm-disinv-norm}. This concludes the proof.
\end{proof}
Lemma \ref{lem-disregj-calc} establishes a relation between the different disregistry notions. Next, we will specify this relation. Define 
\begin{equation}\label{def-jcor}
	\jcor \, := \, A_j\begin{pmatrix}
		1 \\ 1
	\end{pmatrix} \, .
\end{equation}

\begin{lemma}\label{lem-disreg-transf}
	Let $R_j\in \gls{cRj}$, $j\in\{1,2\}$. Then we have that
	\begin{align}
		\Lfrac{R_1}{2} \, & 
		= \, \bot \mfrac{R_1} \, + \, \tcor \, ,
		\\
		\Lfrac{R_2}{1} \, & 
		= \, \bto \mfrac{R_2} \, .
	\end{align}
\end{lemma}
\begin{proof}
	Observe that $A_2A_k^{-1}R_k\in \cR_2$ for $R_k\in \cR_k$, $k\in\{1,\cM\}$. Hence we find that
	\begin{align}
		\begin{split}
			\Lfrac{R_1}{2} \, &= \, \Lfrac{\bot R_1}{2} \, = \, \Lfrac{-A_2\MoireRLV^{-1}R_1}{2} \\
			&= \, \Lfrac{-A_2\MoireRLV^{-1}\mfrac{R_1}}{2} \, = \, \bot \mfrac{R_1}+ \tcor \in \Gamma_2 \, ,
		\end{split}
	\end{align}
	where, in the last step, we used the fact that $\tcor\in \cR_2$. The second identity follows analogously. 
	This concludes the proof.
\end{proof}

\begin{remark}\label{rem-gen-disreg-trafo-nonexist}
	Lemma \ref{lem-disreg-transf} is sensitive to the lattice origin. More precisely, there is no affine map between $\Lfrac{x}{3-j}$ and $\gls{mfracx}$ for general $x\in\R^2$. In fact, assume that there is a matrix $M$ such that for all $x\in\R^2$,
	\begin{equation}
		x \, = \, M(x+\gls{cRm}) \, + \, \cR_{3-j} \, .
	\end{equation}
	In particular, this requires $M\gls{cRm} \subseteq \cR_{3-j}$, which implies 
	\begin{equation}
		M\in \MoireRLV^{-1}A_{3-j} \Z^{2\times 2} \, .    
	\end{equation}
	However, e.g., $(I-M)\sqrt{2}v_{\Gamma_{3-j}}\not\in\cR_{3-j}$. Thus, no such $M$ exists.
\end{remark}

\subsection{Lattice reduction formulae\label{sec-latt-red-form}}

Recall definitions \eqref{def-sitepotdouble} of $\gls{SPinterj} [\gls{bu}]$ and \eqref{def-sitepot} of $\gls{SPinterjdiag} [\gls{bu}]$.

\begin{lemma}\label{lem-site-pot}
	For all $x\in\R^2$ and all $\gamma_{3-j}\in\Gamma_{3-j}$, we have that
	\begin{align} \label{eq-site-pot-gen}
		\begin{split}
			\MoveEqLeft \frac1{|\gls{Gamj}|}\sum_{\alpha_j\in\gls{cAj}}\gls{Vinterj}\bigg(\Big(\gls{Yjx}\, -\, Y_{3-j}(R_{3-j}+\gamma_{3-j},\alpha_{3-j})\Big)_{\substack{R_{3-j}\in\cR_{3-j},\\\alpha_{3-j}\in\cA_{3-j}}}\bigg) \\
			&= \,  \gls{SPinterj} [\gls{bu}](\gls{mfracx},\Lfrac{x-\gamma_{3-j}}{3-j}) \, .
		\end{split}
	\end{align}
	In the case $x=R_j\in\gls{cRj}$ and $\gamma_{3-j}=0$, we even have
	\begin{align}
		\begin{split}
			\MoveEqLeft \frac1{|\gls{Gamj}|}\sum_{\alpha_j\in\gls{cAj}}\gls{Vinterj}\bigg(\Big(Y_j(R_j,\alpha_j)\, -\, Y_{3-j}(R_{3-j},\alpha_{3-j})\Big)_{\substack{R_{3-j}\in\cR_{3-j},\\\alpha_{3-j}\in\cA_{3-j}}}\bigg) \\
			& = \, \sitepot_{j,0}[\gls{bu}](\mfrac{R_j}) \, .
		\end{split}
	\end{align}
\end{lemma}
\begin{proof}
	By translation-invariance of $\gls{Vinterj}$, we may shift the index 
	\begin{equation}
		R_{3-j}\to R_{3-j} + \Lfloor{x-\gamma_{3-j}}{3-j}.
	\end{equation}
	Employing $\gls{cRm}$-periodicity of $u_j$ and definition \eqref{def-Lfrac-Lfloor} of $\Lfloor{x}{3-j}$, the first claim then follows from definition \eqref{def-sitepotdouble}.
	\par In the case $x=R_j\in\gls{cRj}$ and $\gamma_{3-j}=0$, recall from Lemma \ref{lem-disreg-transf} that
	\begin{equation}
		\Lfrac{R_j}{3-j} \, = \, \gls{disregj}\mfrac{R_j} \, + \, \tcor\delta_{j,1} \, .
	\end{equation}
	Then the statement follows from \eqref{def-sitepot} after shifting $R_{3-j}\to R_{3-j}+\tcor\delta_{j,1}$ and using translation-invariance. This concludes the proof.
\end{proof}

\vspace*{2ex}

Define for all $y\in\R^2$
\begin{align}\label{def-galpha}
	\begin{split}
		g_{j,\gls{balph},\gamma_{3-j}}(y) \, :=& \, \tau_j^{(\alpha_j)} - \tau_{3-j}^{(\alpha_{3-j})} \, + \, u_j(\disregj^{-1}(y+ \gamma_{3-j}),\alpha_j) \\
		& - u_{3-j}(\disregj^{-1}(y+\gamma_{3-j})-y,\alpha_{3-j}) \, .
	\end{split}
\end{align}
\par We now establish sufficient conditions on the pair potentials $\gls{pairpot}$ for the ergodic theorems above to be applicable.

\begin{lemma}\label{lem-f-galpha-L1-suff-cond}
	Let $r>1$. Assume that, for $\alpha_j\in\gls{cAj}$, $j=1,2$, $u_j(\cdot,\alpha_j)\in L^{\infty}(\gls{GamM})$, and that $f \in L^{\infty}_{2r}(\R^2\times[-L_{\gls{bu}}^z,L_{\gls{bu}}^z])$. Then we have that
	\begin{align}
		\begin{split}
			\int_{\R^2} \dx{x} |f(x+g_{j,\gls{balph},\gamma_{3-j}}(x))|\, \leq & \, \frac{5^{r-1}\pi}{r-1} \,\big(1+ \subldist \, + \, \|u_1\|_\infty + \|u_2\|_\infty\big)^{2r}\\
			& \, \|\jb{\cdot}^{2r}f\|_{L^\infty(\R^2\times[-L_{\gls{bu}}^z,L_{\gls{bu}}^z])} \, .
		\end{split}
	\end{align}
\end{lemma}

\begin{proof}
	Recalling definition \ref{def-jb} of $\jb{\cdot}$, we have that
	\begin{align}\label{eq-v-g-L1-1}
		\begin{split}
			\MoveEqLeft\int_{\R^2} \dx{x} |f(x+g_{j,\gls{balph},\gamma_{3-j}}(x))| \\
			&\leq \, \|\jb{\cdot}^{2r}f(\cdot+g_{j,\gls{balph},\gamma_{3-j}})\|_{L^\infty(\R^2)}\int_{\R^2} \dx{x}  \jb{x}^{-2r} \, .
		\end{split}
	\end{align}
	Observe that
	\begin{align}
		\begin{split}
			\jb{x+y}^{2r} \, =& \, (1+|x+y|^2)^r \\
			\leq & \, \big(1+2(|x|^2+|y|^2)\big)^r \, ,
		\end{split}
	\end{align}
	by Cauchy-Schwarz and monotonicity of $t\mapsto t^r$. Using convexity of $t\mapsto t^r$, we thus find that
	\begin{align} \label{eq-jb-triangle}
		\begin{split}
			\jb{x+y}^{2r} \, \leq& \, 5^r\Big(\frac15+\frac25\big(|x|^2+|y|^2\big)\Big)^r\\
			\leq & 2\cdot5^{r-1}\big(1+|x|^{2r}+|y|^{2r}\big) \, .
		\end{split}
	\end{align}
	Recall from \eqref{def-subldist} that
	\begin{equation}
		\subldist \, = \, \max_{\substack{\alpha_j\in\gls{cAj}\\j=1,2}}|\tau_1^{(\alpha_1)}- \tau_2^{(\alpha_2)}| \, .
	\end{equation}
	The definition \eqref{def-galpha} of $g_{j,\gls{balph} ,\gamma_{3-j}}$ thus implies
	\begin{align}\label{eq-galpha-Linfty}
		\begin{split}
			\MoveEqLeft\|g_{j,\gls{balph} ,\gamma_{3-j}}\|_\infty \, \leq\, \subldist \, + \, \|u_1\|_\infty + \|u_2\|_\infty \, ,
		\end{split}
	\end{align}
	and that
	\begin{align}\label{eq-galphaz-Linfty}
		\begin{split}
			\MoveEqLeft\|(g_{j,\gls{balph} ,\gamma_{3-j}})_z\|_\infty \, \leq\,  \subldist \, + \, \|u_{1,z}\|_\infty + \|u_{2,z}\|_\infty \, = \, L_{\gls{bu}}^{z} \, ,
		\end{split}
	\end{align}
	see definition \eqref{def-Luz} of $L_{\gls{bu}}^z$. A straight-forward computation yields
	\begin{equation} \label{eq-jbah-L1}
		\int_{\R^2} \dx{x}  \jb{x}^{-2r} \, = \, \frac{\pi}{2(r-1)} \, .
	\end{equation}
	In addition, observe that for any $x\in \ell^1(\Z)$ we have that
	\begin{equation}\label{eq-ellr-ell1-emb}
		\|x\|_{\ell^r} \, \leq \, \|x\|_{\ell^\infty}^{\frac{r-1}r}\|x\|_{\ell^1}^{\frac1r} \, \leq\, \|x\|_{\ell^1} \, .
	\end{equation}
	Collecting \eqref{eq-jb-triangle}, \eqref{eq-galpha-Linfty}, \eqref{eq-galphaz-Linfty}, and \eqref{eq-jbah-L1}, \eqref{eq-v-g-L1-1} implies
	\begin{align}
		\begin{split}
			\MoveEqLeft\int_{\R^2} \dx{x} |f(x+g_{j,\gls{balph},\gamma_{3-j}}(x))| \\
			\leq & \, \frac{5^{r-1}\pi}{r-1} \Big(\|\big(1+|\cdot+g_{j,\gls{balph} ,\gamma_{3-j}}|^{2r}\big)f(\cdot+g_{j,\gls{balph} ,\gamma_{3-j}})\|_{L^\infty(\R^2)} \\
			& \, + \, \||g_{j,\gls{balph} ,\gamma_{3-j}}|^{2r}f(\cdot+g_{j,\gls{balph} ,\gamma_{3-j}})\|_{L^\infty(\R^2)}\Big) \\
			\leq & \, \frac{5^{r-1}\pi}{r-1} \,\big(1+ \subldist \, + \, \|u_1\|_\infty + \|u_2\|_\infty\big)^{2r}\\
			& \, \|\jb{\cdot}^{2r}f\|_{L^\infty(\R^2\times[-L_{\gls{bu}}^z,L_{\gls{bu}}^z])} \, ,
		\end{split}
	\end{align}
	where, in the last step, we applied \eqref{eq-ellr-ell1-emb}. This concludes the proof.
\end{proof}

\begin{lemma}\label{lem-MBpot-pair-FT-formula}
	Let $r>1$ and $\gamma_{3-j}\in\Gamma_{3-j}$. Assume that, for $\alpha_j\in\gls{cAj}$, $j=1,2$, $u_j(\cdot,\alpha_j)\in L^{\infty}(\gls{GamM})$, and that
	\begin{align} 
		\gls{pairpot} \in L^{\infty}_{2r}(\R^2\times[-L_{\gls{bu}}^z,L_{\gls{bu}}^z]) \, .
	\end{align}
	Then we have that $\gls{SPinterjdiag} [\gls{bu}]\in L^1(\gls{GamM})$ and that
	\begin{align}
		\begin{split}
			\MoveEqLeft\gls{FTMoire}(\gls{SPinterjdiag} [\gls{bu}])(\MoireG)\\
			=& \, \frac{1}{2|\gls{GamM}|}\sum_{\substack{\alpha_k\in\cA_k\\k=1,2}}\int_{\R^2} \dx{x} e^{-i\MoireG \cdot (A_jx+\gamma_j)}\gls{pairpot}\big(A_1x + \gamma_1  +\tau_1^{(\alpha_1)}\\
			& +u_1(A_1x+\gamma_1,\alpha_1)-A_2x-\gamma_2-\tau_2^{(\alpha_2)}- u_2(A_2x+\gamma_2,\alpha_2)\big) \, .
		\end{split}
	\end{align}
	In particular, we have that
	\begin{align}
		\mavint \dx{x} \gls{SPinterjdiag} [\gls{bu}](x) \, = \, \mavint \dx{x} \sitepot_{3-j,\gls{bgam}}[\gls{bu}](x) \, .
	\end{align}
\end{lemma}

\begin{proof}
	
	Assume that $\gls{pairpot}$ and $u_j$ are simple functions. Interchanging integration and summation, we then have that
	\begin{align}
		\begin{split}
			\MoveEqLeft\gls{FTMoire}(\gls{SPinterjdiag} [\gls{bu}])(\MoireG) \, =\\
			& \, \frac1{2|\gls{Gamj}|}\sum_{\substack{\alpha_k\in\cA_k\\k=1,2}}\sum_{R_{3-j}\in\cR_{3-j}}\mavint \dx{x} e^{-i\MoireG \cdot x} \gls{pairpot}\big(\gls{disregj} x -R_{3-j}\\
			& +A_{3-j}A_j^{-1}\gamma_j-\gamma_{3-j}  + \tau_j^{(\alpha_j)} - \tau_{3-j}^{(\alpha_{3-j}) } \, + \, u_j(x,\alpha_j) \\
			& - \, u_{3-j}(R_{3-j}+A_{3-j}A_j^{-1}\gamma_j+ \gamma_{3-j} + (I- \gls{disregj})x,\alpha_{3-j})\big) \, .
		\end{split}
	\end{align}
	Substituting $y:=\gls{disregj} x + A_{3-j}A_j^{-1}\gamma_j - R_{3-j}-\gamma_{3-j}$ and interchanging integration and summation again, we obtain
	\begin{align} \label{eq-vbar-FT-1}
		\begin{split}
			\MoveEqLeft\gls{FTMoire}(\gls{SPinterjdiag} [\gls{bu}])(\MoireG) \, = \\
			& \frac{|\det(\disregj^{-1})|}{2|\gls{GamM}||\gls{Gamj}|}\sum_{\substack{\alpha_k\in\cA_k\\k=1,2}}\int_{\R^2} \dx{y} \sum_{R_{3-j}\in\cR_{3-j}} e^{-i\MoireG \cdot \disregj^{-1}(y-A_{3-j}A_j^{-1}\gamma_j+R_{3-j}+\gamma_{3-j})}\\
			& \mathds{1}_{\gls{disregj}\gls{GamM}+A_{3-j}A_j^{-1}\gamma_j-R_{3-j}-\gamma_{3-j}}(y)\,  \gls{pairpot}\Big(y + \tau_j^{(\alpha_j)}- \tau_{3-j}^{(\alpha_{3-j}) } \\
			& + \, u_j(\disregj^{-1}(y-A_{3-j}A_j^{-1}\gamma_j+R_{3-j}+\gamma_{3-j}),\alpha_j)\\
			&  \, - \, u_{3-j}(\disregj^{-1}(y-A_{3-j}A_j^{-1}\gamma_j+R_{3-j}+ \gamma_{3-j})-y,\alpha_{3-j})\Big) \, .
		\end{split}
	\end{align} 
	Observe that, by \eqref{eq-disregj-moire-per}, \eqref{eq-disregj-moire-per}, \eqref{def-Gammaj}, and \eqref{def-moirecell}, 
	\begin{equation}\label{eq-bot-det}
		|\det(\gls{disregj})| \, = \, \frac{|\det(A_{3-j})|}{|\det(\gls{Am})|} \, = \, \frac{|\Gamma_{3-j}|}{|\gls{GamM}|} \, .
	\end{equation}
	In addition, by \eqref{eq-disregj-moire-per}, \eqref{eq-disregj-moire-per}, and \eqref{def-rec-lat}, we have that
	\begin{equation} \label{eq-bot-t-MoireG}
		\gls{disregj}^{-T}\gls{cRm*} \, =  \, \cR_{3-j}^* \, .
	\end{equation}
	Furthermore, Lemma \ref{lem-disreg-transf} implies that
	\begin{equation} \label{eq-bot-MoireCell}
		\gls{disregj} \gls{GamM} \, = \, (-1)^j\Gamma_{3-j} \, ,
	\end{equation}
	and that
	\begin{equation} \label{eq-bot-1-R2}
		\disregj^{-1}\cR_{3-j} \, = \, \gls{cRm} \, .
	\end{equation}
	Collecting \eqref{eq-bot-det}, \eqref{eq-bot-t-MoireG}, \eqref{eq-bot-MoireCell}, \eqref{eq-bot-1-R2}, and using the facts that
	\begin{equation} \label{eq-R2-R2-decomp}
		\bigcup_{R_{3-j}\in \cR_{3-j}}\big((-1)^j\Gamma_{3-j}-R_{3-j}-\gamma_{3-j}\big) \, = \, \R^2 \, ,
	\end{equation}
	and that $u_k(\cdot,\alpha_k)$ is $\gls{cRm}$-periodic,
	\eqref{eq-vbar-FT-1} yields
	\begin{align}\label{eq-MBpot-FT-pair-formula-pre} 
		\begin{split}
			\MoveEqLeft\gls{FTMoire}(\gls{SPinterjdiag} [\gls{bu}])(\MoireG)\\
			= \, & \frac{1}{2|\Gamma_1||\Gamma_2|}\sum_{\substack{\alpha_k\in\cA_k\\k=1,2}}\int_{\R^2} \dx{y} e^{-i\MoireG \cdot \disregj^{-1}(y-A_{3-j}A_j^{-1}\gamma_j+\gamma_{3-j})}\\
			& \gls{pairpot}\Big(y + \tau_j^{(\alpha_j)}- \tau_{3-j}^{(\alpha_{3-j}) } \\
			& + \, u_j(\disregj^{-1}(y-A_{3-j}A_j^{-1}\gamma_j+\gamma_{3-j}),\alpha_j)\\
			&  \, - \, u_{3-j}(\disregj^{-1}(y-A_{3-j}A_j^{-1}\gamma_j+ \gamma_{3-j})-y,\alpha_{3-j})\Big) \, .
		\end{split}
	\end{align}
	Using the fact that
	\begin{equation}
		A_j-A_{3-j} \, = \, \gls{disregj} A_j \, = \, (-1)^jA_{3-j}\MoireRLV^{-1}A_j\, , 
	\end{equation}
	see Lemma \ref{lem-disregj-calc}, and substituting $y=(A_j-A_{3-j}) \xi + \gamma_j -\gamma_{3-j}$, we obtain
	\begin{align}\label{eq-MBpot-FT-pair-formula} 
		\begin{split}
			\MoveEqLeft\gls{FTMoire}(\gls{SPinterjdiag} [\gls{bu}])(\MoireG)\\
			=& \, \frac{1}{2|\gls{GamM}|}\sum_{\substack{\alpha_k\in\cA_k\\k=1,2}}\int_{\R^2} \dx{\xi} e^{-i\MoireG \cdot (A_j\xi+\gamma_j)}\gls{pairpot}\big(A_1\xi + \gamma_1  +\tau_1^{(\alpha_1)}\\
			& +u_1(A_1\xi+\gamma_1,\alpha_1)-A_2\xi-\gamma_2-\tau_2^{(\alpha_2)}- u_2(A_2\xi+\gamma_2,\alpha_2)\big) \, ,
		\end{split}
	\end{align}
	where we also employed that $\gls{pairpot}$ is even. Due to Lemma \ref{lem-f-galpha-L1-suff-cond}, we can extend \eqref{eq-MBpot-FT-pair-formula-pre}, and thus \eqref{eq-MBpot-FT-pair-formula}, to all 
	\begin{equation}
		\gls{pairpot} \in L^{\infty}_{2r}(\R^2\times[-L_{\gls{bu}}^z,L_{\gls{bu}}^z]) \, .    
	\end{equation}
	This finishes the proof.
\end{proof}

Define for almost all $x,y\in\R^2$
\begin{align}\label{def-halpha}
	\begin{split}
		h_{j,\gls{balph}}(x,y) \, := \, \tau_j^{(\alpha_j)}- \tau_{3-j}^{(\alpha_{3-j}) } + u_j(x,\alpha_j)  \, -\, u_{3-j}(x-y,\alpha_{3-j}) \, .
	\end{split}
\end{align}
For a pair-potential \eqref{def-pair-potential}, \eqref{def-sitepotdouble} becomes
\begin{align}\label{eq-sitepotdouble-pair-pot}
	\begin{split}
		\MoveEqLeft\gls{SPinterj} [\gls{bu}](x,y) \, =\\
		&\frac1{2|\gls{Gamj}|}\sum_{\substack{\alpha_k\in\cA_k\\k=1,2}} \sum_{\substack{R_{3-j}\\ \in\cR_{3-j}}} \gls{pairpot}\big(y -R_{3-j} + h_{j,\gls{balph}}(x,y-R_{3-j}) \big) \, .
	\end{split}
\end{align}

\begin{lemma}\label{lem-MBpotshift-L1}
	Let $r>1$. Assume that, for $\alpha_j\in\gls{cAj}$, $j=1,2$, $u_j(\cdot,\alpha_j)\in L^{\infty}(\gls{GamM})$, and that
	\begin{align} 
		\gls{pairpot} \in L^{\infty}_{2r}(\R^2\times[-L_{\gls{bu}}^z,L_{\gls{bu}}^z]) \, .
	\end{align}
	Then we have that $\gls{SPinterj} [\gls{bu}]\in L^1(\gls{GamM}\times\Gamma_{3-j})$.
\end{lemma}

\begin{proof}
	The idea of this proof follows the steps in the proofs of the two previous Lemmata. As above, assume that $\gls{pairpot}$ and $u_j$ are simple functions.
	\par Let
	\begin{equation}\label{def-w-temp}
		\psi(x,y) \, := \, \frac1{2|\gls{Gamj}|}\sum_{\substack{\alpha_k\in\cA_k\\k=1,2}} \gls{pairpot}\big(y  +h_{j,\gls{balph}}(x,y)\big) \, .
	\end{equation}
	Then we have that 
	\begin{align}
		\begin{split}
			\MoveEqLeft\int_{\gls{GamM}}  \dx{x} \int_{\Gamma_{3-j}}\dx{y} |\gls{SPinterj} [\gls{bu}](x,y)| \, \leq \\
			& \sum_{R_{3-j}\in\cR_{3-j}}\int_{\gls{GamM}}  \dx{x} \int_{\Gamma_{3-j}}\dx{y} |\psi(x-R_{3-j},y)| \, .
		\end{split}
	\end{align}
	Shifting $x-R_{3-j}\to x$, we obtain
	\begin{align}\label{eq-MBpotshift-L2-bd-0}
		\begin{split}
			\MoveEqLeft\int_{\gls{GamM}}  \dx{x} \int_{\Gamma_{3-j}}\dx{y} |\gls{SPinterj} [\gls{bu}](x,y)| \, \leq\\
			& \sum_{R_{3-j}\in\cR_{3-j}}\int_{\Gamma_{3-j}-R_{3-j}} \dx{x} \int_{\gls{GamM}} \dx{y} |\psi(x,y)|\, = \\
			& \int_{\R^2} \dx{x} \int_{\gls{GamM}} \dx{y} |\psi(x,y)| \, ,
		\end{split}
	\end{align}
	where, again, we used the fact that
	\begin{equation}
		\bigcup_{R_{3-j}\in \cR_{3-j}}\big(\Gamma_{3-j}-R_{3-j}\big) \, = \, \R^2 \, .
	\end{equation}
	Due to
	\begin{equation}
		\int_{\R^2} \dx{x} \jb{x}^{-2r} \, < \, \infty 
	\end{equation}
	for $r>1$, \eqref{eq-MBpotshift-L2-bd-0} implies
	\begin{align}
		\begin{split}
			\MoveEqLeft\int_{\gls{GamM}}  \dx{x} \int_{\Gamma_{3-j}}\dx{y} |\gls{SPinterj} [\gls{bu}](x,y)| \, \lesssim\\
			& \|\jb{x}^{2r}\psi(x,y)\|_{L^\infty_{x,y}} \, .
		\end{split}
	\end{align}
	With analogous steps as in the proof of Lemma \ref{lem-f-galpha-L1-suff-cond}, we have that 
	\begin{align} \label{eq-jbx-psi-bd}
		\begin{split}
			\|\jb{x}^{2r} \psi(x,y)\|_{L^\infty_{x,y}} \, \lesssim& \, (1+\subldist+\|u_1\|_\infty+\|u_2\|_\infty)^{2r}\\ 
			& \sum_{\substack{\alpha_k\in\cA_k\\k=1,2}} \|\jb{\cdot}^{2r}\gls{pairpot}\|_{L^\infty(\R^2\times[-L_{\gls{bu}}^z,L_{\gls{bu}}^z])}  \, < \, \infty 
		\end{split}
	\end{align}
	by assumption on $\gls{pairpot}$. This concludes the proof.
\end{proof}

\subsection{Proofs of main theorems\label{sec-proofs-main-thm}}

\begin{proof}[Proof of Theorem \ref{thm-smooth-conv}] 
	We are left with proving the upper bound in the case of interlayer pair potentials. Observe that, due to $\sigma>\frac{1433}{1248}$ and $s>1$, we have that $\ceil{\sigma+s}\geq 3$. We choose $s=3-\sigma$ for some $\sigma\in\big(\frac{1433}{1248},2\big)$.
	\par Employing the first part of the theorem, we have that
	\begin{equation}\label{eq-2B-LB}
		|\inter_{j,N,0}(\gls{bu})-\inter_{j,0}(\gls{bu})| \, \leq \, \frac{\errMB}{2N+1}
	\end{equation}
	with $\errMB$ given in \eqref{def-error-MB}, with the choice $s=3-\sigma$, $\ell=\mathrm{inter}$, by
	\begin{align}\label{eq-errTB-lb}
		\begin{split}
			\errMB_j \, &= \, \frac{2\sqrt{2}}{K}\Big(\frac{\gls{moirelen} \|A\|_2}{2\pi}\Big)^6\big(\zeta(6-2\sigma)+2^{\sigma-3}\zeta(3-\sigma)^2\big)\\
			& \quad \sup_{\MoireG\in\gls{cRm*}}|\MoireG|^6 |\gls{FTMoire}(\sitepot_{j,0}[\gls{bu}])(\MoireG)| \, .
		\end{split}
	\end{align}

	Since $\sigma\in\big(\frac{1433}{1248},2\big)$ was arbitrary, \eqref{eq-2B-LB} and \eqref{eq-errTB-lb} imply
	\begin{align}\label{eq-2B-TD-UB}
		\begin{split}
			|\inter_{j,N,0}(\gls{bu})-\inter_{j,0}(\gls{bu})| \, \leq & \, 2\sqrt{2}\Big(\frac{\sqrt{3}|\gls{GamM}|}{2^3\pi^2}\Big)^3 \pairUB(\theta)\\
			& \, \sup_{\MoireG\in\gls{cRm*}}|\MoireG|^6 |\gls{FTMoire}(\sitepot_{j,0}[\gls{bu}])(\MoireG)| \, ,
		\end{split}
	\end{align}
    where we recall definition \eqref{def-pairUB} of $\pairUB(\theta)$. \eqref{eq-MBpot-FT-pair-formula-pre} in the proof of Lemma \ref{lem-MBpot-pair-FT-formula} implies 
	\begin{align}
		\begin{split}
			\MoveEqLeft\gls{FTMoire}(\gls{SPinterjdiag} [\gls{bu}])(\MoireG)\\
			=& \frac{1}{2|\Gamma_1||\Gamma_2|}\sum_{\substack{\alpha_k\in\cA_k\\k=1,2}}\int_{\R^2} \dx{y} e^{-i\gls{disregj}^{-T}\MoireG \cdot (y+\gamma_{3-j})} \gls{pairpot}\big(y +\tau_j^{(\alpha_j)} \\
			&  - \tau_{3-j}^{(\alpha_{3-j})} \, + \, u_j(\disregj^{-1}y,\alpha_j) - u_{3-j}(\disregj^{-1}y-y,\alpha_{3-j})\big) \, .
		\end{split}
	\end{align}
	In particular, we have that $\gls{FTMoire}(\sitepot_{j,0}[\gls{bu}])(0)= \inter_{j,0}(\gls{bu})$. As in \eqref{def-galpha}, we define 
	\begin{align}\label{def-galpha-pair-rate}
		g_{j,\gls{balph}}(y) \, :=& \, \tau_j^{(\alpha_j)} - \tau_{3-j}^{(\alpha_{3-j})} \, + \, u_j(\disregj^{-1}y,\alpha_j) - u_{3-j}(\disregj^{-1}y-y,\alpha_{3-j}) \, .
	\end{align}
	Now let $\sigma\in\big(\frac{1433}{1248},2\big)$ be arbitrary. Using Lemma \ref{lem-disregj-calc} and the definition of the fractional derivative $|\nabla|$ via Fourier multiplication, we find that
	\begin{align} \label{eq-vbar-FT-3}
		\begin{split}
			\MoveEqLeft|\MoireG|^6\gls{FTMoire}(\sitepot_{j,0}[\gls{bu}])(\MoireG)\\
			=& \, \frac{1}{2|\Gamma_1||\Gamma_2|}\Big(\frac{\gls{moirelen}}{q^{(-1)^{j+1}1/2}}\Big)^6\sum_{\substack{\alpha_k\in\cA_k\\k=1,2}}\int_{\R^2} \dx{y} e^{-i\gls{disregj}^{-T}\MoireG \cdot y} \\
			& |\nabla_y|^6\Big(\gls{pairpot}\big(y + g_{j,\gls{balph}}(y)\big)\Big) \\
			\leq & \, \frac{1}{2|\Gamma_1||\Gamma_2|}\Big(\frac{\gls{moirelen}}{q^{(-1)^{j+1}1/2}}\Big)^6\, \sum_{\substack{\alpha_k\in\cA_k\\k=1,2}} \|(-\Delta)^3\gls{pairpot}(\cdot+g_{j,\gls{balph}})\|_1\, .
		\end{split}
	\end{align}
	For any $n\in\N$, define
	\begin{equation}\label{def-R-FdB}
		R(n) \, := \, \{{\bf r}_n\in \N_0^n \mid \sum_{j=1}^njr_j \, = \, n\} \, .
	\end{equation}
	Using the Fa\'a di Bruno formula and recalling \eqref{def-R-FdB}, we have that
	\begin{align}\label{eq-erg-FdB-1}
		\begin{split}
			\MoveEqLeft\|(-\Delta)^3\gls{pairpot}(\cdot+g_{j,\gls{balph}})\|_1 \\
			\leq & \, 6!\sum_{{\bf r}_{6}\in R(6)}\big\|\big(D^{|{\bf r}_{6}|_1}\gls{pairpot}\big)(\cdot+g_{j,\gls{balph}})\big\|_1\\
			& \, \prod_{k=1}^{6}\frac1{r_k!}\bigg(\frac{\delta_{k,1}+\|D^kg_{j,\gls{balph}}\|_\infty}{k!}\bigg)^{r_k} \\
			\leq& \, 6!\sum_{{\bf r}_{6}\in R(6)}\big\|\big(D^{|{\bf r}_{6}|_1}\gls{pairpot}\big)(\cdot+g_{j,\gls{balph}})\big\|_1 \\
			& \, \Big(1+\sum_{k=1}^6\|D^kg_{j,\gls{balph}}\|_\infty\Big)^{|{\bf r}_6|_1} \prod_{k=1}^{6}\frac1{r_k!(k!)^{r_k}} \, .
		\end{split}
	\end{align}
	Employing definition \eqref{def-R-FdB}, we obtain that
	\begin{equation}
		1 \, \leq \, |{\bf r}_6|_1 \, = \, \sum_{k=1}^{6} r_k \, \leq \, \sum_{k=1}^{6} kr_k \, = \, 6 
	\end{equation}
	for all ${\bf r}_{6}\in R(6)$. Hence, \eqref{eq-erg-FdB-1} yields
	\begin{align}\label{eq-erg-error-bd-1}
		\MoveEqLeft\|(-\Delta)^3\gls{pairpot}(\cdot+g_{\gls{balph}})\|_1 \, \leq  \, B_6 \, \sum_{k=1}^{6}\big\|\big(D^k\gls{pairpot}\big)(\cdot+g_{j,\gls{balph}})\big\|_1\Big(1+\sum_{\ell=1}^6\|D^\ell g_{\gls{balph}}\|_\infty\Big)^6 \, ,
	\end{align}
	where 
	\begin{equation}\label{eq-B6}
		B_6 \, := \, 6!\sum_{{\bf r}_{6}\in R(6)}\prod_{k=1}^{6}\frac1{r_k!(k!)^{r_k}} \, = \, \partial_\lambda^6\Big|_{\lambda=0}e^{e^\lambda-1} \, = \, 203    
	\end{equation}
	denotes the 6$^{th}$ Bell number.
	\par Recalling definition \eqref{def-galpha-pair-rate} of $g_{j,\gls{balph}}$ and employing Lemma \ref{lem-disregj-calc}, we find that
	\begin{align} \label{eq-Dl-galpha}
		\begin{split}
			\|D^\ell g_{j,\gls{balph}}\|_\infty \, \leq& \, \|D^\ell u_j(\disregj^{-1}(\cdot),\alpha_j)\|_\infty \\
			&+ \, \|D^\ell u_{3-j}((\disregj^{-1}-I)(\cdot),\alpha_{3-j})\|_\infty \\
			\leq & \, \Big(\frac{q^{(-1)^{j+1}1/2}}{\gls{moirelen}}\Big)^\ell\|D^\ell u_1(\cdot,\alpha_1)\|_\infty \\
			& \, + \, \Big(\frac{q^{(-1)^j1/2}}{\gls{moirelen}}\Big)^\ell\|D^\ell u_2(\cdot,\alpha_2)\|_\infty \\
			\leq & \,\frac{q^{(-1)^{j+1}\ell/2}+q^{(-1)^j\ell/2}}{\gls{moirelen}^\ell}\big(\|D^\ell u_1\|_\infty+\|D^\ell u_2\|_\infty\big)\, .
		\end{split}
	\end{align}
	Lemma \ref{lem-f-galpha-L1-suff-cond} implies
	\begin{align}\label{eq-v-g-L1-2}
		\begin{split}
			\MoveEqLeft\big\|\big(D^k\gls{pairpot}\big)(\cdot+g_{j,\gls{balph}})\big\|_1 \\
			\leq & \, \frac{5^{r-1}\pi}{r-1} \,\big(1+ \subldist \, + \, \|u_1\|_\infty + \|u_2\|_\infty\big)^{2r}\\
			& \, \|\jb{\cdot}^{2r}\big(D^k\gls{pairpot}\big)\|_{L^\infty(\R^2\times[-L_{\gls{bu}}^z,L_{\gls{bu}}^z])} \, .
		\end{split}
	\end{align}
	\par Collecting \eqref{eq-vbar-FT-3}, \eqref{eq-erg-FdB-1}, \eqref{eq-erg-error-bd-1}, \eqref{eq-B6}, \eqref{eq-Dl-galpha}, and \eqref{eq-v-g-L1-2}, and recalling the weighted Sobolev norm \eqref{def-weighted-sob}, we obtain
	\begin{align} 
		\begin{split}
			\MoveEqLeft\sup_{\MoireG}\Big||\MoireG|^{6}\gls{FTMoire}(\sitepot_{j,0}[\gls{bu}])(\MoireG)\Big| \,\\
			\leq & \, 203\frac{1+q^{(-1)^j6}}{2|\Gamma_1||\Gamma_2|}\frac{5^{r-1}\pi}{r-1}\\
			& \, \big(1+\subldist+\|u_1\|_{W^{6,\infty}} + \|u_2\|_{W^{6,\infty}}\big)^{6+2r} \\
			& \, \sum_{\substack{\alpha_k\in\cA_k\\k=1,2}} \|\gls{pairpot}\|_{W^{6,\infty}_{2r}(\R^2\times[-L_{\gls{bu}}^z,L_{\gls{bu}}^z])} \, .
		\end{split}
	\end{align}
	Together with \eqref{eq-2B-TD-UB}, this concludes the proof.
\end{proof} 

\begin{proof}[Proof of Theorem \ref{thm-rough-conv}]
	We are left with proving that the conditions in the case of interlayer pair potentials are sufficient.
    \par Lemma \ref{lem-MBpotshift-L1} implies that $\gls{SPinterj} [\gls{bu}]\in L^1(\gls{GamM}\times\Gamma_{3-j})$. Moreover, Lemma \ref{lem-MBpot-pair-FT-formula} implies that $\sitepot_1[u]\in L^1(\gls{GamM})$. Thus, the conditions of Proposition \ref{prop-double-erg-thm} are satisfied.
\end{proof}

\appendix

\section{Lattice calculus \label{app-latt-calc}}

\begin{proof}[Proof of Lemma \ref{lem-moire-len}]
	By definitions \eqref{def-MoireMLV} and \eqref{def-Aj}, we have that
	\begin{align}\label{eq-moirerlv-formula}
		\gls{Am} \, = \, (A_1^{-1}-A_2^{-1})^{-1} \, = \, (q^{1/2}R_{\theta/2}-q^{-1/2}R_{-\theta/2})^{-1}A \, .
	\end{align}
	Consequently, we obtain that
	\begin{align}\label{eq-moiremlv-formula}
		\gls{Bm} \, = \,  (q^{1/2}R_{-\theta/2}-q^{-1/2}R_{\theta/2}) 2\pi A^{-T} \, .
	\end{align}
	Observe that
	\begin{equation}\label{eq-moire-scale-matrix}
		q^{1/2}R_{\theta/2}-q^{-1/2}R_{-\theta/2} \, = \, \begin{pmatrix}
			(q^{1/2}-q^{-1/2})\cos(\theta/2) & -(q^{1/2}+q^{-1/2})\sin(\theta/2) \\
			(q^{1/2}+q^{-1/2})\sin(\theta/2) & (q^{1/2}-q^{-1/2})\cos(\theta/2)
		\end{pmatrix} \, .
	\end{equation}
	A straight-forward calculation yields
	\begin{align}
		|(q^{1/2}R_{\theta/2}-q^{-1/2}R_{-\theta/2})x|^2 \, = \, [(q^{1/2}-q^{-1/2})^2 \, + \, 4\sin^2(\theta/2)]|x|^2 \, .
	\end{align}
	Consequently, we have that
	\begin{equation}
		|(q^{1/2}R_{\theta/2}-q^{-1/2}R_{-\theta/2})^{-1}x| \, = \, [(q^{1/2}-q^{-1/2})^2 \, + \, 4\sin^2(\theta/2)]^{-\frac12}|x| \, .
	\end{equation}
	Together with \eqref{eq-moirerlv-formula} resp. \eqref{eq-moiremlv-formula}, this concludes the proof.
\end{proof}

\begin{proof}[Proof of Proposition \ref{prop-reconstr}]
	By Proposition \ref{prop-pw-erg-thm-gen}, let $(\cN_{\MoireG}))_{\MoireG\in\gls{cRm*}}$ be a sequence of nullsets such that for all $\MoireG\in\gls{cRm*}$ and all $\gamma_j\in\gls{Gamj}\setminus\cN_{\MoireG}$
	\begin{align}
		\MoveEqLeft \lim_{N\to\infty} \frac1{(2N+1)^2}\sum_{R_j\in\gls{cRjN}} e^{-i\MoireG\cdot (\gamma_j+R_j)} u_j(\gamma_j+R_j,\alpha_j) \\
		& = \, \mavint \dx{x} e^{-i\MoireG\cdot x} u_j(x,\alpha_j) \, = \, \hat{u}_j(\MoireG,\alpha_j) \, ,
	\end{align}
	where we used the fact that
	\begin{equation}
		e^{-i\MoireG\cdot x} \, = \, e^{-i\MoireG\cdot \gls{mfracx}} \, .
	\end{equation}
	Then
	\begin{align}
		\cN \, := \,  \bigcup_{\MoireG\in\gls{cRm*}}\cN_{\MoireG}
	\end{align}
	is a nullset, and for all $\gamma_j\in\gls{Gamj}\setminus\cN$, we have that
	\begin{align}
		\lim_{N\to\infty} \frac1{(2N+1)^2}\sum_{R_j\in\gls{cRjN}} e^{-i\MoireG\cdot (\gamma_j+R_j)} u_j(\gamma_j+R_j,\alpha_j)  \, = \, \hat{u}_j(\MoireG,\alpha_j) \, .
	\end{align}
	This concludes the proof.
\end{proof}

\begin{lemma}\label{lem-sum-1/Gs}
	Let $s>1$. Then we have that
	\begin{equation}
		\sum_{\MoireG\in\gls{cRm*}\setminus\{0\}}\frac1{|\MoireG|^{2s}} \, \leq\, 4\Big(\frac{\gls{moirelen} \|A\|_2}{2\pi}\Big)^{2s}\big(\zeta(2s)+2^{-s}\zeta(s)^2\big) \, .
	\end{equation}
\end{lemma}
\begin{proof}
	Lemma \ref{lem-moire-len} yields
	\begin{subequations}
		\begin{align}
			|\gls{Bm} n|^2 \, &= \, (2\pi)^2\gls{moirelen}^{-2}|A^{-T}n|^2\nonumber\\
			& \geq \, (2\pi)^2\gls{moirelen}^{-2} \|A\|_2^{-2} |n|^2 \label{eq-MoireRL-QM-lb}\\
			& \geq \, 2\cdot (2\pi)^2\gls{moirelen}^{-2} \|A\|_2^{-2} |n_1| \cdot |n_2| \, , \label{eq-MoireRL-GM-lb}
		\end{align}
	\end{subequations}
	where, in the last step, we applied the geometric-quadratic-mean inequality. We split the sum over $\gls{cRm*}\setminus\{0\}$ into 
	\begin{equation}
		\{\gls{Bm} n \mid n_1n_2=0,(n_1,n_2)\neq(0,0)\} \dot\cup\{\gls{Bm} n \mid n_1,n_2\neq0\} \, .
	\end{equation}
	Applying \eqref{eq-MoireRL-QM-lb} on $\MoireG$ in the first set and \eqref{eq-MoireRL-GM-lb} on the latter, we obtain
	\begin{align}
		\sum_{\MoireG\in\gls{cRm*}\setminus\{0\}}\frac1{|\MoireG|^{2s}} \, &\leq \, \Big(\frac{\gls{moirelen} \|A\|_2}{2\pi}\Big)^{2s}\Big[4\sum_{n\in\N}\frac1{n^{2s}} \, + \, 2^{-s}\Big(2\sum_{n\in\N}\frac1{n^s}\Big)^2\Big] \\
        &= \, 4\Big(\frac{\gls{moirelen} \|A\|_2}{2\pi}\Big)^{2s}\big(\zeta(2s)+2^{-s}\zeta(s)^2\big) \, ,\label{eq-1/GS-sum-proof-0}
	\end{align}
    where we recognize the Riemann zeta function $\zeta(\sigma)=\sum_{m\in\N} \frac1{m^\sigma}$, $\sigma>1$. This concludes the proof. 
\end{proof}

\section{Existence of Diophantine 2D rotations \label{app-dio-proof}}

We now prove Proposition \ref{prop-dio-existence}. It suffices to show that for almost every $\theta\in\R$ and all $\sigma>\frac{1433}{1248}$, there exists $K>0$ such that for all $n\in\Z^2\setminus\{0\}$
\begin{equation}
	\dist\big(qA^TR_{\theta}A^{-T} n,\Z^2\big) \geq \frac{K}{|n|^{2\sigma}}
\end{equation}
holds.
\par Due to $2\pi$-periodicity of $R_\theta$, it suffices to show that Lebesgue-almost every $\theta\in[0,2\pi)$ is a Diophantine 2D rotation. We follow the classical strategy to apply the Borel-Cantelli Lemma to the complement. In particular, let
\begin{equation}
	\Omega \, := \, \bigcap_{\sigma>\frac{1433}{1248}}\bigcup_{\substack{K>0,\\N\in\N}}\bigcap_{\substack{n\in\Z^2:\\ |n|^2\geq N}}\bigcap_{m\in\Z^2}\Big\{\theta\in[0,2\pi) \mid |qA^TR_\theta A^{-T} n-m| \geq \frac{K}{|n|^{2\sigma}} \Big\} \, .
\end{equation}
$\Omega$ consists of all $\theta\in[0,2\pi)$ such that for all $\sigma>\frac{1433}{1248}$ there exists $K>0$ and $N\in\N$ s.t. for all $n\in\Z^2$ with $|n|^2\geq N$ we have that
\begin{equation} \label{eq-almost-all-diophantine}
	\dist\big(qA^TR_\theta A^{-T} n,\Z^2\big) \geq \frac{K}{|n|^{2\sigma}} \, .
\end{equation}
Taking $\tilde{K}:=\min\big\{K,\min_{|n|^2<N}\dist\big(qA^TR_\theta A^{-T} n,\Z^2\big)\big\}$, we can extend \eqref{eq-almost-all-diophantine} to hold for all $n\in\Z^2\setminus\{0\}$ by lowering $K$ to $\tilde{K}$. Our goal is to show that $|\Omega| \, = \, 2\pi$. 

\par Let 
\begin{equation}
	E_N(\sigma) \, := \, \bigcap_{K>0}\bigcup_{\substack{n\in\Z^2\\ |n|^2=N}}\bigcup_{m\in\Z^2} \Big\{\theta\in[0,2\pi) \mid |qA^TR_\theta A^{-T} n-m| < \frac{K}{|n|^{2\sigma}} \Big\} \, .
\end{equation}
Since $A$ is invertible, we have that $|A^T\cdot|$ and $|\cdot|$ are equivalent norms. More precisely, we have that
\begin{equation} \label{eq-AT-equiv-norm}
	\|A^T\|_2^{-1}|A^Tv| \, \leq\,  |v| \, \leq \, \|A^{-T}\|_2|A^Tv| \, ,
\end{equation}
where $\|\cdot\|_2$ denotes the spectral/Hilbert-Schmidt norm. In particular, we have that
\begin{equation}
	E_N(\sigma) \, \subseteq \, \bigcup_{\substack{n\in\Z^2\\ |n|^2=N}}\bigcup_{m\in\Z^2} \Big\{\theta\in[0,2\pi) \mid |qR_\theta A^{-T} n-A^{-T} m| < \frac{1}{|n|^{2\sigma}} \Big\} \, ,
\end{equation}

In order to apply Borel-Cantelli, we need to control the asymptotic behavior of the Lebesgue measures of the sets $E_N$ in order to study their summability. We have that
\begin{align}\label{eq-En-til-bd-0}
	\begin{split}
		\lefteqn{|E_N(\sigma)|}\\
		&\leq \, \sum_{\substack{n\in\Z^2\\ |n|^2=N}}\sum_{\substack{m\in\Z^2\\\big||A^{-T} m|-|qA^{-T} n|\big|<\frac{1}{|n|^{2\sigma}}}} \int_0^{2\pi} d\theta \,  \mathds{1}_{B_{\frac{1}{|n|^{2\sigma}}}(A^{-T} m)}\big(qR_\theta A^{-T} n\big) \, .
	\end{split}
\end{align}

Observe that for large $|n|$, we have that $|A^{-T}m|\approx_{\sigma} |qA^{-T}n|$. In particular, the integrand vanishes unless $\theta$ is on an arc of length $\sim |n|^{-2\sigma}$ for a circle of radius $|qA^{-T}n|\sim|n|$, see \eqref{eq-AT-equiv-norm}, about the angle defined by $A^{-T}n$. In particular, we obtain that 
\begin{equation}\label{eq-arclen-bd-0}
	\int_0^{2\pi} \dx{\theta} \mathds{1}_{B_{\frac{1}{|n|^{2\sigma}}}(A^{-T} m)}\big(R_\theta A^{-T} n\big) \, \lesssim_{q,\sigma,A} \frac1{|n|^{1+2\sigma}} 
\end{equation}
for sufficiently large $n$. In addition, for $r>0$, Levitan \cite{levitan1987} obtained the bound 
\begin{equation} \label{eq-gen-lattice-ball-asymp}
	|B_r\cap qA^{-T}\Z^2| \, = \, \frac{\pi r^2}{|\det(qA^{-T})|}\, + \, O_{r\to\infty}(r^{\frac23}) \, ,
\end{equation} see also \cite{laxphillips} for a related works, and also \cite{Frickergitterpunktlehre}. As a consequence, we find that
\begin{align} \label{eq-ring-bd-0}
	\begin{split}
		\MoveEqLeft\sum_{m\in\Z^2} \mathds{1}_{B_{|n|^{-2\sigma}}(qA^{-T} n)}(A^{-T} m)\\
		& \leq \, \Big|\Big(B_{|qA^{-T}n|+\frac{1}{|n|^{2\sigma}}}\setminus B_{|qA^{-T}n|-\frac{1}{|n|^{2\sigma}}}\Big)\cap A^{-T}\Z^2\Big|  \\
		& = \, \big(|qA^{-T}n|+\frac{1}{|n|^{2\sigma}}\big)^2 - \big(|qA^{-T}n|-\frac{1}{|n|^{2\sigma}}\big)^2 + O_{|n|\to\infty}(|n|^{\frac23})\\
		&\lesssim_{q,\sigma,A} \, |n|^{\frac23} \, ,
	\end{split}
\end{align}
for $2\sigma>1$. In addition, Huxley \cite{huxley2003exponential} showed that
\begin{align}
	\label{eq-sq-lattice-ball-asymp}
	|B_r\cap \Z^2| \, = \, \pi r^2 \, + \, O_{r\to\infty}(r^{\frac{131}{208}}\log(r)^{\frac{18627}{8320}}) \, ,
\end{align}
see also \cite{berndt2018circle,bourgain2017mean,KURATSUBO2022} for recent progress. Consequently, we obtain that
\begin{align} \label{eq-sq-part-func-0}
	\begin{split}
		|\{n\in \Z^2 \mid |n|^2=N\}| \, &\leq \, \big|\big((B_{\sqrt{N}+\delta}\setminus B_{\sqrt{N}-\delta})\cap\Z^2\big)\big|\\
		&\lesssim N^{\frac{131}{416}+\vep} 
	\end{split}
\end{align}
for any $\delta,\vep>0$.
\par Collecting \eqref{eq-arclen-bd-0}, \eqref{eq-ring-bd-0}, and \eqref{eq-sq-part-func-0}, \eqref{eq-En-til-bd-0} yields
\begin{equation}
	|E_N(\sigma)| \, \lesssim_{\sigma,A} \, N^{\frac{131}{416}+\frac13-\frac12-\sigma+\vep} \, = \, N^{\frac{185}{1248}-\sigma +\vep}
\end{equation}
for any $N\in\N$ large enough. In particular, choosing any
\begin{equation} \label{eq-sigma-lbd-0}
	\sigma \, > \, \frac{185}{1248} \, + \, 1 \, = \, \frac{1433}{1248}
\end{equation}
and any $0<\vep<\sigma-\frac{1433}{1248}$, we obtain that
\begin{equation}
	\sum_{N\in\N} |E_N(\sigma)| \, < \, \infty \, .
\end{equation}
Then the Borel-Cantelli Lemma yields that
\begin{equation}\label{eq-borelcant-0}
	\Big|\bigcap_{N_0=1}^\infty \bigcup_{N=N_0}^\infty E_N(\sigma)\Big| \, = \, 0 \, .
\end{equation}
Observe that $E_N(\sigma)$ is decreasing in $\sigma$, i.e., 
\begin{equation}
	E_N(\sigma') \, \subseteq E_N(\sigma)
\end{equation}
for $\sigma\leq \sigma'$. Thus, and using \eqref{eq-borelcant-0}, we obtain that 
\begin{equation}
	\begin{split}
		\Big|\bigcup_{\sigma>\frac{1433}{624}}\bigcap_{N_0=1}^\infty \bigcup_{N=N_0}^\infty E_N(\sigma)\Big| \, &\leq \, \sum_{\ell=1}^\infty \Big|\bigcap_{N_0=1}^\infty \bigcup_{N=N_0}^\infty E_N\big(\frac{1433}{1248}+\frac1\ell\big)\Big| \\
		&= \, 0 \, .
	\end{split}
\end{equation}
Using the fact that
\begin{align}
	\Omega \, \supseteq \, \Big(\bigcup_{\sigma>\frac{1433}{1248}}\bigcap_{N_0=1}^\infty \bigcup_{N=N_0}^\infty E_N(\sigma)\Big)^c \, ,
\end{align}
we finish the proof.

\section{Proofs of ergodic theorems\label{app-erg-thm-proof}}

\begin{proof}[Proof of Proposition \ref{prop-double-erg-thm}]
	
	Proposition \ref{prop-pw-erg-thm-gen} implies almost sure pointwise convergence. For both, almost sure and $L^1$-convergence, imply convergence in probability, it suffices to compute the $L^1$-limit. Using standard approximation arguments, we may assume that $h$ is a $C^\infty$-smooth function. By dominated convergence, it suffices the almost sure pointwise limit.  
	\par W.l.o.g. we set $j=1$. \eqref{eq-sum-double-trafo} implies that
	\begin{align}
		\begin{split}
			\MoveEqLeft \ergav_{\cR_1(N)}(h)(\omega_\cM,\omega_2) \\
			&= \, \sum_{\substack{G_2\in\cR_2^*,\\\MoireG\in\gls{cRm*}}}  \appdelN{1}(G_2+\MoireG) \,  \hat{h}(\MoireG,G_2) e^{i[\MoireG\cdot \omega_\cM+G_2\cdot\omega_2]}
		\end{split}
	\end{align}
	Using dominated convergence, we obtain that
	\begin{align} \label{eq-erav-conv-1}
		\begin{split}
			\MoveEqLeft \lim_{N\to\infty} \ergav_{\cR_1(N)}(h)(\omega_\cM,\omega_2) \\
			&= \, \sum_{\substack{G_2\in\cR_2^*,\\\MoireG\in\gls{cRm*}}}  \mathds{1}_{\cR_1^*}(G_2+\MoireG) \,  \hat{h}(\MoireG,G_2) e^{i[\MoireG\cdot \omega_\cM+G_2\cdot\omega_2]} \, .
		\end{split}
	\end{align} 
	Now observe that, due to \eqref{eq-disregj-moire-per}, \eqref{eq-disregj-moire-per},
	\begin{align} \label{eq-G-3-j-decomp}
		\begin{split}
			\bot^T G_2 \, &= \, -\gls{Am}^{-T}A_2^T G_2 \, \in \gls{cRm*} \, ,\\ 
			(I-\bot^T)G_2 \, &= \, A_1^{-T}A_2^T G_2 \, \in \cR_1^* \, .
		\end{split}
	\end{align}
	Thus, and using the fact that $(\cR_1,\gls{cRm})$ is incommensurate, we have that
	\begin{equation} \label{eq-i-ind-G2-MoireG}
		\mathds{1}_{\cR_1^*}(G_2+\MoireG) \, = \, \mathds{1}_{\cR_1^*}(\bot^T G_2+\MoireG) \, = \, \delta_{\MoireG,-\bot^T G_2} \, .
	\end{equation}
	Employing \eqref{eq-i-ind-G2-MoireG}, \eqref{eq-erav-conv-1} implies
	\begin{align} \label{eq-erav-conv-2}
		\begin{split}
			\MoveEqLeft \lim_{N\to\infty} \ergav_{\cR_1(N)}(h)(\omega_\cM,\omega_2) \\
			&= \, \sum_{G_2\in\cR_2^*}  \hat{h}(-\bot^TG_2,G_2) e^{i[-\bot^TG_2\cdot \omega_\cM+G_2\cdot\omega_2]} \, .
		\end{split}
	\end{align}
	Recalling \eqref{def-FT} and \eqref{def-double-FT}, and using Fubini, we have that
	\begin{align}
		\begin{split}
			\MoveEqLeft \sum_{G_2\in\cR_2^*}  \hat{h}(-\bot^TG_2,G_2) e^{i[-\bot^TG_2\cdot \omega_\cM+G_2\cdot\omega_2]} \\
			\, = & \sum_{G_2\in\cR_2^*} \mavint \dx{x} \cF_{\Gamma_2}[h(x,\cdot)](G_2) e^{iG_2\cdot[\bot (x-\omega_\cM)+\omega_2]} \, .
		\end{split}
	\end{align}
	Using Fubini again, we thus find that 
	\begin{align} \label{eq-HN-lim-value}
		\begin{split}
			\MoveEqLeft \sum_{G_2\in\cR_2^*}  \hat{h}(-\bot^TG_2,G_2) e^{i[-\bot^TG_2\cdot \omega_\cM+G_2\cdot\omega_2]} \\
			= \, & \mavint \dx{x} \sum_{G_2\in\cR_2^*} e^{iG_2\cdot[\bot (x-\omega_\cM)+\omega_2]} \cF_{\Gamma_2}[h(x,\cdot)](G_2) \\
			= \, & \mavint \dx{x} h\big(x,\bot (x-\omega_\cM)+\omega_2\big) \, ,
		\end{split}
	\end{align}
	where, in the second-to-last step, we used the Fourier inversion formula. Substituting $x=A_1 \xi+\omega_\cM$, and recalling \eqref{eq-disregj-moire-per}, yields
	\begin{align}\label{eq-HN-sub}
		\begin{split}
			\MoveEqLeft\mavint \dx{x} h\big(x,\bot (x-\omega_\cM)+\omega_2\big) \, = \\
			&  \mavintresc{1} \dx{\xi} h\big(A_1\xi+\omega_\cM,(A_1-A_2) \xi+\omega_2\big)
		\end{split}
	\end{align}
	Collecting  \eqref{eq-erav-conv-2}, \eqref{eq-HN-lim-value}, and \eqref{eq-HN-sub}, we conclude the proof.
\end{proof}

\begin{proof}[Proof of Proposition \ref{prop-ergod-thm}]
	We start by calculating for $\MoireG=\gls{Bm} n\in\gls{cRm*}$
	\begin{equation} \label{eq-A1t-MoireG}
		\begin{split}
			A_1^T \MoireG \, &= \, A_1^T \gls{Bm} n\\
			&= \, 2\pi(I-A_1^TA_2^{-T})n \\
			&= 2\pi (I-A^TR_\theta A^{-T})n \, ,
		\end{split}        
	\end{equation}
	where we recall \eqref{def-Aj}, \eqref{def-Bj}, and \eqref{def-MoireMLV}. Similarly, we find that
	\begin{equation}\label{eq-A2t-MoireG}
		A_2^T\MoireG \, = \, 2\pi (q^{-1}A^TR_{-\theta}A^{-T}-I)n \, .
	\end{equation}
	In particular, using \eqref{eq-A1t-MoireG}, \eqref{eq-A2t-MoireG}, and \eqref{def-appdelNi}, a straight-forward calculation yields
	\begin{align}
		\begin{split} \label{eq-appdelNi@MoireML}
			\MoveEqLeft \gls{Dirichlet}(\gls{Bm} n)\, = \\
			& \prod_{\ell=1}^2 \frac{\sin\Big(\pi\dist\big((q^{(-1)^{j+1}}A^TR_{(-1)^{j+1}\theta} A^{-T} n)_\ell,\Z\big)\big(2N+1\big)\Big)}{(2N+1)\sin\big(\pi\dist\big((q^{(-1)^{j+1}}A^TR_{(-1)^{j+1}\theta} A^{-T} n)_\ell,\Z\big)\big)} \, .
		\end{split}
	\end{align}    
	Observe that
	\begin{equation}
		\frac{\dist(x,\Z)}{|\sin(\pi x)|} \, = \, \frac{\dist(x,\Z)}{|\sin\big(\pi \dist(x,\Z)\big)|} \, \leq\, \frac1{\pi}\sup_{0<y<\pi/2}\frac{y}{\sin(y)} \, \leq \, \frac12 \, .
	\end{equation}
	Let $n\in \Z^2\setminus\{0\}$. Recalling \eqref{eq-appdelNi@MoireML}, we thus obtain that
	\begin{equation} \label{eq-appdelNi-bd-0}
		\begin{split}
			\lefteqn{\max_{\ell=1,2}\dist\big((q^{(-1)^{j+1}}A^TR_{(-1)^{j+1}\theta} A^{-T} n)_\ell,\Z\big)\gls{Dirichlet}(\gls{Bm} n)}\\
			& \leq\, \frac1{\pi(2N+1)}\sup_{0<y<\pi/2}\frac{y}{\sin(y)} \\
			& = \, \frac1{2(2N+1)} \, .
		\end{split}
	\end{equation}
	\par Using \eqref{eq-appdelNi-bd-0} and employing the fact that $\sqrt{a^2+b^2}\leq \sqrt{2}\max\{|a|,|b|\}$, we find that 
	\begin{align}\label{eq-dirichlet-expand}
		\begin{split}
			\lefteqn{|\gls{Dirichlet}(\gls{Bm} n)|}\\
			=& \, \frac{\max_{\ell=1,2}\dist\big((q^{(-1)^{j+1}}A^TR_{(-1)^{j+1}\theta} A^{-T} n)_\ell,\Z\big)|\gls{Dirichlet}(\gls{Bm} n)|}{\dist\big(q^{(-1)^{j+1}}A^TR_{(-1)^{j+1}\theta} A^{-T} n,\Z^2\big)}\\
			& \, \times \frac{\dist\big(q^{(-1)^{j+1}}A^TR_{(-1)^{j+1}\theta} A^{-T} n,\Z^2\big)}{\max_{\ell=1,2}\dist\big((q^{(-1)^{j+1}}A^TR_{(-1)^{j+1}\theta} A^{-T} n)_\ell,\Z\big)} \\
			\leq& \, \frac{|n|^{2\sigma}}{\sqrt{2}K(2N+1)} \, .
		\end{split}
	\end{align}
    Here, we recognize the Diophantine condition \eqref{def-diophantine} in the second line of \eqref{eq-dirichlet-expand}.
	\par Recall from \eqref{eq-MoireRL-QM-lb} that
	\begin{equation} \label{eq-n-MoireG-bd}
		|n| \, \leq\, \frac{\gls{moirelen} \|A\|_2}{2\pi} |\gls{Bm} n| \, .
	\end{equation}
	In particular, we arrive at
	\begin{equation} \label{eq-appdelN-bd}
		|\gls{Dirichlet}(\MoireG)| \, \leq\, \frac{1}{\sqrt{2}K}\Big(\frac{\gls{moirelen} \|A\|_2}{2\pi}\Big)^{2\sigma}|\MoireG|^{2\sigma}\frac1{2N+1}
	\end{equation}
	for any $\MoireG\in\gls{cRm*}\setminus\{0\}$. Using  \eqref{eq-appdelN-bd} and the steps leading to \eqref{eq-sum-trafo-approx}, we thus conclude the proof.
\end{proof}

\vspace*{2ex}

\noindent

\bibliographystyle{plain}  
\bibliography{references}  

\begin{thebibliography}{10}

\bibitem{arendt2023semilinear}
Wolfgang Arendt and Daniel Daners.
\newblock Semilinear elliptic equations on rough domains.
\newblock {\em J. Differential Equations}, 346:376--415, 2023.

\bibitem{1978Aubry}
Serge Aubry.
\newblock The new concept of transitions by breaking of analyticity in a
  crystallographic model.
\newblock In Alan~R. Bishop and Toni Schneider, editors, {\em Solitons and
  Condensed Matter Physics}, pages 264--277, Berlin, Heidelberg, 1978. Springer
  Berlin Heidelberg.

\bibitem{1980Aubry}
Serge Aubry and Gilles Andr{\'e}.
\newblock Analyticity breaking and {Anderson} localization in incommensurate
  lattices.
\newblock {\em Annals of the Israel Physical Society}, 3, 1980.

\bibitem{aubry1983frenkel-kontorova}
Serge Aubry and Pierre-Yves Le~Daeron.
\newblock The discrete {F}renkel-{K}ontorova model and its extensions: I.
  {E}xact results for the ground-states.
\newblock {\em Physica D: Nonlinear Phenomena}, 8(3):381--422, 1983.

\bibitem{beckerzworski2022magicangle}
Simon Becker, Mark Embree, Jens Wittsten, and Maciej Zworski.
\newblock Mathematics of magic angles in a model of twisted bilayer graphene.
\newblock {\em Probability and Mathematical Physics}, 3(1):69--103, 2022.

\bibitem{becker-humbert-zworski2022finestructure-magicangle}
Simon Becker, Tristan Humbert, and Maciej Zworski.
\newblock Fine structure of flat bands in a chiral model of magic angles.
\newblock {\em arXiv preprint arXiv:2208.01628}, 2022.

\bibitem{bellissard-quasi-per1982}
J.~Bellissard and D.~Testard.
\newblock Quasi periodic {Hamiltonians}: a mathematical approach.
\newblock {\em Operator Algebras and Applications}, 2:579, 1982.

\bibitem{bellissard2002coherent}
Jean Bellissard.
\newblock Coherent and dissipative transport in aperiodic solids: An overview.
\newblock {\em Dynamics of Dissipation}, pages 413--485, 2002.

\bibitem{bellissard1994noncommutative}
Jean Bellissard, Andreas van Elst, and Hermann Schulz-Baldes.
\newblock The noncommutative geometry of the quantum {Hall} effect.
\newblock {\em Journal of Mathematical Physics}, 35(10):5373--5451, 1994.

\bibitem{beresnevich2016metric}
Victor Beresnevich, Felipe Ram{\'\i}rez, and Sanju Velani.
\newblock Metric {D}iophantine approximation: aspects of recent work.
\newblock {\em Dynamics and analytic number theory}, 437:1--95, 2016.

\bibitem{berndt2018circle}
Bruce~C Berndt, Sun Kim, and Alexandru Zaharescu.
\newblock The {C}ircle problem of {G}auss and the divisor problem of
  {D}irichlet—still unsolved.
\newblock {\em The American Mathematical Monthly}, 125(2):99--114, 2018.

\bibitem{bilyk2014discrepancy}
Dmitriy Bilyk.
\newblock Discrepancy theory and harmonic analysis.
\newblock {\em Uniform Distribution and Quasi-Monte Carlo Methods}, pages
  45--62, 2014.

\bibitem{bilyk2011directional}
Dmitriy Bilyk, Xiaomin Ma, Jill Pipher, and Craig Spencer.
\newblock Directional discrepancy in two dimensions.
\newblock {\em Bulletin of the London Mathematical Society}, 43(6):1151--1166,
  2011.

\bibitem{bilyk2016diophantine}
Dmitriy Bilyk, Xiaomin Ma, Jill Pipher, and Craig Spencer.
\newblock Diophantine approximations and directional discrepancy of rotated
  lattices.
\newblock {\em Transactions of the American Mathematical Society},
  368(6):3871--3897, 2016.

\bibitem{bistritzer2011moire}
Rafi Bistritzer and Allan~H MacDonald.
\newblock Moir{\'e} bands in twisted double-layer graphene.
\newblock {\em Proceedings of the National Academy of Sciences},
  108(30):12233--12237, 2011.

\bibitem{blanc2002molecular}
Xavier Blanc, Claude Le~Bris, and P-L Lions.
\newblock From molecular models to continuum mechanics.
\newblock {\em Archive for Rational Mechanics and Analysis}, 164:341--381,
  2002.

\bibitem{bourgain2016oppenheim}
Jean Bourgain.
\newblock A quantitative {O}ppenheim theorem for generic diagonal quadratic
  forms.
\newblock {\em Israel J. Math.}, 215(1):503--512, 2016.

\bibitem{bourgain2017mean}
Jean Bourgain and Nigel Watt.
\newblock Mean square of zeta function, circle problem and divisor problem
  revisited.
\newblock {\em arXiv preprint arXiv:1709.04340}, 2017.

\bibitem{goetzemargulis2022distribution}
P.~Buterus, F.~G\"{o}tze, T.~Hille, and G.~Margulis.
\newblock Distribution of values of quadratic forms at integral points.
\newblock {\em Invent. Math.}, 227(3):857--961, 2022.

\bibitem{cabrecinti2021allencahn-stable}
Xavier Cabre, Eleonora Cinti, and Joaquim Serra.
\newblock Stable solutions to the fractional {A}llen-{C}ahn equation in the
  nonlocal perimeter regime.
\newblock {\em arXiv preprint arXiv:2111.06285}, 2021.

\bibitem{generalizedkubo2017}
Eric Canc{\`e}s, Paul Cazeaux, and Mitchell Luskin.
\newblock Generalized {K}ubo formulas for the transport properties of
  incommensurate 2d atomic heterostructures.
\newblock {\em Journal of Mathematical Physics}, 58(6):063502, 2017.

\bibitem{Cao2018a}
Yuan Cao, Valla Fatemi, Ahmet Demir, Shiang Fang, Spencer~L. Tomarken, Jason~Y.
  Luo, Javier~D. Sanchez-Yamagishi, Kenji Watanabe, Takashi Taniguchi,
  Efthimios Kaxiras, Ray~C. Ashoori, and Pablo Jarillo-Herrero.
\newblock {Correlated insulator behaviour at half-filling in magic-angle
  graphene superlattices}.
\newblock {\em Nature}, 556(7699):80--84, 2018.

\bibitem{Cao2018}
Yuan Cao, Valla Fatemi, Shiang Fang, Kenji Watanabe, Takashi Taniguchi,
  Efthimios Kaxiras, and Pablo Jarillo-Herrero.
\newblock {Unconventional superconductivity in magic-angle graphene
  superlattices}.
\newblock {\em Nature}, 556(7699):43--50, 2018.

\bibitem{Carr2018}
Stephen Carr, Daniel Massatt, Steven~B. Torrisi, Paul Cazeaux, Mitchell Luskin,
  and Efthimios Kaxiras.
\newblock Relaxation and domain formation in incommensurate two-dimensional
  heterostructures.
\newblock {\em Physical Review B}, 98, 12 2018.

\bibitem{relaxation-domain-wall-formation2018}
Stephen Carr, Daniel Massatt, Steven~B. Torrisi, Paul Cazeaux, Mitchell Luskin,
  and Efthimios Kaxiras.
\newblock Relaxation and domain formation in incommensurate two-dimensional
  heterostructures.
\newblock {\em Phys. Rev. B}, 98(22):224102, December 2018.

\bibitem{catarina-amorim-castro-lopes-peres2019tBGintro}
Gon{\c{c}}alo Catarina, Bruno Amorim, Eduardo~V Castro, Jo{\~a}o~MVP Lopes, and
  Nuno~MR Peres.
\newblock Twisted bilayer graphene: low-energy physics, electronic and optical
  properties.
\newblock {\em Handbook of Graphene}, 3:177--232, 2019.

\bibitem{cazeaux-clark-engelke-kim2022relaxation}
Paul Cazeaux, Drake Clark, Rebecca Engelke, Philip Kim, and Mitchell Luskin.
\newblock Relaxation and domain wall structure of bilayer moir\'e systems.
\newblock {\em Journal of Elasticity}, 2023.

\bibitem{Cazeaux-Massatt-Luskin-ARMA2020}
Paul Cazeaux, Mitchell Luskin, and Daniel Massatt.
\newblock {E}nergy minimization of two dimensional incommensurate
  heterostructures.
\newblock {\em Archive for Rational Mechanics and Analysis}, 235:1289--1325, 2
  2020.

\bibitem{cazeauxrippling2017}
Paul Cazeaux, Mitchell Luskin, and Ellad~B. Tadmor.
\newblock Analysis of {Rippling} in {Incommensurate} {One}-{Dimensional}
  {Coupled} {Chains}.
\newblock {\em Multiscale Modeling \& Simulation}, 15(1):56--73, January 2017.

\bibitem{cazeaux-luskin-CB-2017}
{Cazeaux, Paul} and {Luskin, Mitchell}.
\newblock {C}auchy-{B}orn strain energy density for coupled incommensurate
  elastic chains.
\newblock {\em ESAIM: M2AN}, 52(2):729--749, 2018.

\bibitem{cazenaveSLEE2006book}
Thierry Cazenave.
\newblock An introduction to semilinear elliptic equations.
\newblock {\em Editora do IM-UFRJ, Rio de Janeiro}, 2006.

\bibitem{champagne2023dio-constr}
J\'{e}r\'{e}my Champagne and Damien Roy.
\newblock Diophantine approximation with constraints.
\newblock {\em Acta Arith.}, 207(1):57--99, 2023.

\bibitem{chenhott2023}
T.~Chen and M.~Hott.
\newblock On the emergence of quantum {B}oltzmann fluctuation dynamics near a
  {B}ose--{E}instein condensate.
\newblock {\em J. Stat. Phys.}, 190(4):85, 2023.

\bibitem{chodosh2023lecture}
Otis Chodosh.
\newblock Lecture notes on geometric features of the {A}llen--{C}ahn equation
  ({P}rinceton, 2019).
\newblock https://web.stanford.edu/~ochodosh/AllenCahnSummerSchool2019.pdf,
  2023.
\newblock Accessed 06/15/2024.

\bibitem{cinti-davila-delpino2016allencahn}
Eleonora Cinti, Juan Davila, and Manuel Del~Pino.
\newblock Solutions of the fractional {A}llen-{C}ahn equation which are
  invariant under screw motion.
\newblock {\em Journal of the London Mathematical Society}, 94(1):295--313,
  2016.

\bibitem{dai2016twisted}
Shuyang Dai, Yang Xiang, and David~J Srolovitz.
\newblock Twisted bilayer graphene: Moir{\'e} with a twist.
\newblock {\em Nano letters}, 16(9):5923--5927, 2016.

\bibitem{Damanik2022Oerg-schroedinger-book}
David Damanik.
\newblock {\em One-dimensional ergodic {S}chr\"odinger operators}.
\newblock Graduate studies in mathematics, volume 221. American Mathematical
  Society, Providence, Rhode Island, 2022.

\bibitem{KimHof13}
C~Dean, Lei Wang, P~Maher, Carlos Forsythe, Fereshte Ghahari, Y~Gao, Jyoti
  Katoch, Masa Ishigami, Pilkyung Moon, Mikito Koshino, Takashi Taniguchi,
  K~Watanabe, K~Shepard, James Hone, and Phaly Kim.
\newblock Hofstadter's butterfly and the fractal quantum {Hall} effect in
  moir\'e superlattices.
\newblock {\em Nature}, 497:598, 05 2013.

\bibitem{diaz2023semilinear}
Gregorio D\'{\i}az.
\newblock Large solutions of elliptic semilinear equations non-degenerate near
  the boundary.
\newblock {\em Commun. Pure Appl. Anal.}, 22(3):686--735, 2023.

\bibitem{duoandikoetxea2001fourier}
Javier Duoandikoetxea.
\newblock {\em Fourier analysis}.
\newblock Graduate studies in mathematics v. 29. American Mathematical Society,
  Providence, Rhode Island, 2001.

\bibitem{dupaigneStableSolBook}
Louis Dupaigne.
\newblock {\em Stable solutions of elliptic partial differential equations},
  volume 143 of {\em Chapman \& Hall/CRC Monographs and Surveys in Pure and
  Applied Mathematics}.
\newblock Chapman \& Hall/CRC, Boca Raton, FL, 2011.

\bibitem{dupaigne2021regularity}
Louis Dupaigne and Alberto Farina.
\newblock Regularity and symmetry for semilinear elliptic equations in bounded
  domains.
\newblock {\em Commun. Contemp. Math.}, 25(5):Paper No. 2250018, 27, 2023.

\bibitem{eskin2005quadratic}
Alex Eskin, Gregory Margulis, and Shahar Mozes.
\newblock Quadratic forms of signature (2, 2) and eigenvalue spacings on
  rectangular 2-tori.
\newblock {\em Annals of mathematics}, pages 679--725, 2005.

\bibitem{malena2018}
Malena~I. Espa{\~{n}}ol, Dmitry Golovaty, and J.~Patrick Wilber.
\newblock Discrete-to-continuum modelling of weakly interacting incommensurate
  two-dimensional lattices.
\newblock {\em Proc. R. Soc. A.}, page 20170612, 2018.

\bibitem{malena2023}
Malena~I. Espa{\~{n}}ol, Dmitry Golovaty, and J.~Patrick Wilber.
\newblock A discrete-to-continuum model of weakly interacting incommensurate
  two-dimensional lattices: The hexagonal case.
\newblock {\em Journal of the Mechanics and Physics of Solids}, 173:105229,
  2023.

\bibitem{figalli2020stable}
Alessio Figalli and Joaquim Serra.
\newblock On stable solutions for boundary reactions: a {D}e {G}iorgi-type
  result in dimension 4+ 1.
\newblock {\em Inventiones mathematicae}, 219(1):153--177, 2020.

\bibitem{Frickergitterpunktlehre}
Fran\c{c}ois Fricker.
\newblock {\em {E}inf\"{u}hrung in die {G}itterpunktlehre}, volume~73 of {\em
  Lehrb\"{u}cher und Monographien aus dem Gebiete der Exakten Wissenschaften
  (LMW). Mathematische Reihe [Textbooks and Monographs in the Exact Sciences.
  Mathematical Series]}.
\newblock Birkh\"{a}user Verlag, Basel-Boston, Mass., 1982.

\bibitem{ghosh2018oppenheim}
Anish Ghosh and Dubi Kelmer.
\newblock A quantitative {O}ppenheim theorem for generic ternary quadratic
  forms.
\newblock {\em J. Mod. Dyn.}, 12:1--8, 2018.

\bibitem{guopei2022ginzburglandau}
Boling Guo and Yitong Pei.
\newblock Periodic solutions of {G}inzburg-{L}andau theory for atomic {F}ermi
  gases near the {BCS}-{BEC} crossover.
\newblock {\em Appl. Anal.}, 101(4):1199--1210, 2022.

\bibitem{hammonds2023kdiophantinemtuples}
Trajan Hammonds, Seoyoung Kim, Steven~J. Miller, Arjun Nigam, Kyle Onghai,
  Dishant Saikia, and Lalit~M. Sharma.
\newblock {$k$}-{D}iophantine {$m$}-tuples in finite fields.
\newblock {\em Int. J. Number Theory}, 19(4):891--912, 2023.

\bibitem{huxley2003exponential}
Martin~N Huxley.
\newblock Exponential sums and lattice points iii.
\newblock {\em Proceedings of the London Mathematical Society}, 87(3):591--609,
  2003.

\bibitem{ignatjerrard2021GinzburgLandau}
Radu Ignat and Robert~L Jerrard.
\newblock Renormalized energy between vortices in some {G}inzburg-{L}andau
  models on 2-dimensional {R}iemannian manifolds.
\newblock {\em Archive for Rational Mechanics and Analysis}, 239(3):1577--1666,
  2021.

\bibitem{joseph2023semilinear}
Anumol Joseph and Lakshmi Sankar.
\newblock Singular semilinear elliptic problems on unbounded domains in
  {$\mathbb{R}^n$}.
\newblock {\em J. Math. Anal. Appl.}, 520(2):Paper No. 126903, 12, 2023.

\bibitem{kechris2010global}
Alexander~S. Kechris.
\newblock {\em Global aspects of ergodic group actions}, volume 160 of {\em
  Mathematical Surveys and Monographs}.
\newblock American Mathematical Society, Providence, RI, 2010.

\bibitem{rubio2021quantum-simulator}
Dante~M Kennes, Martin Claassen, Lede Xian, Antoine Georges, Andrew~J Millis,
  James Hone, Cory~R Dean, DN~Basov, Abhay~N Pasupathy, and Angel Rubio.
\newblock Moir{\'e} heterostructures as a condensed-matter quantum simulator.
\newblock {\em Nature Physics}, 17(2):155--163, 2021.

\bibitem{KerrLi-erg-th}
David Kerr and Hanfeng Li.
\newblock {\em Ergodic theory : independence and dichotomies}.
\newblock Springer monographs in mathematics. Springer, Cham, Switzerland,
  2016.

\bibitem{KolmogorovCrespi}
Aleksey~N. Kolmogorov and Vincent~H. Crespi.
\newblock Registry-dependent interlayer potential for graphitic systems.
\newblock {\em Phys. Rev. B}, 71:235415, Jun 2005.

\bibitem{Koshino2019}
Mikito Koshino and Young~Woo Son.
\newblock Moir\'e phonons in twisted bilayer graphene.
\newblock {\em Physical Review B}, 100, 8 2019.

\bibitem{KuipersNiederreiter1974}
{Kuipers, Lauwerens} and {Niederreiter, Harald}.
\newblock {\em Uniform distribution of sequences}.
\newblock Pure and applied mathematics. Wiley-Interscience, New York, 1974.

\bibitem{KURATSUBO2022}
Shigehiko Kuratsubo and Eiichi Nakai.
\newblock Multiple {F}ourier series and lattice point problems.
\newblock {\em Journal of Functional Analysis}, 282(1):109272, 2022.

\bibitem{laxphillips}
Peter~D Lax and Ralph~S Phillips.
\newblock The asymptotic distribution of lattice points in euclidean and
  non-euclidean spaces.
\newblock {\em Journal of Functional Analysis}, 46(3):280--350, 1982.

\bibitem{leven2016interlayer}
Itai Leven, Tal Maaravi, Ido Azuri, Leeor Kronik, and Oded Hod.
\newblock Interlayer potential for graphene/h-{BN} heterostructures.
\newblock {\em Journal of chemical theory and computation}, 12(6):2896--2905,
  2016.

\bibitem{levitan1987}
Boris~M Levitan.
\newblock Asymptotic formulae for the number of lattice points in {E}uclidean
  and {L}obachevskii spaces.
\newblock {\em Russian Mathematical Surveys}, 42(3):13, 1987.

\bibitem{Lu2022}
Jonathan~Z. Lu, Ziyan Zhu, Mattia Angeli, Daniel~T. Larson, and Efthimios
  Kaxiras.
\newblock Low-energy moir\'e phonons in twisted bilayer van der waals
  heterostructures.
\newblock {\em Phys. Rev. B}, 106:144305, Oct 2022.

\bibitem{restrepo-maggi2022allencahn}
Francesco Maggi and Daniel Restrepo.
\newblock Uniform stability in the euclidean isoperimetric problem for the
  {A}llen-{C}ahn energy.
\newblock {\em arXiv preprint arXiv:2202.11583}, 2022.

\bibitem{Hod10}
Noa Marom, Jonathan Bernstein, Jonathan Garel, Alexandre Tkatchenko, Ernesto
  Joselevich, Leeor Kronik, and Oded Hod.
\newblock Stacking and registry effects in layered materials: The case of
  hexagonal boron nitride.
\newblock {\em Phys. Rev. Lett.}, 105:046801, Jul 2010.

\bibitem{massatt2021confiso}
Daniel Massatt, Stephen Carr, and Mitchell Luskin.
\newblock Electronic observables for relaxed bilayer {2D} heterostructures in
  momentum space.
\newblock {\em Multiscale Model. Simul.}, 21(4):1344--1378, 2023.

\bibitem{dos17}
Daniel Massatt, Mitchell Luskin, and Christoph Ortner.
\newblock Electronic density of states for incommensurate layers.
\newblock {\em SIAM J. Multiscale Modeling \& Simulation}, 15:476--499, 2017.

\bibitem{massatt-luskin-ortner2017electronicdensity}
Daniel Massatt, Mitchell Luskin, and Christoph Ortner.
\newblock {E}lectronic density of states for incommensurate layers.
\newblock {\em Multiscale Modeling \& Simulation}, 15(1):476--499, 2017.

\bibitem{ortner21}
Felix Musil, Andrea Grisafi, Albert~P. Bartók, Christoph Ortner, Gábor
  Csányi, and Michele Ceriotti.
\newblock Physics-inspired structural representations for molecules and
  materials.
\newblock {\em Chemical Reviews}, 121(16):9759--9815, 2021.

\bibitem{Mutalik2019quanta-article}
Pradeep Mutalik.
\newblock Solution: {M}agic {M}oir\'e in {T}wisted {G}raphene.
\newblock {\em Quanta Magazine}, July 2019.
\newblock
  \url{https://www.quantamagazine.org/puzzle-solution-magic-moire-in-graphene-20190726/},
  Accessed: 2023-03-22.

\bibitem{Nam2017}
Nguyen~N.T. Nam and Mikito Koshino.
\newblock Lattice relaxation and energy band modulation in twisted bilayer
  graphene.
\newblock {\em Physical Review B}, 96, 8 2017.

\bibitem{odorney2023diophantine}
Evan O'Dorney.
\newblock Diophantine approximation on conics.
\newblock {\em Proc. Amer. Math. Soc.}, 151(5):1889----1905, 2023.

\bibitem{ortnertheil2013}
C.~Ortner and F.~Theil.
\newblock Justification of the {Cauchy}-{Born} {Approximation} of
  {Elastodynamics}.
\newblock {\em Archive for Rational Mechanics and Analysis}, 207(3):1025--1073,
  March 2013.
\newblock Publisher: Springer Science and Business Media, LLC.

\bibitem{oconnor2015airebo}
Thomas~C O’Connor, Jan Andzelm, and Mark~O Robbins.
\newblock Airebo-m: A reactive model for hydrocarbons at extreme pressures.
\newblock {\em The Journal of chemical physics}, 142(2), 2015.

\bibitem{peng2020strain}
Zhiwei Peng, Xiaolin Chen, Yulong Fan, David~J Srolovitz, and Dangyuan Lei.
\newblock Strain engineering of 2d semiconductors and graphene: from strain
  fields to band-structure tuning and photonic applications.
\newblock {\em Light: Science \& Applications}, 9(1):190, 2020.

\bibitem{SanJose2014sg}
Pablo San-Jose, A.~Guti\'errez-Rubio, Mauricio Sturla, and Francisco Guinea.
\newblock Spontaneous strains and gap in graphene on boron nitride.
\newblock {\em Phys. Rev. B}, 90:075428, Aug 2014.

\bibitem{savin2010GinzburgLandau-minimal}
Ovidiu Savin.
\newblock Minimal surfaces and minimizers of the {G}inzburg-{L}andau energy.
\newblock {\em Cont. Math. Mech. Analysis AMS}, 526:43--58, 2010.

\bibitem{schmidt2015interatomic}
Kevin~M Schmidt, Alex~B Buettner, Olivia~A Graeve, and Victor~R Vasquez.
\newblock Interatomic pair potentials from {DFT and molecular dynamics for Ca,
  Ba, and Sr hexaborides}.
\newblock {\em Journal of Materials Chemistry C}, 3(33):8649--8658, 2015.

\bibitem{serfaty2017meanfield-GinzburgLandau-GP}
Sylvia Serfaty.
\newblock Mean field limits of the {G}ross-{P}itaevskii and parabolic
  {G}inzburg-{L}andau equations.
\newblock {\em Journal of the American Mathematical Society}, 30(3):713--768,
  2017.

\bibitem{shapeev16}
Alexander~V. Shapeev.
\newblock Moment tensor potentials: A class of systematically improvable
  interatomic potentials.
\newblock {\em Multiscale Modeling \& Simulation}, 14(3):1153--1173, 2016.

\bibitem{tachim2022allencahn-large}
Theodore Tachim~Medjo.
\newblock Large deviation principles for a 2{D} stochastic
  {A}llen-{C}ahn-{N}avier-{S}tokes driven by jump noise.
\newblock {\em Stoch. Dyn.}, 22(4):Paper No. 2250005, 36, 2022.

\bibitem{TKV2019magicanglecond}
Grigory Tarnopolsky, Alex~Jura Kruchkov, and Ashvin Vishwanath.
\newblock Origin of magic angles in twisted bilayer graphene.
\newblock {\em Physical review letters}, 122(10):106405, 2019.

\bibitem{Trefethen2014}
Lloyd~N. Trefethen and J.~A.C. Weideman.
\newblock The exponentially convergent trapezoidal rule.
\newblock {\em SIAM Review}, 56:385--458, 2014.

\bibitem{2DPerturb15}
Georgios~A. Tritsaris, Sharmila~N. Shirodkar, Efthimios Kaxiras, Paul Cazeaux,
  Mitchell Luskin, Petr Plech\'a\v{c}, and Eric Canc\`es.
\newblock Perturbation theory for weakly coupled two-dimensional layers.
\newblock {\em Journal of Materials Research}, 31:959--966, 4 2016.

\bibitem{watsonluskin2021}
Alexander~B. Watson and Mitchell Luskin.
\newblock Existence of the first magic angle for the chiral model of bilayer
  graphene.
\newblock {\em Journal of Mathematical Physics}, 62(9):091502, 2021.

\bibitem{KimRelax18}
Hyobin Yoo, Rebecca Engelke, Stephen Carr, Shiang Fang, Kuan Zhang, Paul
  Cazeaux, Suk~Hyun Sung, Robert Hovden, Adam~W. Tsen, Takashi Taniguchi, Kenji
  Watanabe, Gyu-Chul Yi, Miyoung Kim, Mitchell Luskin, Ellad~B. Tadmor,
  Efthimios Kaxiras, and Philip Kim.
\newblock Atomic and electronic reconstruction at van der {Waals} interface in
  twisted bilayer graphene.
\newblock {\em Nature Materials}, pages 448--453, 2019.

\bibitem{Tadmor2017}
Kuan Zhang and Ellad B.~Tadmor.
\newblock Structural and electron diffraction scaling of twisted graphene
  bilayers.
\newblock {\em Journal of the Mechanics and Physics of Solids}, 112, 12 2017.

\bibitem{zhang-tsai-zhu2021correlatedinsulating}
Xi~Zhang, Kan-Ting Tsai, Ziyan Zhu, Wei Ren, Yujie Luo, Stephen Carr, Mitchell
  Luskin, Efthimios Kaxiras, and Ke~Wang.
\newblock Correlated insulating states and transport signature of
  superconductivity in twisted trilayer graphene superlattices.
\newblock {\em Physical review letters}, 127(16):166802, 2021.

\bibitem{Srolovitz2015}
Songsong Zhou, Jian Han, Shuyang Dai, Jianwei Sun, and David~J. Srolovitz.
\newblock van der {W}aals bilayer energetics: Generalized stacking-fault energy
  of graphene, boron nitride, and graphene/boron nitride bilayers.
\newblock {\em Phys. Rev. B}, 92:155438, Oct 2015.

\bibitem{zhu-carr-massatt2020twisted-tunable}
Ziyan Zhu, Stephen Carr, Daniel Massatt, Mitchell Luskin, and Efthimios
  Kaxiras.
\newblock Twisted trilayer graphene: A precisely tunable platform for
  correlated electrons.
\newblock {\em Physical review letters}, 125(11):116404, 2020.

\bibitem{zhu_relaxation_2020}
Ziyan Zhu, Paul Cazeaux, Mitchell Luskin, and Efthimios Kaxiras.
\newblock Modeling mechanical relaxation in incommensurate trilayer van der
  {Waals} heterostructures.
\newblock {\em Phys. Rev. B}, 101(22):224107, June 2020.

\end{thebibliography}

\noindent

\printnoidxglossary[type=symbols,style=long,title={List of Symbols}]

\end{document}